\documentclass[12pt]{article}
\usepackage{color,upref,amsmath,amssymb,amscd}
\usepackage{amsthm}
\usepackage[dvipdfmx]{hyperref}
\usepackage{graphicx}
\usepackage{color}

\def\hybrid{\topmargin 0pt      \oddsidemargin 0pt
        \headheight 0pt \headsep 0pt
       \voffset-1cm
        \textwidth 6.25in       
       \textheight 9.5in       
        \marginparwidth 0.0in
        \parskip 5pt plus 1pt   \jot = 1.5ex}
\catcode`\@=11
\def\marginnote#1{}

\newcount\hour
\newcount\minute
\newtoks\amorpm
\hour=\time\divide\hour by60
\minute=\time{\multiply\hour by60 \global\advance\minute by-\hour}
\edef\standardtime{{\ifnum\hour<12 \global\amorpm={am}%
        \else\global\amorpm={pm}\advance\hour by-12 \fi
        \ifnum\hour=0 \hour=12 \fi
        \number\hour:\ifnum\minute<10 0\fi\number\minute\the\amorpm}}
\edef\militarytime{\number\hour:\ifnum\minute<10 0\fi\number\minute}

\def\draftlabel#1{{\@bsphack\if@filesw {\let\thepage\relax
   \xdef\@gtempa{\write\@auxout{\string
      \newlabel{#1}{{\@currentlabel}{\thepage}}}}}\@gtempa
   \if@nobreak \ifvmode\nobreak\fi\fi\fi\@esphack}
        \gdef\@eqnlabel{#1}}
\def\@eqnlabel{}
\def\@vacuum{}
\def\draftmarginnote#1{\marginpar{\raggedright\scriptsize\tt#1}}

\def\draftlabel#1{{\@bsphack\if@filesw {\let\thepage\relax
   \xdef\@gtempa{\write\@auxout{\string
      \newlabel{#1}{{\@currentlabel}{\thepage}}}}}\@gtempa
   \if@nobreak \ifvmode\nobreak\fi\fi\fi\@esphack}
        \gdef\@eqnlabel{#1}}
\def\@eqnlabel{}
\def\@vacuum{}
\def\draftmarginnote#1{\marginpar{\raggedright\scriptsize\tt#1}}

\def\draft{\oddsidemargin -.5truein
        \def\@oddfoot{\sl preliminary draft \hfil
        \rm\thepage\hfil\sl\today\quad\militarytime}
        \let\@evenfoot\@oddfoot \overfullrule 3pt
        \let\label=\draftlabel
        \let\marginnote=\draftmarginnote
   \def\@eqnnum{(\theequation)\rlap{\kern\marginparsep\tt\@eqnlabel}%
\global\let\@eqnlabel\@vacuum}  }

\def\numberbysection{\@addtoreset{equation}{section}
        \def\theequation{\thesection.\arabic{equation}}}

\def\underline#1{\relax\ifmmode\@@underline#1\else
        $\@@underline{\hbox{#1}}$\relax\fi}

\def\titlepage{\@restonecolfalse\if@twocolumn\@restonecoltrue\onecolumn
     \else \newpage \fi \thispagestyle{empty}\c@page\z@
        \def\thefootnote{\fnsymbol{footnote}} }

\def\endtitlepage{\if@restonecol\twocolumn \else  \fi
        \def\thefootnote{\arabic{footnote}}
        \setcounter{footnote}{0}}  
\relax

\numberbysection
\hybrid

\newfont{\Bbbb}{msbm7 scaled 1\@ptsize00}
\newcommand{\CC}{\mathbb C}

\newcommand{\DDD}{\raise-1pt\hbox{$\mbox{\Bbbb D}$}}

\newcommand{\UUU}{\raise-1pt\hbox{$\mbox{\Bbbb U}$}}
\newcommand{\ZZ}{\mathbb Z}
\newcommand{\z}{\raise-1pt\hbox{$\mbox{\Bbbb Z}$}}

\newtheorem{predl}{Proposition}[section]
\newtheorem{lem}{Lemma}[section]
\newtheorem{theorem}{Theorem}[section]

\def\beq{\begin{equation}}
\def\eeq{\end{equation}}
\def\p{\partial}
\newcommand{\der}{\partial}

\def\a{\alpha}
\def\b{\beta}

\def\normord{ {\scriptstyle {{\bullet}\atop{\bullet}}} }
\def\lbr{\left <}
\def\rbr{\right >}

\newcommand{\bolda}{{\boldsymbol a}}
\newcommand{\boldA}{{\boldsymbol A}}
\newcommand{\boldb}{{\boldsymbol b}}
\newcommand{\boldB}{{\boldsymbol B}}
\newcommand{\boldBbar}{\bar{\boldsymbol B}}

\newcommand{\boldL}{{\boldsymbol L}}
\newcommand{\boldLbar}{\bar{\boldsymbol L}}

\newcommand{\boldone}{{\boldsymbol 1}}
\newcommand{\boldP}{{\boldsymbol P}}
\newcommand{\boldPbar}{\bar{\boldsymbol P}}
\newcommand{\boldPsi}{{\boldsymbol\Psi}}
\newcommand{\boldPsibar}{\bar{\boldsymbol\Psi}}
\newcommand{\boldQ}{{\boldsymbol Q}}
\newcommand{\boldQbar}{\bar{\boldsymbol Q}}
\newcommand{\bolds}{{\boldsymbol s}}

\newcommand{\boldt}{{\boldsymbol t}}

\newcommand{\boldtbar}{\bar{\boldsymbol t}}
\newcommand{\boldU}{{\boldsymbol U}}
\newcommand{\boldUbar}{\bar{\boldsymbol U}}
\newcommand{\boldW}{{\boldsymbol W}}
\newcommand{\boldWbar}{\bar{\boldsymbol W}}
\newcommand{\boldzero}{{\boldsymbol 0}}
\newcommand{\calA}{{\mathcal A}}
\newcommand{\calF}{{\mathcal F}}
\newcommand{\calW}{{\mathcal W}}
\newcommand{\Comp}{{\mathbb C}}
\newcommand{\diag}{\operatorname{diag}\nolimits}
\newcommand{\Integer}{{\mathbb Z}}

\DeclareMathOperator{\Mat}{Mat}

\newcommand{\secref}[1]{\S\ref{#1}}
\newcommand{\temp}{{\text{temp}}}
\newcommand\tensor{\otimes}

\newcommand\lemref[1]{Lemma~\ref{#1}}
\newcommand\propref[1]{Proposition~\ref{#1}}

\newcommand\commentout[1]{}%

\def\normordbare{ {\scriptstyle {{ \times}\atop{ \times}}} }

\theoremstyle{remark}
\newtheorem{zam}{Remark}[section]

\newcommand\remref[1]{Remark~\ref{#1}}

\begin{document}

\begin{titlepage}

\title{Multi-component Toda lattice hierarchy}

\author{T.~Takebe\thanks{Takashi TAKEBE:
Beijing Institute of Mathematical Sciences and Applications,
No. 544, Hefangkou Village, Huaibei Town, Huairou District,
Beijing, 101408, People's Republic of China;
e-mail: takebe@bimsa.cn }
\and
A.~Zabrodin\thanks{Anton ZABRODIN:
Skolkovo Institute of Science and Technology, 143026, Moscow, Russia and
National Research University Higher School of
Economics,
20 Myasnitskaya Ulitsa, Moscow 101000, Russia, 
and NRC ``Kurchatov institute'', Moscow, Russia;
e-mail: zabrodin@itep.ru}}

\date{December 2024}
\maketitle

\vspace{-7cm} \centerline{ \hfill ITEP-TH-38/24}\vspace{7cm}

\hfill
{\it Dedicated to the memory of Masatoshi Noumi}

\vspace{1.0cm}

\begin{abstract}

We give a detailed account of the $N$-component Toda lattice
hierarchy which can be regarded as a generalization of 
the well-known Toda chain model and its non-abelian version. 
This hierarchy is an extension of the one
introduced earlier by Ueno and Takasaki. 
Our version contains $N$ 
discrete variables rather than one.
We start from the Lax formalism, deduce the bilinear
relation for the wave functions from it and then, based on the 
latter, prove existence of the tau-function. We also show how the
multi-component Toda lattice hierarchy is embedded into the 
universal hierarchy which is basically the multi-component KP
hierarchy. At last, we show how the bilinear integral equation
for the tau-function can be obtained using the free fermion
technique. An example of exact solutions (a multi-component
analogue of one-soliton solutions) is given. 

\end{abstract}

\end{titlepage}

\vspace{5mm}

%

\tableofcontents

\vspace{5mm}

\section{Introduction}

The Toda lattice \cite{Toda1,Toda2} 
is one of the most important  
integrable systems known at
present. It is of great interest from various points of view,
both mathematical and physical. For analysis of this system 
important mathematical methods were applied and 
further developed, such as 
inverse scattering method, B\"acklund transformations,
algebraic methods based on
representation theory of Lie groups and algebras, to name only
few. 

The original version of the model 
(sometimes called Toda chain) is the following non-linear
differential-difference equation:
\beq\label{int0}
    \frac{d^2 \varphi_s(t)}{dt^2}
    =
    e^{\varphi_{s}(t) - \varphi_{s-1}(t)}
    -
    e^{\varphi_{s+1}(t) - \varphi_{s}(t)}, \quad s\in \ZZ .
\eeq
It has two different physical interpretations. 

One of them
is to view (\ref{int0}) as Newtonian equations of motion for a
system of particles on the line with exponential interaction 
between nearest neighbors. In this case $\varphi_s(t)$ is coordinate
of the $s$-th particle at the time $t$. 
This is exactly the point of view that Toda had in mind when he
introduced equations (\ref{int0}) in \cite{Toda1}. In
fact, inspired by the work of Fermi-Pasta-Ulam \cite{FPU} and Ford
\cite{Ford}, Toda searched for nonlinear lattices that have exact
periodic solutions, using the concept of dual lattices, which he
introduced earlier \cite{Toda3}. He found the exponential interaction
in (\ref{int0}) from the addition formula for the function
$\mbox{sn}^2(u,k)$, where
$\mbox{sn}(u,k)$ is Jacobi's elliptic sinus function. 
(See Ch.2 of \cite{Toda2} or \cite{Toda4} for details.)

Another interpretation of \eqref{int0} is field-theoretical:
$\varphi_s(t)$ 
is regarded as a field in the two-dimensional space-time, with
the discrete space variable $s\in \ZZ$ and 
the continuous time variable $t$. To set up the problem, one should
fix the boundary conditions on the lattice. Three different cases
should be considered separately: 
a) the problem on the whole infinite lattice $\ZZ$
with the condition that $\varphi_s(t)\to 0$ as $s\to \pm \infty$,
b) the field $\varphi_s(t)$ is periodic in $s$ with some period
$n$, c) ``open'' boundary conditions when $s$ takes a finite
number of values, say from $1$ to $n$; in the mechanical 
interpretation this means that the system contains
only a finite number of particles. In the latter case, the system
(\ref{int0}) is sometimes referred to as ``Toda molecule''.
The solution in each of these three cases requires application
of different methods.

Soon after the discovery of the Toda chain various
methods of solving it (for example, the Hamilton-Jacobi theory
(\cite{McLaughlin75}), B\"acklund transformations
\cite{Flashka74,Wadati-Toda75}, Hirota's bilinear method
\cite{Hirota-Satsuma76}) were suggested. Some other methods
will be pointed out below. 

One of the most important methods of solving nonlinear
integrable systems is the inverse scattering
approach, which is based on the representations of the Lax or
Zakharov-Shabat type.
(The very fact 
that such representations exist is a consequence of rich 
hidden symmetries of the systems.) For the Toda chain, the Lax
representation was discovered in the works \cite{Flashka74,Manakov75},
and the complete integrability of the Toda chain, i.e., existence 
of enough number of independent integrals of motion in involution, 
was proved on this basis. (Also, in 1974 
Henon \cite{Henon74} found the integrals of
motion in the framework of the Hamiltonian formalism.)

Further, it was realized \cite{Bogoyavlensky76} that more
general integrable systems of the Toda type can be associated
with root systems of arbitrary simple Lie algebras (for a
generalization to superalgebras see \cite{Kirillova83}).
In the works \cite{OP79,K79} (see also \cite{Perelomov77,OP81},
the book \cite{Perelomov} and references therein), 
the equations of motion for the 
non-periodic case were explicitly integrated by means of
group-theoretical methods. It was also realized that the generalized 
Toda chains are connected with orbits of coadjoint representations
of solvable Lie groups \cite{Adler78,Symes80}. 

The periodic problem for the Toda chain is much more complicated
and requires methods of the theory of algebraic curves for its
solution. The solution of the periodic Toda chain in terms of
Riemann theta-functions associated with algebraic
curves was obtained by Date and Tanaka in \cite{Date-Tanaka75}, Mumford
in \cite{Mumford77} and Krichever in \cite{Krichever78}.

There exists also a two-dimensional extension of the Toda chain, 
in which the field $\varphi_s(x,y)$ depends on 
two continuous space-time variables $x,y$, one discrete 
variable $s\in\ZZ$ and satisfies the equation
\beq\label{int0a}
    \frac{\der^2 \varphi_s(x,y)}{\der x\, \der y} 
    =
    e^{\varphi_{s}(x,y) - \varphi_{s-1}(x,y)}
    -
    e^{\varphi_{s+1}(x,y) - \varphi_{s}(x,y)}, \quad s\in \ZZ .
\eeq
Its discovery goes back 
to works by Darboux \cite{Darboux} who 
implicitly introduced it as a chain
of Laplace-Darboux transformations of partial differential equations
(see \cite{Shabat95} and \cite{Okamoto87-88} for details). 
Nowadays, this model is called 2D Toda lattice. 
Its integrability was proved
by Mikhailov in \cite{Mikhailov79}. 

Like for the Toda chain, in the two-dimensional case
different types of boundary conditions
on the lattice are possible. 
A family of exact solutions in the case of 
open boundary conditions
(the 2D Toda molecule)
was found by Leznov and Saveliev \cite{LS79,LS81} 
by group-theoretical methods.
The multi-soliton solutions on the infinite lattice were obtained 
by Hirota (see his book \cite{Hirota} and references therein).
Periodic and quasi-periodic solutions were found by Krichever
in \cite{Krichever81} using the algebro-geometrical
methods earlier developed by him in \cite{Krichever78}.
In what follows, by the Toda lattice we understand its
two-dimensional version on the infinite lattice.

Later it became clear that the Toda lattice
equation can be embedded into an infinite hierarchy of compatible
nonlinear partial differential-difference equations. The theory of the
Toda lattice hierarchy in the framework 
of the approach suggested by the Kyoto school
\cite{Sato,Sato-Sato,DJKM83,JM83} was developed by Ueno and Takasaki in 
\cite{UT84}. Starting from the Lax and Zakharov-Shabat representations
for the Toda lattice,
they proved existence of a remarkable function called
{\it tau-function} which is common and fundamental for all known
infinite hierarchies of integrable 
differential or difference equations.
In a sense, the tau-function serves as 
a universal dependent variable of the hierarchy. 
It is a function
of infinitely many independent variables which are ``times'' 
parametrizing evolution along different integrable flows. 
In terms of the tau-function, all equations of the hierarchy are 
encoded in a generating bilinear integral equation. 
In our present work, we follow this approach. 
(See Takasaki's book \cite{Takasaki} for details.)

Since the seminal papers \cite{Sato,Sato-Sato,DJKM81} published in
1981--1982, it has been known that integrable equations 
and their hierarchies admit
generalizations in which the dependent variables become 
non-commutative quantities usually represented by matrices
of some finite size $N\times N$. They are referred to as
multi-component (or matrix) hierarchies.
For example, the multi-component
generalization of the Kadomtsev-Petviashvili (KP) hierarchy was 
discussed in \cite{DJKM81}--\cite{Teo11}. 
A matrix extension of the Toda chain equation (\ref{int0})
has the following form:
\beq\label{int0d}
\p_{t}(\p_t g_s\, g^{-1}_s)=g_sg^{-1}_{s-1}-
g_{s+1}g^{-1}_s,
\eeq
where $g_s=g_s(t)$ is an $N\times N$ invertible matrix depending 
on the time variable $t$. In the literature
this equation is known as the non-abelian Toda chain.
It can be regarded as 
a discrete analog of the principal chiral field
model. As is mentioned in \cite{Krichever81}, it was suggested 
by Polyakov, who also found its integrals of motion (unpublished).
Periodic and quasi-periodic solutions to the non-abelian Toda chain
were constructed by Krichever in \cite{Krichever81} 
by methods of algebraic geometry.
For the Hamiltonian formulation, classical $r$-matrix structure
and quantization see \cite{Korepin81}.

An $N$-component
generalization of the Toda lattice hierarchy was suggested 
by Ueno and Takasaki in \cite{UT84}. 
In the framework of the theory developed in \cite{UT84}, the non-abelian 
Toda chain can be obtained from the more general 
$N$-component Toda lattice hierarchy 
by a certain reduction and a special restriction imposed to 
the independent variables. 
The two-dimensional version of (\ref{int0}), 
equation (\ref{int0a}), is generalized to 
the following equation for 
an $N\times N$ invertible matrix $g_s=g_s(x,y)$:
\beq\label{int0b}
\p_{y}(\p_x g_s\, g^{-1}_s)=g_sg^{-1}_{s-1}-
g_{s+1}g^{-1}_s.
\eeq
This equation is called the non-abelian two-dimensional Toda 
lattice (or simply non-abelian Toda lattice).

As it becomes clear from the contemporary point of view,
the treatment of the multi-component systems
in \cite{UT84} 
was not complete. In particular, the tau-function
of the multi-component Toda lattice hierarchy was not discussed there.
The aim of our paper is to
give a more comprehensive account of the theory and suggest a natural
extension of the system by including into play 
$N$ discrete variables rather than
one (as it was in \cite{UT84}). There are also 
some other differences between \cite{UT84} and our treatment.
In what follows,
we comment on them as they occur in the presentation, in footnotes.

The independent variables of the $N$-component Toda lattice
hierarchy are $2N$ 
infinite sets of ``times''
$$
{\boldsymbol t} =\{{\boldsymbol t} _1, {\boldsymbol t} _2, \ldots , {\boldsymbol t} _N\}, \qquad
{\boldsymbol t}_{\alpha}=\{t_{\alpha , 1}, t_{\alpha , 2}, 
t_{\alpha , 3}, \ldots \, \},
\qquad \alpha = 1, \ldots , N
$$
and
$$
\bar {\boldsymbol t} =\{\bar {\boldsymbol t} _1, \bar {\boldsymbol t} _2, \ldots , 
\bar {\boldsymbol t} _N\}, \qquad
\bar {\boldsymbol t} _{\alpha}=\{ \bar t_{\alpha , 1}, \bar t_{\alpha , 2}, 
\bar t_{\alpha , 3}, \ldots \, \},
\qquad \alpha = 1, \ldots , N
$$
which in the algebraic approach can be in general 
regarded as complex numbers.
Besides, in our, more general, version there are $N$ discrete 
variables $\bolds =\{s_1, \ldots , s_N\}$ which are
integers. Equation (\ref{int0a}) corresponds to the simplest 
case $N=1$, $x=t_{1,1}$, $y=\bar t_{1,1}$.

Let us describe the contents of the paper in some detail. 

We develop three different approaches 
which are shown to lead to the same result.
The starting point of one of them is the Lax representation
in terms of difference Lax operators $\boldL$,
$\boldLbar$ and auxiliary difference operators 
$\boldU_{\alpha}$, $\boldUbar_{\alpha}$,
$\boldQ_{\alpha}$, $\boldQbar_{\alpha}$
(or $\boldP_{\alpha}$, $\boldPbar_{\alpha}$)
with $N\times N$ 
matrix coefficients. The logical structure of our consideration
can be illustrated by the following chain:
\begin{multline*}
\mbox{Lax representation} \longrightarrow
\mbox{Zakharov-Shabat representation} \longrightarrow
\\
\longrightarrow
\mbox{wave functions} \longrightarrow
\mbox{tau-function}.
\end{multline*}
Namely, starting from the Lax representation, we prove its equivalence
to the Zakharov-Shabat representation, then prove existence of so-called
wave operators (which are sometimes called dressing operators) and use them
to introduce $N\times N$ matrix wave functions 
$\boldPsi (\bolds , {\boldsymbol t} , \bar {\boldsymbol t} ;z)$,
$\boldPsibar (\bolds , {\boldsymbol t} , \bar {\boldsymbol t} ;z)$
(and the adjoint wave functions
$\boldPsi^* (\bolds , {\boldsymbol t} , \bar {\boldsymbol t} ;z)$,
$\boldPsibar^* (\bolds , {\boldsymbol t} , \bar {\boldsymbol t} ;z)$)
depending on all the times and on
a complex spectral parameter $z$. The wave functions 
satisfy a system of linear equations whose compatibility
implies nonlinear equations of the Toda lattice hierarchy and thus
provide a linearization of it.
We give a detailed proof that the wave functions of the 
multi-component Toda lattice hierarchy satisfy the integral bilinear
identity
\beq\label{int1}
\oint_{C_\infty}\!\!
    \boldPsi(\bolds,{\boldsymbol t} ,\bar {\boldsymbol t};z)\,
    \boldPsi^*(\bolds',{\boldsymbol t} ',\bar {\boldsymbol t}';z)\,
    dz
    =
    \oint_{C_0}\!\!
    \bar\boldPsi(\bolds,{\boldsymbol t} ,\bar {\boldsymbol t};z)\,
    \bar\boldPsi^*(\bolds',{\boldsymbol t} ',\bar {\boldsymbol t}';z)\,
    dz,
    \eeq
valid for all $\bolds , {\boldsymbol t} , \bar {\boldsymbol t} , 
\bolds ', {\boldsymbol t} ', \bar {\boldsymbol t} '$, where $C_{\infty}$ and $C_0$
are contours encircling $\infty$ and $0$ respectively
(Proposition \ref{predl:bil-res}). Next we prove that the identity
(\ref{int1}) implies existence of a tau-function which is an
$N\times N$ matrix with matrix elements $\tau_{\alpha \beta}
(\bolds , {\boldsymbol t} , \bar {\boldsymbol t} )$, $\alpha , \beta =1, \ldots , N$
(Theorem \ref{theorem:tau-function}).
The tau-function is shown to satisfy the following integral
bilinear equation:
\beq\label{int2}
\begin{array}{l}
\displaystyle{
\sum_{\gamma =1}^N (-1)^{\delta_{\beta \gamma}}\oint_{C_{\infty}}
z^{s_{\gamma}-s'_{\gamma}+\delta_{\alpha \gamma}+\delta_{\beta \gamma}
-2}e^{\xi ({\boldsymbol t} _{\gamma}-{\boldsymbol t} '_{\gamma},z)}
\tau_{\alpha \gamma}(\bolds , {\boldsymbol t} -[z^{-1}]_{\gamma}, \bar {\boldsymbol t} )
\tau_{\gamma \beta}(\bolds ', {\boldsymbol t} '+[z^{-1}]_{\gamma}, \bar {\boldsymbol t} ')
\, dz}
\\ \\
\displaystyle{
=\sum_{\gamma =1}^N (-1)^{\delta_{\alpha \gamma}}\oint_{C_{\infty}}
z^{s_{\gamma}'-s_{\gamma}-2}e^{\xi (\bar 
{\boldsymbol t} _{\gamma}-\bar {\boldsymbol t} '_{\gamma},z)}}
\\ \\
\displaystyle{\phantom{aaaaaaaaaaaaaaaa}
\times \tau_{\alpha \gamma}(\bolds +[1]_{\gamma}, 
{\boldsymbol t} , \bar {\boldsymbol t} -[z^{-1}]_{\gamma})
\tau_{\gamma \beta}(\bolds '-[1]_{\gamma}, 
{\boldsymbol t} ', \bar {\boldsymbol t} '+[z^{-1}]_{\gamma})
\, dz},
\end{array}
\eeq
which is valid for all
$\bolds , {\boldsymbol t} , \bar {\boldsymbol t} , 
\bolds ', {\boldsymbol t} ', \bar {\boldsymbol t} '$ (Theorem \ref{theorem:tautau}). 
Here
$\displaystyle{
\xi ({\boldsymbol t} _{\gamma},z)=\sum_{k\geq 1}t_{\gamma ,k}z^k,}
$
$
{\boldsymbol t} \pm [z^{-1}]_{\gamma}$ is the set ${\boldsymbol t} $ in which
$t_{\gamma ,k}$ is shifted by $\pm \frac{1}{k}\, z^{-k}$ and
$\bolds \pm [1]_{\gamma}$ is the set $\bolds $ in which 
$s_{\gamma}$ is shifted by $\pm 1$. Equation (\ref{int2}) is the
generating equation which encodes all differential-difference 
equations of the hierarchy. They are obtained from it by expanding
both sides 
in the Taylor series in ${\boldsymbol t} -{\boldsymbol t} '$, $\bar {\boldsymbol t} -\bar {\boldsymbol t} '$.
This approach is developed in Sections 2--5. In Section 
\ref{subsection:Hirota-Miwa} we also obtain various bilinear
equations of the Hirota-Miwa type as corollaries of (\ref{int2}).

The second approach developed in Section 6 
is based on the universal hierarchy
named so in the recent paper \cite{KZ23}. It is basically the 
multi-component KP hierarchy in which the discrete variables 
are allowed to take arbitrary complex values (but equations 
with respect to them are difference).
The tau-function of this hierarchy is a matrix-valued function obeying
a bilinear integral equation. We show how the $N$-component
Toda lattice hierarchy can be embedded into the $2N$-component universal
hierarchy and define the tau-function of the former in terms of
that of the latter. As a result, we obtain a bilinear integral
equation for the tau-function of the Toda lattice hierarchy introduced
in
this way and show that it becomes equivalent to (\ref{int2})
after a simple redefinition of the tau-function by multiplying its
matrix elements by certain sign factors. 

The third approach which is followed in Section 7 is based on the free
fermion technique developed by the Kyoto school 
\cite{DJKM83,JM83,DJM93}. In this approach, the tau-function is
represented as a vacuum expectation value of certain fermionic
operators. In order to describe multi-component hierarchies,
one should work with multi-component fermionic operator fields
$\psi^{(\alpha )}(z)$, $\psi^{*(\alpha )}(z)$,
$\alpha =1, \ldots , N$ obeying the standard
anticommutation relations. The basic fact from which the 
bilinear equation for the tau-function follows is the operator
bilinear identity
\begin{equation}
    \sum_{\gamma =1}^N \mbox{res} \left [\frac{dz}{z} \, 
    \psi^{(\gamma )}(z)g\otimes \psi^{*(\gamma )}(z)g\right ]
    =
    \sum_{\gamma =1}^N \mbox{res} \left [\frac{dz}{z} \, 
    g\psi^{(\gamma )}(z)\otimes g\psi^{*(\gamma )}(z)\right ]
\label{int3}
\end{equation}
characterizing the Clifford group elements $g$ of the general form
\beq\label{int4}
g=\exp \Bigl (\sum_{\alpha, \beta =1}^N
\oint \!\! \oint A^{(\alpha \beta )}(z, w)
\psi^{*(\alpha )}(z)\psi^{(\beta )}(w)
dz dw\Bigr ),
\eeq
where $A^{(\alpha \beta )}(z, w)$ is some function of two
complex variables. 
(The operation res in (\ref{int3}) 
is defined as $\displaystyle{\mbox{res}\Bigl (
\sum_k a_k z^k \, dz\Bigr )=a_{-1}}$.)
The bilinear equation for the tau-function is
obtained as a corollary of (\ref{int3}) after acting
by both sides to a tensor product of 
certain states from the fermionic Fock space
and applying the 
(non-abelian) bosonization rules \cite{KL93}. This job is done
in Section 7.2. As a result, we obtain the integral bilinear
equation for the tau-function defined as the expectation value
\beq\label{int5}
\tau_{\alpha \beta}(\bolds , {\boldsymbol t} , \bar {\boldsymbol t} )=
(-1)^{|\bolds |(|\bolds |-1)/2}\, 
\bigl <\bolds +[1]_{\alpha}-[1]_{\beta}\bigl |e^{J({\boldsymbol t} )}g
e^{-\bar J(\bar {\boldsymbol t} )}\bigr |\bolds \bigr >, \qquad
|\bolds |=\sum_{\gamma =1}^N s_{\gamma}.
\eeq
The operators $J({\boldsymbol t} )$, $\bar J(\bar {\boldsymbol t} )$ are defined in 
(\ref{JJ}). 
The bilinear equation obtained in this way 
is the same as the one from Section 6. (For the one-component Toda
lattice hierarchy this approach was studied in \cite{KMMOZ91},
\cite{Tak91} and \cite{GM92}.) In Section \ref{subsection:example}
we give an example of exact solution which is a multi-component
analogue of the one-soliton solution to the Toda lattice.

In Appendix we show that the non-abelian two-dimensional Toda lattice
\eqref{int0b} is indeed included in our multi-component Toda lattice
hierarchy.

\section{Lax formalism for the multi-component Toda lattice hierarchy}
\label{sec:def-multi-toda}

In this section we define the $N$-component Toda
lattice hierarchy in the Lax formalism, essentially following \S3.1 of
\cite{UT84} but there are differences from their formulation in the
following  points:
\begin{itemize}
 \item We replace $\Integer\times\Integer$-matrices in \cite{UT84} by
       difference operators.
 \item We introduce the set of $N$ discrete variables
 $\{s_1,\dotsc, s_N\}\in\Integer^N$ which we denote by $\bolds$.
 In \cite{UT84} only one discrete variable
       $s$ was introduced and the shift operator corresponded to a matrix
       ${\boldsymbol\Lambda}$ acting on the space
       $\Comp^\Integer\tensor\Comp^N$. 
\end{itemize}

The continuous independent variables of the multi-component Toda
hierarchy are $2N$ infinite sets of ``times'' 
$$ {\boldsymbol t}
=\{{\boldsymbol t} _1, {\boldsymbol t} _2, \ldots , {\boldsymbol t}
_N\}, \qquad {\boldsymbol t} _{\alpha}=\{t_{\alpha , 1}, t_{\alpha , 2},
t_{\alpha , 3}, \ldots \, \}, \qquad \alpha = 1, \ldots , N 
$$ 
and 
$$
\bar {\boldsymbol t} =\{\bar {\boldsymbol t} _1, \bar {\boldsymbol t}
_2, \ldots , \bar {\boldsymbol t} _N\}, \qquad 
\bar {\boldsymbol t}_{\alpha}=
\{ \bar t_{\alpha , 1}, \bar t_{\alpha , 2}, 
\bar t_{\alpha , 3}, \ldots \, \}, 
\qquad \alpha = 1, \ldots , N $$ which are in general
complex numbers.  The sets ${\boldsymbol t} $ and $\bar {\boldsymbol t}
$ are often called ``positive'' and ``negative'' times.  Besides, there
are discrete ``zeroth times'' $s_1, \ldots , s_N$ which are integer
numbers. The set of ``zeroth times'' is denoted as $\bolds =\{s_1,
\ldots , s_N\}$.

We use the following notations:
\begin{itemize}
 \item Diagonal matrix units
$
    E_{\alpha}
    :=
    (\delta_{\mu \alpha}\delta_{\nu \alpha})_{\mu ,\nu=1,\dotsc,N}
    \in\Mat(N\times N,\Comp)
$.
%

 \item The unit matrix, $1_N:=(\delta_{\a \b})_{\a,\b=1,\dotsc,N}$,

 \item The shift operators
       $e^{\der_{s_\a}}f(\bolds):=f(\bolds+[1]_\a)$, where
       the notation $\bolds+[1]_\a$ means that the $\alpha$-th
       component is shifted by 1 and all other components remain
       unchanged. We define the total
       shift operator $e^{\der_s}$ by $e^{\der_s}=
       e^{\der_{s_1}+\dotsb+\der{s_N}}$. The operator $e^{\der_s}$ acts as
       $e^{\der_s}f(\bolds)=f(\bolds+\boldone)$, where
       $\bolds+\boldone=\{s_1+1,\dotsc,s_N+1\}$. 

 \item      A difference operator $A$ acting on an $N\times N$-matrix
       function $f(\bolds)$ as
\[
    Af(\bolds)
    =
    \sum_{j\in\Integer} a_j(\bolds)\, f(\bolds+j \boldone),
    \qquad
    a_j(\bolds) \in\Mat(N\times N,\Comp),
\]
 is expressed in the form
\[
    A
    =
    \sum_{j\in\Integer} a_j(\bolds)\, e^{j\der_s}.
\]
       For such an operator we define the non-negative shift part and
       the negative shift part as
\begin{equation}
    A_{\geq0}:=
    \sum_{j\geq0} a_j(\bolds)\, e^{j\der_s},
    \quad
    A_{<0}:=
    \sum_{j<0} a_j(\bolds)\, e^{j\der_s}.
\label{A+-}
\end{equation}

\end{itemize}

\subsection{Lax and Zakharov-Shabat equations}
\label{subsection:LZS}

In the framework of the Lax formalism, 
the multi-component Toda lattice hierarchy is a
system of partial differential equations for the following matrix-valued
difference operators\footnote{%
In \cite{UT84} it was assumed that $\bar b_0(\bolds)$ and $\bar
u_{\alpha,0}(\bolds)$ were expressed through a 
separately introduced matrix
$\tilde w_0(\bolds)$. We will show in \propref{proposition:tildew0} that
this follows from the algebraic relations imposed below on the 
difference operators. The 
operators $\boldQ_\alpha(\bolds)$
and $\bar \boldQ_\alpha(\bolds)$ did not appear in \cite{UT84}.
}:
\begin{equation}
 \begin{aligned}
    \boldL(\bolds) 
    &= \sum_{j=0}^\infty b_j(\bolds)\, e^{(1-j)\der_s}, &
    b_0(\bolds) &=1_N,
\\
    \boldLbar(\bolds) 
    &= \sum_{j=0}^\infty \bar b_j(\bolds)\, e^{(j-1)\der_s}, &
    \bar b_0(\bolds)
    &\in\Mat(N\times N,\Comp)\; \mbox{is invertible},
 \end{aligned}
\label{def:L,Lbar}
\end{equation}
\begin{equation}
 \begin{aligned}
    \boldU_\alpha(\bolds) 
    &= \sum_{j=0}^\infty u_{\alpha,j}(\bolds)\, e^{-j\der_s}, &
    u_{\alpha,0}(\bolds) &= E_\alpha,
\\
    \boldUbar_\alpha(\bolds) 
    &= \sum_{j=0}^\infty \bar u_{\alpha,j}(\bolds)\, e^{j\der_s},&
    \bar u_{\alpha,0}(\bolds) 
    &\in\Mat(N\times N,\Comp),
 \end{aligned}
\label{def:U,Ubar}
\end{equation}
\begin{equation}
 \begin{aligned}
    \boldQ_\alpha(\bolds) &= e^{-\der_{s_\alpha}} \boldP_\alpha(\bolds),
    \ 
    \boldP_\alpha(\bolds) 
    = \sum_{j=0}^\infty p_{\alpha,j}(\bolds)\, e^{-j\der_s}, &
    p_{\alpha,0}(\bolds) &= 1_N,
\\
    \boldQbar_\alpha(\bolds) 
    &= e^{-\der_{s_\alpha}} \boldPbar_\alpha(\bolds),
    \ 
    \boldPbar_\alpha(\bolds) 
    = \sum_{j=0}^\infty \bar p_{\alpha,j}(\bolds)\, e^{j\der_s},&
    \bar p_{\alpha,0}(\bolds) 
    &\in\Mat(N\times N,\Comp).
 \end{aligned}
\label{def:Q,Qbar}
\end{equation}
Here $\alpha$ is an index from $1$ to $N$ and $b_j(\bolds)$, 
$\bar b_j(\bolds)$, $u_{\alpha,j}(\bolds)$, $\bar u_{\alpha,j}(\bolds)$,
$p_{\alpha,j}(\bolds)$ and $\bar p_{\alpha,j}(\bolds)$ 
are matrix-valued functions of size $N\times N$ depending
on the variables $\bolds=\{s_1,\dotsc,s_N\}$,
${\boldsymbol t} $ and
$\boldtbar$.
(We do not write dependence on $\boldt$ and $\boldtbar$ explicitly
unless it is necessary.)

We require that these operators satisfy 
the following algebraic conditions: for any
$\alpha,\beta=1,\dotsc,N$ it holds that
\begin{gather}
    [\boldL(\bolds), \boldU_\alpha(\bolds)] =
    [\boldL(\bolds), \boldQ_\alpha(\bolds)] =
    [\boldU(\bolds), \boldQ_\alpha(\bolds)] =
    [\boldQ_\alpha(\bolds), \boldQ_\beta(\bolds)] = 0, 
\label{commutative:L,U,Q}
\\
    \boldU_\alpha(\bolds) \boldU_\beta(\bolds)
    = \delta_{\alpha\beta} \boldU_\beta(\bolds),\qquad
    \sum_{\alpha=1}^N \boldU_\alpha(\bolds) = 1_{N}, 
\label{U:alg-cond}
\\
    \prod_{\alpha=1}^N \boldQ_\alpha(\bolds)
    =
    \boldL^{-1}(\bolds),
\label{Q:alg-cond}
\\
    [\boldLbar(\bolds), \boldUbar_\alpha(\bolds)] =
    [\boldLbar(\bolds), \boldQbar_\alpha(\bolds)] =
    [\boldUbar(\bolds), \boldQbar_\alpha(\bolds)] =
    [\boldQbar_\alpha(\bolds), \boldQbar_\beta(\bolds)] = 0, 
\label{commutative:Lbar,Ubar,Qbar}
\\
    \boldUbar_\alpha(\bolds) \boldUbar_\beta(\bolds) 
    = \delta_{\alpha\beta} \boldUbar_\beta(\bolds),\qquad
    \sum_{\alpha=1}^N \boldUbar_\alpha(\bolds) = 1_{N},
\label{Ubar:alg-cond}
\\
    \prod_{\alpha=1}^N \boldQbar_\alpha(\bolds)
    =
    \boldLbar(\bolds),
\label{Qbar:alg-cond}
\end{gather}
and 
\begin{equation}
    \boldQ_\a(\bolds)
    \sum_{\b =1}^N
    \boldU_\b(\bolds) \boldL^{\delta_{\a \b}}(\bolds)
    =
    \boldQbar_\a(\bolds)
    \sum_{\b =1}^N
    \boldUbar_\b (\bolds) \boldLbar^{-\delta_{\a \b }}(\bolds)
\label{ULQ=UbarLbarQbar:basis}
\end{equation}
for each $\alpha=1,\dotsc,N$.

Under these conditions, the Lax equations for the multi-component Toda
lattice hierarchy is the following system: for any indices
$\alpha,\beta=1,\dotsc,N$, $n=1,2,\dotsc$ and for any operators $\boldA=
\boldL(\bolds)$, $\boldLbar(\bolds)$, $\boldU_\beta(\bolds)$,
$\boldUbar_\beta(\bolds)$, $\boldQ_\beta(\bolds)$,
$\boldQbar_\beta(\bolds)$ it holds
\begin{equation}
    \frac{\der \boldA}{\der t_{\alpha,n}}
    =
    [\boldB_{\alpha,n}(\bolds),\boldA],
    \qquad
    \frac{\der \boldA}{\der \bar t_{\alpha,n}}
    =
    [\boldBbar_{\alpha,n}(\bolds),\boldA],
\label{lax:L,Lbar,U,Ubar,Q,Qbar}
\end{equation}%
where
\begin{equation}
    \boldB_{\alpha,n}(\bolds)
    := (\boldL^n (\bolds) \boldU_\alpha(\bolds))_{\geq0},\qquad
    \boldBbar_{\alpha,n}(\bolds)
    := (\boldLbar^n (\bolds) \boldUbar_\alpha(\bolds))_{<0}.
\label{def:BBbar}
\end{equation}
Actually, discrete Lax equations for the shifts of
$s_{\alpha}$ by 1 are included in \eqref{commutative:L,U,Q} and
\eqref{commutative:Lbar,Ubar,Qbar}. (See (\ref{PL,PU}) and
(\ref{PbarLbar,PbarUbar}) below.) The meaning of condition
\eqref{ULQ=UbarLbarQbar:basis} is related to this discrete
evolution. (See \remref{rem:meaning-ULQ-condition}.)

\begin{lem}
\label{lem:ULQ=UbarLbarQbar:general}
 The set of conditions \eqref{ULQ=UbarLbarQbar:basis}
 ($\alpha=1,\dotsc,N$) is equivalent to
\begin{equation}
    \prod_{\alpha=1}^N
    \boldQ_\alpha^{a_\alpha}(\bolds)
    \sum_{\beta=1}^N
    \boldU_\beta(\bolds) \boldL^{a_\beta}(\bolds)
    =
    \prod_{\alpha=1}^N
    \boldQbar_\alpha^{a_\alpha}(\bolds)
    \sum_{\beta=1}^N
    \boldUbar_\beta(\bolds) \boldLbar^{-a_\beta}(\bolds).
\label{ULQ=UbarLbarQbar:general}
\end{equation}
 for all $\bolda=\{a_1,\dotsc,a_N\}\in\Integer^N$.
\end{lem}

\begin{proof}
%
 As the $\bolds$-variables are fixed in this lemma, we do not write them
 explicitly in the proof. 

 Because of the commutativity \eqref{commutative:L,U,Q} and the
 idempotence \eqref{U:alg-cond} we have:
\begin{multline*}
    \prod_{\alpha=1}^N
    \boldQ_\alpha^{a_\alpha}
    \bigl(
     \boldU_1 \boldL^{a_1} + \dotsb + \boldU_N \boldL^{a_N}
    \bigr)
    \times
    \prod_{\alpha=1}^N
    \boldQ_\alpha^{b_\alpha}
    \bigl(
     \boldU_1 \boldL^{b_1} + \dotsb + \boldU_N \boldL^{b_N}
    \bigr)
\\
    =
    \prod_{\alpha=1}^N
    \boldQ_\alpha^{a_\alpha+b_\alpha}
    \bigl(
     \boldU_1 \boldL^{a_1+b_1} + \dotsb + \boldU_N \boldL^{a_N+b_N}
    \bigr),
\end{multline*}
and, in particular,
\[
    \left(
    \prod_{\alpha=1}^N
    \boldQ_\alpha^{a_\alpha}
    \sum_{\beta=1}^N
    \boldU_\beta \boldL^{a_\beta}
    \right)^{-1}
    =
    \prod_{\alpha=1}^N
    \boldQ_\alpha^{-a_\alpha}
    \sum_{\beta=1}^N
    \boldU_\beta \boldL^{-a_\beta}.
\]
 Similar equations hold also for $\boldLbar$ and $\boldUbar_\alpha$'s
 because of \eqref{commutative:Lbar,Ubar,Qbar} and
 \eqref{Ubar:alg-cond}. 
 The equivalence of \eqref{ULQ=UbarLbarQbar:basis} ($\alpha=1,\dotsc,N$)
 and  \eqref{ULQ=UbarLbarQbar:general} follows from these formulae.
\end{proof}
 
\begin{zam}
\label{rem:ueno-takasaki}
 The above defined system is an extended version of the multi-component
 Toda lattice hierarchy introduced by Ueno and Takasaki in \S3.1 of
 \cite{UT84}. To recover their version we have only to restrict
 $\bolds$-variables to the form
 $\bolds=\bolds^{(0)}+s\boldone=\{s_1^{(0)}+s,\dotsc,s_N^{(0)}+s\}$ for
 a fixed $\bolds^{(0)}=\{s_1^{(0)},\dotsc,s_N^{(0)}\}$ passing to the
 single variable $s\in\Integer$. In this case, there are only shifts of
 the form $\bolds\mapsto \bolds+n \boldone$ ($n\in\Integer$). Then all
 functions $f(\bolds)$ are regarded as functions $f(s)$ of $s$ only and
 the operator $e^{\der_s}$ acts as the shift operator:
 $e^{\der_s}f(s)=f(s+1)$.

 The algebraic conditions \eqref{commutative:L,U,Q}, \eqref{U:alg-cond}
 for $\boldL$ and $\boldU_\alpha$ and conditions
 \eqref{commutative:Lbar,Ubar,Qbar}, \eqref{Ubar:alg-cond} for
 $\boldLbar$ and $\boldUbar_\alpha$ are those required by (3.1.6) in
 \cite{UT84}. The differential equations
 \eqref{lax:L,Lbar,U,Ubar,Q,Qbar} for $\boldL$, $\boldLbar$, $\boldU$
 and $\boldUbar$ are the system (3.1.8) in \cite{UT84}. The operators
 $\boldQ_\alpha$ and $\boldQbar_\alpha$ were not considered in
 \cite{UT84}.  Note that for $\bolda$ of the form $\bolda = a{\bf 1}_N$
 condition \eqref{ULQ=UbarLbarQbar:general} is trivial because of the
 other algebraic conditions and we do not have to require
 \eqref{ULQ=UbarLbarQbar:basis}.
\end{zam}

\begin{predl}
\label{predl:zs}
 (i)
 The system \eqref{lax:L,Lbar,U,Ubar,Q,Qbar} for $\boldL(\bolds)$,
 $\boldLbar(\bolds)$, $\boldU_\alpha(\bolds)$ and
 $\boldUbar_\alpha(\bolds)$ implies the following compatibility
 conditions under the algebraic constraints \eqref{commutative:L,U,Q},
 \eqref{U:alg-cond}, \eqref{commutative:Lbar,Ubar,Qbar} and
 \eqref{Ubar:alg-cond}:
\begin{align}
    &\left[
     \frac{\der}{\der t_{\alpha,m}} - \boldB_{\alpha,m}(\bolds),
     \frac{\der}{\der t_{\beta ,n}} - \boldB_{\beta, n}(\bolds)
    \right]
    = 0,
\label{zs-BB}
\\
    &\left[
     \frac{\der}{\der \bar t_{\alpha,m}} - \boldBbar_{\alpha,m}(\bolds),
     \frac{\der}{\der \bar t_{\beta ,n}} - \boldBbar_{\beta, n}(\bolds)
    \right]
    = 0,
\label{zs-BbarBbar}
\\
    &\left[
     \frac{\der}{\der t_{\alpha,m}} - \boldB_{\alpha,m}(\bolds),
     \frac{\der}{\der \bar t_{\beta ,n}} - \boldBbar_{\beta, n}(\bolds)
    \right]
    = 0.
\label{zs-BBbar}
\\
    &\left[
     \frac{\der}{\der t_{\alpha,m}}
     + (\boldL^m(\bolds) \boldU_\alpha(\bolds))_{<0},
     \frac{\der}{\der t_{\beta ,n}}
     + (\boldL^n(\bolds) \boldU_\beta(\bolds))_{<0}
    \right]
    = 0,
\label{zs-LU<0LU<0}
\\
    &\left[
     \frac{\der}{\der \bar t_{\alpha,m}} 
      + (\boldLbar^m(\bolds) \boldUbar_\alpha(\bolds))_{\geq0},
     \frac{\der}{\der \bar t_{\beta ,n}}
      + (\boldLbar^n(\bolds) \boldUbar_\beta(\bolds))_{\geq0}
    \right]
    = 0,
\label{zs-LUbar>0LUbar>0}
\\
    &\left[
     \frac{\der}{\der t_{\alpha,m}}
      + (\boldL^m(\bolds) \boldU_\alpha(\bolds))_{<0},
     \frac{\der}{\der \bar t_{\beta ,n}} - \boldBbar_{\beta,n}(\bolds)
    \right]
    = 0,
\label{zs-LU<0Bbar}
\\
    &\left[
     \frac{\der}{\der t_{\alpha,m}} - \boldB_{\alpha,m}(\bolds),
     \frac{\der}{\der \bar t_{\beta ,n}}
      + (\boldLbar^n(\bolds) \boldUbar_\beta(\bolds))_{\geq0}
    \right]
    = 0,
\label{zs-BLUbar>0}
\end{align}

 (ii)
 Conversely, the system \eqref{zs-BB}, \eqref{zs-BbarBbar} and
 \eqref{zs-BBbar} with the algebraic constraints
 \eqref{commutative:L,U,Q}, \eqref{U:alg-cond},
 \eqref{commutative:Lbar,Ubar,Qbar} and \eqref{Ubar:alg-cond} implies
 the Lax representation \eqref{lax:L,Lbar,U,Ubar,Q,Qbar} for
 $\boldL(\bolds)$, $\boldLbar(\bolds)$, $\boldU_\alpha(\bolds)$ and
 $\boldUbar_\alpha(\bolds)$.
\end{predl}

We call the system \eqref{zs-BB}, \eqref{zs-BbarBbar} and
\eqref{zs-BBbar} the {\em Zakharov-Shabat representation} of the
multi-component Toda lattice hierarchy.

\begin{proof}

 As the $\bolds$-variables are fixed in this proposition, we do
 not write them explicitly in the proof, unless it is necessary.

 (i)
 Derivation of the Zakharov-Shabat type equations from the Lax
 representation is the same as that for the one-component system as
 shown in \cite{UT84}. Indeed, applying
 $\der_{t_{\alpha,m}}-[\boldB_{\alpha,m}, \cdot]$ to $\boldL^n
 \boldU_{\beta}$, we obtain
\[
    \frac{\der}{\der t_{\alpha,m}}(\boldL^n \boldU_\beta)
    =
    [\boldB_{\alpha,m}, \boldL^n \boldU_\beta]
    =
    [\boldB_{\alpha,m}, \boldB_{\beta,n}] +
    [\boldB_{\alpha,m}, (\boldL^n \boldU_\beta)_{<0}].
\]
 from the Lax equation \eqref{lax:L,Lbar,U,Ubar,Q,Qbar} for $\boldL$ and
 $\boldU$. Exchanging $(\alpha,m)$ and $(\beta,n)$, we also have
\[
 \begin{split}
    \frac{\der}{\der t_{\beta,n}}(\boldL^m \boldU_\alpha)
    &=
    [\boldB_{\beta,n}, \boldL^m \boldU_\alpha]
    =
    [\boldL^n \boldU_\beta - (\boldL^n \boldU_\beta)_{<0}, 
     \boldL^m \boldU_\alpha]
\\
    &=
    -[(\boldL^n \boldU_\beta)_{<0}, \boldL^m \boldU_\alpha]
    =
    -[(\boldL^n \boldU_\beta)_{<0}, \boldB_{\alpha,m}]
    -[(\boldL^n \boldU_\beta)_{<0}, (\boldL^m \boldU_\alpha)_{<0}],
 \end{split}
\]
 because of \eqref{commutative:L,U,Q} and
 \eqref{U:alg-cond}. Subtracting the latter equation from the former, we
 have
\begin{multline*}
    \frac{\der \boldB_{\beta,n}}{\der t_{\alpha,m}}
    -
    \frac{\der \boldB_{\alpha,m}}{\der t_{\beta,n}}
    - [\boldB_{\alpha,m},\boldB_{\beta,n}]
\\
    =
    \frac{\der (\boldL^m\boldU_\alpha)_{<0}}{\der t_{\beta,n}}
    -
    \frac{\der (\boldL^n\boldU_\beta)_{<0}}{\der t_{\alpha,m}}
    +
    [(\boldL^n \boldU_\beta)_{<0}, (\boldL^m \boldU_\alpha)_{<0}].
\end{multline*}
 The left-hand side is a difference operator with non-negative shifts,
 while the right-hand side is a difference operator with negative
 shifts, which implies that the both sides are zero. This implies
 \eqref{zs-BB} and \eqref{zs-LU<0LU<0}.

 The equations \eqref{zs-BbarBbar} and \eqref{zs-LUbar>0LUbar>0} are
 shown in the same way by changing ``${\boldsymbol t} $'' to ``$\bar
 {\boldsymbol t} $'', ``$\boldL$'' to ``$\boldLbar$'', ``$\boldB$'' to
 ``$\boldBbar$'' and ``$\boldU$'' to ``$\boldUbar$'' in the above
 argument.

 The equations
\[
    \frac{\der}{\der \bar t_{\beta, n}} \boldL^m\boldU_\alpha
    =
    [\boldBbar_{\beta,n}, \boldL^m\boldU_\alpha], \quad
    \frac{\der}{\der t_{\alpha, m}} \boldLbar^n\boldUbar_\beta
    =
    [\boldB_{\alpha,m}, \boldLbar^n\boldUbar_\beta], 
\]
 follow from the Lax equations
 \eqref{lax:L,Lbar,U,Ubar,Q,Qbar}. Therefore, we have
\begin{align*}
    \frac{\der \boldB_{\alpha,m}}{\der \bar t_{\beta,n}}
    - [\boldBbar_{\beta,n}, \boldB_{\alpha,m}]
    &=
    - \frac{\der}{\der \bar t_{\beta,n}}
      (\boldL^m \boldU_\alpha)_{<0}
    + [\boldBbar_{\beta,n}, (\boldL^m \boldU_\alpha)_{<0}],
\\
    \frac{\der \boldBbar_{\beta,n}}{\der t_{\alpha,m}}
    - [\boldB_{\alpha,m}, \boldBbar_{\beta,n}]
    &=
    - \frac{\der}{\der t_{\alpha,m}}
      (\boldLbar^n \boldUbar_\beta)_{\geq0}
    + [\boldB_{\alpha,m}, (\boldLbar^n \boldUbar_\beta)_{\geq0}].
\end{align*}
 Using these formulae, we can rewrite 
$
    \der_{\bar t_{\beta,n}}\boldB_{\alpha,m}
    -
    \der_{t_{\alpha,m}}\boldBbar_{\beta,n}
    + [\boldB_{\alpha,m},\boldBbar_{\beta,n}]
$
 in two ways:
\[
 \begin{split}
    & \frac{\der \boldB_{\alpha,m}}{\der \bar t_{\beta,n}}
    - \frac{\der \boldBbar_{\beta,n}}{\der t_{\alpha,m}}
    + [\boldB_{\alpha,m},\boldBbar_{\beta,n}]
\\
    ={}&
    - \frac{\der}{\der \bar t_{\beta,n}}(\boldL^m U_\alpha)_{<0}
    - \frac{\der \boldBbar_{\beta,n}}{\der t_{\alpha,m}}
    + [\boldBbar_{\beta,n}, (\boldL^m U_\alpha)_{<0}],
\\
    ={}
    & \frac{\der \boldB_{\alpha,m}}{\der \bar t_{\beta,n}}
    + \frac{\der}{\der t_{\alpha,m}}(\boldLbar^n U_\beta)_{\geq0}
    - [\boldB_{\alpha,m}, (\boldLbar^n U_\beta)_{\geq0}].
 \end{split}
\]
 Here the second line is a difference operator with negative shifts,
 while the third line is a difference operator with non-negative
 shifts. Therefore all the expressions are zero, which implies
 \eqref{zs-BBbar}, \eqref{zs-LU<0Bbar} and \eqref{zs-BLUbar>0}.

\medskip
 (ii)
 Assume that the Zakharov-Shabat equations \eqref{zs-BB},
 \eqref{zs-BbarBbar} and \eqref{zs-BBbar} hold with the algebraic
 conditions \eqref{commutative:L,U,Q}, \eqref{U:alg-cond},
 \eqref{commutative:Lbar,Ubar,Qbar} and \eqref{Ubar:alg-cond}. 
 We will show that the Lax equations \eqref{lax:L,Lbar,U,Ubar,Q,Qbar} for
 $\boldL$, $\boldLbar$, $\boldU_\alpha$ and $\boldUbar_\alpha$ follow
 from these assumptions.

 Equation \eqref{zs-BB} implies
\[
    \frac{\der \boldB_{\beta,n}}{\der t_{\alpha,m}}
    -
    [\boldB_{\alpha,m}, \boldB_{\beta,n}]
    =
    \frac{\der \boldB_{\alpha,m}}{\der t_{\beta,n}}.
\]
 Hence,
\begin{equation}
 \begin{split}
    \frac{\der}{\der t_{\alpha,m}} \boldL^n \boldU_\beta
    -
    [\boldB_{\alpha,m}, \boldL^n \boldU_\beta]
    =
    \frac{\der \boldB_{\alpha,m}}{\der t_{\beta,n}}
    +
    \frac{\der}{\der t_{\alpha,m}} (\boldL^n \boldU_\beta)_{<0}
    -
    [\boldB_{\alpha,m}, (\boldL^n \boldU_\beta)_{<0}].
 \end{split}
\label{dLnUb/dtam-[Bam,LnUb]} 
\end{equation}
 The sum of the left-hand side of \eqref{dLnUb/dtam-[Bam,LnUb]} for
 $\beta=1,\dots,N$ is equal to 
\begin{equation}
    \frac{\der}{\der t_{\alpha,m}} \boldL^n
    -
    [\boldB_{\alpha,m}, \boldL^n]
    =
    \sum_{j=1}^n
    \boldL^{n-j}
    \left(
    \frac{\der}{\der t_{\alpha,m}}\boldL - [\boldB_{\alpha,m},\boldL]
    \right)
    \boldL^{j-1}
\label{dLn/dtam-[Bam,Lnn]}
\end{equation}
 because of the condition \eqref{U:alg-cond}. The right-hand side of
 \eqref{dLnUb/dtam-[Bam,LnUb]} shows that this expression is of order
 less than $m$ for any $n$, where the order of a difference operator
 $\sum_{j=0}^\infty a_j(\bolds)\, e^{(K-j)\der_s}$ is defined to be $K$
 if $a_0(\bolds)$ does not vanish.

 Suppose that
\[
    \frac{\der}{\der t_{\alpha,m}} \boldL - [\boldB_{\alpha,m}, \boldL]
    =
    \sum_{j=0}^\infty a_j(\bolds) e^{(K-j)\der_s}
\]
 does not vanish and is of order $K$ ($a_0(\bolds)\neq0$).
 Since $\boldL^{n-j}=e^{(n-j)\der_s}+\dotsb$ and
 $
    e^{(n-j)\der_s}\,a_0(\bolds)
    =
    a_0(\bolds+(n-j)\boldone)\, e^{(n-j)\der_s}
 $, 
 the right-hand side of \eqref{dLn/dtam-[Bam,Lnn]} is of the form
\[
    \left(
     \sum_{j=1}^n a_0(\bolds+(n-j)\boldone)
    \right) e^{(K+n-1)\der_s} +
    (\text{sum of }\tilde a_k(\bolds)\, e^{k\der_s},\ k<K+n-1),
\]
 whose order grows to infinity when $n\to\infty$, which contradicts the
 previous conclusion that the order of \eqref{dLnUb/dtam-[Bam,LnUb]}
 is bounded ($<m$). Therefore,
\[
    \frac{\der}{\der t_{\alpha,m}}\boldL - [\boldB_{\alpha,m},\boldL] 
    = 0,
\]
 which implies the Lax equation \eqref{lax:L,Lbar,U,Ubar,Q,Qbar} for
 $\boldL$ and $t_{\alpha,m}$. 

 From the equation which we have just proved, it follows that 
 the expression
 \eqref{dLn/dtam-[Bam,Lnn]} vanishes and the left-hand side of
 \eqref{dLnUb/dtam-[Bam,LnUb]} is equal to
\[
    \frac{\der}{\der t_{\alpha,m}} \boldL^n \boldU_\beta
    -
    [\boldB_{\alpha,m}, \boldL^n \boldU_\beta]
    =
    \boldL^n
    \left(\frac{\der}{\der t_{\alpha,m}} \boldU_\beta
    -
    [\boldB_{\alpha,m}, \boldU_\beta]\right).
\]
 Recall that this is of order less than $m$ for any $n$ because of
 \eqref{dLnUb/dtam-[Bam,LnUb]}. However, if
\[
    \frac{\der}{\der t_{\alpha,m}} \boldU_\beta
    -
    [\boldB_{\alpha,m}, \boldU_\beta]
    =
    \sum_{j=0}^\infty u_j(\bolds) e^{(K-j)\der_s}
\]
 does not vanish and is of order $K$ ($u_0(\bolds)\neq0$), the order of
\[
    \boldL^n
    \left(\frac{\der}{\der t_{\alpha,m}} \boldU_\beta
    -
    [\boldB_{\alpha,m}, \boldU_\beta]\right)
    =
    u_0(\bolds+n\boldone) e^{(n+K)\der_s} + \dotsb
\]
 grows to infinity as $n\to\infty$. Hence
\[
    \frac{\der}{\der t_{\alpha,m}} \boldU_\beta
    -
    [\boldB_{\alpha,m}, \boldU_\beta] = 0,
\]
 which is \eqref{lax:L,Lbar,U,Ubar,Q,Qbar} for $\boldU_{\beta}$ and
 $t_{\alpha,m}$. 

 Similarly, the equation \eqref{zs-BBbar} implies
\[
    \frac{\der}{\der \bar t_{\alpha,m}} \boldL^n \boldU_\beta
    -
    [\boldBbar_{\alpha,m}, \boldL^n \boldU_\beta]
    =
    \frac{\der \boldBbar_{\alpha,m}}{\der t_{\beta,n}}
    +
    \frac{\der (\boldL^n \boldU_\beta)_{<0}}{\der \bar t_{\alpha,m}}
    -
    [\boldBbar_{\alpha,m}, (\boldL^n \boldU_\beta)_{<0}],
\]
 which means that the left-hand side is of negative order for any
 $n$. The same argument as the argument for the Lax equations
 \eqref{lax:L,Lbar,U,Ubar,Q,Qbar} for $t_{\alpha,m}$ leads to
 the Lax equations for $\bar t_{\alpha,m}$. 

 The Lax equations for $\boldLbar$ and $\boldUbar_\beta$ can be proved
 in the way parallel to the above proof for $\boldL$ and $\boldU_\beta$.
\end{proof} 

The algebraic conditions for $\boldQ_\alpha(\bolds)$ and
$\boldQbar_\alpha(\bolds)$ in \eqref{commutative:L,U,Q} and
\eqref{commutative:Lbar,Ubar,Qbar} are equivalent to the following
conditions for $\boldP_\alpha(\bolds)$ and $\boldPbar_\alpha(\bolds)$:
\begin{gather}
    \boldP_\alpha(\bolds) \boldL(\bolds)
    =
    \boldL(\bolds+[1]_\alpha) \boldP_\alpha(\bolds),\qquad
    \boldP_\alpha(\bolds) \boldU_\beta(\bolds)
    =
    \boldU_\beta(\bolds+[1]_\alpha) \boldP_\alpha(\bolds),
\label{PL,PU}
\\
    \boldPbar_\alpha(\bolds) \boldLbar(\bolds)
    =
    \boldLbar(\bolds+[1]_\alpha) \boldPbar_\alpha(\bolds),\qquad
    \boldPbar_\alpha(\bolds) \boldUbar_\beta(\bolds)
    =
    \boldUbar_\beta(\bolds+[1]_\alpha) \boldPbar_\alpha(\bolds),
\label{PbarLbar,PbarUbar}
\\
 \begin{aligned}
    \boldP_\beta(\bolds+[1]_\alpha) \boldP_\alpha(\bolds)
    &=
    \boldP_\alpha(\bolds+[1]_\beta) \boldP_\beta(\bolds),
    \\
    \boldPbar_\beta(\bolds+[1]_\alpha) \boldPbar_\alpha(\bolds)
    &=
    \boldPbar_\alpha(\bolds+[1]_\beta) \boldPbar_\beta(\bolds).
 \end{aligned}
\label{PP}
\end{gather}

\begin{zam}
\label{rem:P=shift}
 The role of the operators $\boldP_\alpha(\bolds)$ and
 $\boldPbar_\alpha(\bolds)$ is the shift of the variable $s_\alpha$ by
 $1$. They are analogues of operators $P(s)$ in \cite{Tak02} and
 $P_{\alpha,\beta}(\bolds,{\boldsymbol t} ,\der)$ in \cite{Teo11}.
\end{zam}

\begin{zam}
\label{rem:meaning-ULQ-condition}
 We can interpret the condition \eqref{ULQ=UbarLbarQbar:basis} as
 follows. Let us rewrite it in terms of the evolution operators
 $\boldP_\alpha(\bolds)$ and $\boldPbar_\alpha(\bolds)$:
\begin{equation}
    \boldP_\a(\bolds)
    \sum_{\b =1}^N
    \boldU_\b(\bolds) \boldL^{\delta_{\a \b}}(\bolds)
    =
    \boldPbar_\a(\bolds)
    \sum_{\b =1}^N
    \boldUbar_\b (\bolds) \boldLbar^{-\delta_{\a \b }}(\bolds).
\label{ULP=UbarLbarPbar:basis}
\end{equation}
 This is formally equivalent to
\begin{equation}
    \boldP_\a^{-1}(\bolds)\boldPbar_\a(\bolds)
    =
    \sum_{\b =1}^N
    \boldU_\b(\bolds) \boldL^{\delta_{\a \b}}(\bolds)
    \left(
    \sum_{\b =1}^N
    \boldUbar_\b (\bolds) \boldLbar^{-\delta_{\a \b }}(\bolds)
    \right)^{-1}.
\label{ULUbarLr:rh-decomposition}
\end{equation}
 This means that $\boldP_\alpha$ and $\boldPbar_\alpha$ are obtained by
 the Bruhat or Riemann-Hilbert type decomposition of the right-hand
 side. (In general, both sides of \eqref{ULUbarLr:rh-decomposition} do not
 converge and we need to use the form \eqref{ULP=UbarLbarPbar:basis}.)

 This is a multiplicative analogue of the definitions \eqref{def:BBbar}
 of $\boldB_{\alpha,n}$ and $\boldBbar_{\alpha,n}$ which are obtained
 from the additive Bruhat or Riemann-Hilbert type decomposition of
 a product of $\boldL$ and $\boldU_\alpha$ or a product of $\boldLbar$
 and $\boldUbar_\alpha$.
\end{zam}

Generally, one can define operators $\boldP_\bolda(\bolds)$ and
$\boldPbar_\bolda(\bolds)$ for any $\bolda\in\Integer^N$ which add $\bolda$
to the $\bolds$-variable in the $\boldL$- and $\boldU_\alpha$-operators as
\begin{gather}
    \boldP_\bolda(\bolds) \boldL(\bolds)
    =
    \boldL(\bolds+\bolda) \boldP_\bolda(\bolds),\qquad
    \boldP_\bolda(\bolds) \boldU_\alpha(\bolds)
    =
    \boldU_\alpha(\bolds+\bolda) \boldP_\bolda(\bolds),
\label{PaL,PaU}
\\
    \boldPbar_\bolda(\bolds) \boldLbar(\bolds)
    =
    \boldLbar(\bolds+\bolda) \boldPbar_\bolda(\bolds),\qquad
    \boldPbar_\bolda(\bolds) \boldUbar_\alpha(\bolds)
    =
    \boldUbar_\alpha(\bolds+\bolda) \boldPbar_\bolda(\bolds),
\label{PbaraLbar,PbaraUbar}
\end{gather}
and have the composition rules
\begin{equation}
    \boldP_\boldb(\bolds+\bolda)\boldP_\bolda(\bolds)
    = \boldP_{\bolda+\boldb}(\bolds),\qquad
    \boldPbar_\boldb(\bolds+\bolda)\boldPbar_\bolda(\bolds)
    = \boldPbar_{\bolda+\boldb}(\bolds).
\label{PP=P}
\end{equation}
In fact, one has only to put
\begin{equation}
    \boldP_\bolda(\bolds):=
    e^{\sum_{\alpha=1}^N a_\alpha \der_{s_\alpha}}
    \prod_{\alpha=1}^N
    \boldQ_\alpha^{a_\alpha}(\bolds),
    \quad
    \boldPbar_\bolda(\bolds):=
    e^{\sum_{\alpha=1}^N a_\alpha \der_{s_\alpha}}
    \prod_{\alpha=1}^N
    \boldQbar_\alpha^{a_\alpha}(\bolds).
\label{Pa}
\end{equation}

\commentout{
Or, one can define $\boldP(\bolds,\bolds')$ inductively without using
$\boldQ$ as follows. Set
\begin{gather*}
    \boldP(\bolds+n[1]_\alpha,\bolds)
    :=
    \boldP_\alpha(\bolds+(n-1)[1]_\alpha) \dotsb
    \boldP_\alpha(\bolds+2[1]_\alpha)
    \boldP_\alpha(\bolds+[1]_\alpha)
    \boldP_\alpha(\bolds),
\\
    \boldP(\bolds-n[1]_\alpha,\bolds)
    :=
    \boldP(\bolds,\bolds-n[1]_\alpha)^{-1}
\end{gather*}
for each $\alpha=1,\dotsc,N$ and $n>0$. General
$\boldP(\bolds,\bolds')$'s are composed from such operators according to
the composition rule \eqref{PP=P}. The consistency of this definition is
guaranteed by the condition \eqref{PP}. The operators
$\boldPbar(\bolds,\bolds')$ are defined similarly.
}

\noindent
It is easy to check that the Lax equations
\eqref{lax:L,Lbar,U,Ubar,Q,Qbar} for $\boldQ_\alpha(\bolds)$  and
$\boldQbar_\alpha(\bolds)$ imply
\begin{align}
    \frac{\der\boldP_\bolda(\bolds)}{\der t_{\alpha,n}}
    &=
    \boldB_{\alpha,n}(\bolds+\bolda) \boldP_\bolda(\bolds)
    -
    \boldP_\bolda(\bolds) \boldB_{\alpha,n}(\bolds),
\label{P/t}
\\
    \frac{\der\boldP_\bolda(\bolds)}{\der \bar t_{\alpha,n}}
    &=
    \boldBbar_{\alpha,n}(\bolds+\bolda) \boldP_\bolda(\bolds)
    -
    \boldP_\bolda(\bolds) \boldBbar_{\alpha,n}(\bolds),
\label{P/tbar}
\\
    \frac{\der\boldPbar_\bolda(\bolds)}{\der t_{\alpha,n}}
    &=
    \boldB_{\alpha,n}(\bolds+\bolda) \boldPbar_\bolda(\bolds)
    -
    \boldPbar_\bolda(\bolds) \boldB_{\alpha,n}(\bolds),
\label{Pbar/t}
\\
    \frac{\der\boldPbar_\bolda(\bolds)}{\der \bar t_{\alpha,n}}
    &=
    \boldBbar_{\alpha,n}(\bolds+\bolda) \boldPbar_\bolda(\bolds)
    -
    \boldPbar_\bolda(\bolds) \boldBbar_{\alpha,n}(\bolds).
\label{Pbar/tbar}
\end{align}

The following proposition provides some information\footnote{
In \cite{UT84} it was assumed from the beginning that $\bar b_0(\bolds)$
and $\bar u_{\alpha,0}(\bolds)$ had the form \eqref{bars}.
} 
about
the leading coefficients of the operators $\bar \boldL$,
$\bar \boldU_{\alpha}$, $\bar \boldP_{\alpha}$.

\begin{predl}
\label{proposition:tildew0}
There exists an invertible matrix-valued function 
$\tilde w_0=\tilde w_0 (\bolds ,
\boldt , \bar \boldt )$ such that
\beq\label{bars}
\begin{array}{l}
\bar b_0(\bolds )=\tilde w_0 (\bolds )\tilde w_0^{-1}(\bolds -\boldone ),
\\ \\
\bar u_{\alpha ,0}(\bolds )=\tilde w_0 (\bolds )E_{\alpha}
\tilde w_0^{-1} (\bolds ),
\\ \\
\bar p_{\alpha ,0}(\bolds )=\tilde w_0 (\bolds +[1]_{\alpha})
\tilde w_0^{-1} (\bolds ).
\end{array}
\eeq
\end{predl}

\begin{proof}
The existence of such $\tilde w_0$ follows from 
the algebraic conditions 
(\ref{commutative:L,U,Q})--(\ref{Qbar:alg-cond}) and their
corollaries (\ref{PbarLbar,PbarUbar}), (\ref{PP}).
From (\ref{PP}) we have, comparing the leading coefficients
in the second equation:
\beq\label{cond1}
\bar p_{\alpha ,0}(\bolds +[1]_{\beta})\bar p_{\beta ,0}(\bolds )=
\bar p_{\beta ,0}(\bolds +[1]_{\alpha})\bar p_{\alpha ,0}(\bolds ).
\eeq
It follows from this condition that 
there exists a matrix-valued function $\tilde w_{0}({\bf s})$ such that
\beq\label{barp}
\bar p_{\alpha ,0}(\bolds )=\tilde w_{0}(\bolds +[1]_{\alpha})
\tilde w_0^{-1}(\bolds )
\eeq
(equation (\ref{cond1}) is the compatibility condition for recursive 
relations $\tilde w_{0}({\bf s}+[1]_{\alpha})=\bar p_{\alpha ,0}(\bolds )
\tilde w_0(\bolds )$ which allow one to find values of the function
$\tilde w_{0}({\bf s})$ in all points of the $\bolds$-lattice
starting from some initial value
$\tilde w_{0}(0)$).
Plugging this into the second equation of (\ref{PbarLbar,PbarUbar}) and 
equating the leading coefficients, we have:
$$
\tilde w_0^{-1}(\bolds )\bar u_{\beta ,0}(\bolds )\tilde w_0(\bolds )=
\tilde w_0^{-1}(\bolds +[1]_{\alpha})\bar u_{\beta ,0}(\bolds +[1]_{\alpha})
\tilde w_0(\bolds +[1]_{\alpha})
$$
for all $\alpha$, whence
\beq\label{baru}
\bar u_{\beta ,0}(\bolds )=\tilde w_0(\bolds )\bar u_{\beta ,0}
\tilde w_0^{-1}(\bolds ),
\eeq
where $\bar u_{\beta ,0}$ is an $\bolds $-independent matrix.
From commutativity of the operators $\boldUbar_{\alpha}$ it follows that
the matrices $\bar u_{\beta ,0}$ for $\beta =1, \ldots , N$ 
commute and can be simultaneously diagonalized.
The algebraic conditions (\ref{Ubar:alg-cond}) imply that
$\bar u_{\beta ,0}\bar u_{\alpha ,0}=\delta_{\alpha \beta}\bar 
u_{\beta ,0}$,
$\displaystyle{\sum_{\alpha =1}^N \bar u_{\alpha ,0}=1_N}$, 
hence the matrix $\bar u_{\beta ,0}$ has $N-1$ zero eigenvalues 
and one eigenvalue equal to 1. 
Therefore,
if we choose the order of eigenvectors of $\bar u_{\beta,0}$'s
appropriately,
$\bar u_{\beta ,0}=\bar vE_{\beta}\bar v^{-1}$, where
$\bar v$ is a non-degenerate $\bolds$-independent 
matrix. The redefinition $\tilde w_0 \to \tilde w_0 \bar v^{-1}$
does not spoil the relation (\ref{barp}) and allows one to put
$\bar u_{\beta ,0}=E_{\beta}$ in (\ref{baru}) without loss
of generality.
At last, from (\ref{Qbar:alg-cond}) 
we conclude that
\beq\label{barb}
\bar b_0(\bolds )=\tilde w_0(\bolds )\tilde w_0^{-1}(\bolds -{\bf 1}).
\eeq
Thus all the relations (\ref{bars}) are proved. 
\end{proof}

\begin{zam}
\label{remark:freedom}
There is a freedom in the choice of the matrix $\tilde w_0(\bolds )$:
it can be multiplied from the right by
an arbitrary invertible
diagonal matrix which may depend on 
${\boldsymbol t}, \bar \boldt $ but
not on $\bolds $. This fact will be essential below in
section \ref{subsection:gauge}.
\end{zam}

\subsection{Gauge transformations}
\label{subsection:gauge}

This subsection is an extended remark on the choice of the 
first coefficients in the operators $\boldL$, $\boldU_{\alpha}$,
$\boldP_{\alpha}$ in (\ref{def:L,Lbar})--(\ref{def:Q,Qbar})
($b_0(\bolds )=p_{\alpha ,0}(\bolds )=1_N$, $u_{\alpha ,0}(\bolds )=
E_{\alpha}$) and on the resulting 
asymmetry in the definitions of the 
operators $\boldL$, $\boldU_{\alpha}$,
$\boldP_{\alpha}$ and their bar-counterparts.

Let the coefficients $b_0(\bolds )$, $u_{\alpha ,0}(\bolds )$,
$p_{\alpha ,0}({\bf s})$ of the operators $\boldL$, $\boldU_{\alpha}$,
$\boldP_{\alpha}$ in (\ref{def:L,Lbar})--(\ref{def:Q,Qbar})
be some matrix-valued functions of $\bolds $, $\boldt $,
$\boldtbar$. Suppose that $b_0(\bolds )$ and 
$p_{\alpha ,0}({\bf s})$ are invertible matrices.
Then the Lax equations (\ref{lax:L,Lbar,U,Ubar,Q,Qbar}) imply that 
they do not depend on ${\boldsymbol t}, \boldtbar $. Indeed, from
$$
\begin{array}{l}
\displaystyle{
\frac{\p \boldL}{\p t_{\alpha ,k}}=[(\boldL^k \boldU_{\alpha})_{\geq 0},
\boldL ]=-[(\boldL^k \boldU_{\alpha})_{<0},
\boldL ]},
\\ \\
\displaystyle{
\frac{\p \boldL}{\p \bar 
t_{\alpha ,k}}=[(\boldLbar^k \boldUbar_{\alpha})_{<0},
\boldL ]},
\end{array}
$$
we deduce, comparing the coefficients in front of $e^{\p_s}$, that
$$
\frac{\p b_0(\bolds)}{\p t_{\alpha ,k}}=
\frac{\p b_0(\bolds )}{\p \bar t_{\alpha ,k}}=0 \quad
\mbox{for all $\alpha , k$}.
$$
The argument for
$u_{\alpha ,0}(\bolds )$,
$p_{\alpha ,0}(\bolds )$ is similar. 
To find the dependence of these functions on $\bolds $,
we almost repeat the proof of \propref{proposition:tildew0}, using the
discrete Lax equations (\ref{PL,PU}), (\ref{PP}). 
From (\ref{PP}) we have, comparing the leading coefficients
in the first equation:
\beq\label{cond}
p_{\alpha ,0}(\bolds +[1]_{\beta})p_{\beta ,0}(\bolds )=
p_{\beta ,0}(\bolds +[1]_{\alpha})p_{\alpha ,0}(\bolds ).
\eeq
Similarly to the proof of Proposition \ref{proposition:tildew0},
it follows from this condition that 
there exists a matrix-valued function $w_{0}({\bf s})$ such that
$$
p_{\alpha ,0}(\bolds )=w_{0}(\bolds +[1]_{\alpha})w_0^{-1}(\bolds ).
$$
Plugging this into (\ref{PL,PU}) and equating the leading
coefficients, we have:
$$
w_0^{-1}(\bolds )u_{\beta ,0}(\bolds )w_0(\bolds )=
w_0^{-1}(\bolds +[1]_{\alpha})u_{\beta ,0}(\bolds +[1]_{\alpha})
w_0(\bolds +[1]_{\alpha})
$$
for all $\alpha$, whence
$$
u_{\beta ,0}(\bolds )=w_0(\bolds )u_{\beta ,0}w_0^{-1}(\bolds ),
$$
where $u_{\beta ,0}$ is a constant ($\bolds $-independent) matrix.
The algebraic conditions (\ref{U:alg-cond}) imply that
$u_{\beta ,0}u_{\alpha ,0}=\delta_{\alpha \beta}u_{\beta ,0}$,
$\displaystyle{\sum_{\alpha =1}^N u_{\alpha ,0}=1_N}$, 
whence $u_{\beta ,0}=vE_{\beta}v^{-1}$, where
$v$ is a constant non-degenerate matrix. At last, from (\ref{Q:alg-cond}) 
we conclude
that
$$
b_0(\bolds )=w_0(\bolds )w_0^{-1}(\bolds +{\bf 1}).
$$
Therefore, re-defining $w_0(\bolds )\to w_0(\bolds )v$,
we see that the ``gauge transformation''
$\boldA \to w_0^{-1}(\bolds )\boldA w_0(\bolds )$, where
$\boldA $ is any one of the operators $\boldL, \boldU_{\alpha},
\boldQ_{\alpha}$, $\boldLbar, \boldUbar_{\alpha},
\boldQbar_{\alpha}$, allows us, without loss of generality, to fix the coefficients
$b_0(\bolds ), u_{\alpha ,0}(\bolds ), p_{\alpha ,0}(\bolds )$
to be
$$b_0(\bolds )=p_{\alpha ,0}(\bolds )=1_N, \quad
u_{\alpha ,0}(\bolds )=
E_{\alpha}, 
$$ as in 
(\ref{def:L,Lbar})--(\ref{def:Q,Qbar}). Since 
$w_0(\bolds )$ does not depend on ${\boldsymbol t} $, this transformation preserves
the form of the Lax and Zakharov-Shabat equations. 

One can also consider more general gauge transformations
$\boldA \mapsto \boldA^{(G)}=G^{-1}\boldA G$, where the matrix
$G$ may depend also on ${\boldsymbol t} $ and
$\boldtbar$: $G=G(\bolds , {\boldsymbol t}, \boldtbar )$. 
In particular, the transformed Lax operators become
\beq\label{gauge0}
\begin{array}{l}
\boldL^{(G)}(\bolds )=G^{-1}(\bolds )G(\bolds +\boldone )e^{\p_s}
+\ldots \, ,
\\ \\
\boldLbar^{(G)}(\bolds )=G^{-1}(\bolds ) \tilde w_0(\bolds )
\tilde w_0^{-1} (\bolds -\boldone )G(\bolds -\boldone )e^{-\p_s}
+\ldots \, ,
\end{array}
\eeq
where $\tilde w_0(\bolds )$ is introduced by Proposition
\ref{proposition:tildew0}.
In order 
to preserve the form of the Lax and Zakharov-Shabat equations, the 
difference operators $\boldB_{\alpha , k}$, 
$\boldBbar_{\alpha , k}$ should transform as follows:
\beq\label{gauge}
\begin{array}{l}
\boldB_{\alpha , k}\mapsto \boldB_{\alpha , k}^{(G)}=
G^{-1}\boldB_{\alpha , k}G -G^{-1}\p_{t_{\alpha , k}}G
=((\boldL^{(G)})^k \boldU_{\alpha}^{(G)})_{\geq 0}
-G^{-1}\p_{t_{\alpha , k}}G,
\\ \\
\boldBbar_{\alpha , k}\mapsto \boldBbar_{\alpha , k}^{(G)}=
G^{-1}\boldBbar_{\alpha , k}G -G^{-1}\p_{\bar t_{\alpha , k}}G
=((\boldLbar^{(G)})^k \boldUbar_{\alpha}^{(G)})_{<0}
-G^{-1}\p_{\bar t_{\alpha , k}}G,
\end{array}
\eeq
so that
$$
\frac{\p}{\p t_{\alpha ,k}}-\boldB_{\alpha , k}^{(G)}=
G^{-1} \left (\frac{\p}{\p t_{\alpha ,k}}-\boldB_{\alpha , k}
\right )G,
$$
$$
\frac{\p}{\p \bar t_{\alpha ,k}}-\boldBbar_{\alpha , k}^{(G)}=
G^{-1} \left (\frac{\p}{\p \bar t_{\alpha ,k}}-\boldBbar_{\alpha , k}
\right )G.
$$
In particular, one may take $G=\tilde w_0(\bolds )$.
Let us denote the operators and their coefficients
transformed in this way by adding 
prime: $\boldL \mapsto \boldL '=\tilde w_0^{-1}\boldL \tilde w_0$, etc. 
Then the leading coefficients
of the operators $\boldL', \boldLbar'$, $\boldU_{\alpha}', 
\boldUbar_{\alpha}'$, $\boldP_{\alpha}', 
\boldPbar_{\alpha}'$ are, respectively:
$$
\begin{array}{l}
b_0'(\bolds )
=\tilde w_0^{-1}(\bolds )\tilde w_0(\bolds +{\bf 1}), \quad
\bar b_0'(\bolds )=1_N,
\\ \\
u_{\alpha ,0}'(\bolds )=
\tilde w_0^{-1}(\bolds )E_{\alpha}\tilde w_0(\bolds ), \quad
\bar u_{\alpha ,0}'=E_{\alpha}, 
\\ \\
p_{\alpha ,0}'(\bolds )=\tilde w_0^{-1}(\bolds +[1]_{\alpha})
\tilde w_0(\bolds ), \quad 
\bar p_{\alpha ,0}'(\bolds )=1_N.
\end{array} 
$$
Comparing the leading coefficients in the Lax type equation
$$
\frac{\p \boldPbar_{\alpha}'(\bolds )}{\p t_{\beta ,k}}=
\boldB'_{\beta ,k}(\bolds +[1]_{\alpha})\boldPbar_{\alpha}'(\bolds )
-\boldPbar_{\alpha}'(\bolds )
\boldB'_{\beta ,k}(\bolds ),
$$
one finds that the left-hand side is 0 while the right-hand side 
yields
$$
(\boldB'_{\beta ,k}(\bolds +[1]_{\alpha}))_0=
(\boldB'_{\beta ,k}(\bolds ))_0,
$$
i.e., $(\boldB'_{\beta ,k}(\bolds ))_0$ does not depend on $\bolds $:
$(\boldB'_{\beta ,k}(\bolds ))_0 =
{\bf c}_{\beta ,k}({\boldsymbol t} , \bar {\boldsymbol t} )$
(by $({\bf A})_0$ we denote the coefficient in front of $e^{0\p_s}$
in the difference operator ${\bf A}$). 
The transformation laws 
(\ref{gauge}) give in our case:
\beq\label{gauge1}
\begin{array}{l}
\boldB_{\beta ,k}'(\bolds )=((\boldL')^k\boldU_{\beta}')_{\geq 0}-
\tilde w_0^{-1}(\bolds )\p_{t_{\beta ,k}}\tilde w_0(\bolds ),
\\ \\
\boldBbar_{\beta ,k}'(\bolds )=
((\boldLbar')^k\boldUbar_{\beta}')_{<0}-
\tilde w_0^{-1}(\bolds )\p_{\bar t_{\beta ,k}}\tilde w_0(\bolds ).
\end{array}
\eeq
Therefore, taking the $(\,\, )_0$-part of the first equation,
we have:
\beq\label{gauge2}
((\boldL')^k\boldU_{\beta}')_{0}=
\tilde w_0^{-1}(\bolds )\p_{t_{\beta ,k}}\tilde w_0(\bolds )+
{\bf c}_{\beta ,k}({\boldsymbol t} , \bar {\boldsymbol t} )
\eeq
which means that
$$
\boldB_{\beta ,k}'(\bolds )=((\boldL')^k\boldU_{\beta}')_{>0}+
{\bf c}_{\beta ,k}({\boldsymbol t} , \bar {\boldsymbol t} ).
$$
Similar manipulations with the equation
$$
\frac{\p \boldP_{\alpha}'(\bolds )}{\p \bar t_{\beta ,k}}=
\boldBbar'_{\beta ,k}(\bolds +[1]_{\alpha})\boldP_{\alpha}'(\bolds )
-\boldP_{\alpha}'(\bolds )
\boldBbar'_{\beta ,k}(\bolds )
$$
lead to 
\beq\label{gauge2a}
((\boldLbar')^k\boldUbar_{\beta}')_{0}=
-\tilde w_0^{-1}(\bolds )\p_{t_{\beta ,k}}\tilde w_0(\bolds )+
\tilde w_0^{-1}(\bolds ) \bar {\bf c}_{\beta ,k}
({\boldsymbol t} , \bar {\boldsymbol t} )\tilde w_0(\bolds )
\eeq
which means that
$$
\boldBbar_{\beta ,k}'(\bolds )=
((\boldLbar')^k\boldUbar_{\beta}')_{\leq 0}-
\tilde w_0^{-1}(\bolds ) \bar {\bf c}_{\beta ,k}
({\boldsymbol t} , \bar {\boldsymbol t} )\tilde w_0(\bolds )
$$
with some $\bolds $-independent matrix-valued function 
$\bar {\bf c}_{\beta ,k}$.

Some additional arguments allow one to put ${\bf c}_{\beta ,k}
({\boldsymbol t} , \bar {\boldsymbol t} )=
\bar {\bf c}_{\beta ,k}({\boldsymbol t} , \bar {\boldsymbol t} )=0$ without any loss of generality. 
Indeed, taking the $(\,\, )_0$-part of 
the equation 
$$
\frac{\p \boldUbar_{\alpha}}{\p t_{\beta ,k}}=[\boldB_{\beta ,k},
\boldUbar_{\alpha}],
$$
we deduce that
$$
[\tilde w_0^{-1}(\boldB_{\beta ,k})_0 \tilde w_0-
\tilde w_0^{-1}\p_{t_{\beta ,k}}\tilde w_0, \, E_{\alpha}]=0.
$$
Together with (\ref{gauge2}) this means that
\beq\label{gauge3}
[{\bf c}_{\beta ,k}({\boldsymbol t} , \bar {\boldsymbol t} ), 
E_{\alpha}]=0 \quad \mbox{for all 
$\alpha$},
\eeq
whence ${\bf c}_{\beta ,k}({\boldsymbol t} , \bar {\boldsymbol t} )$ is
a diagonal matrix. 
Let us rewrite (\ref{gauge2}) in the form
\beq\label{gauge5}
(\boldB_{\beta ,k})_0=
(\boldL^k\boldU_{\beta})_{0}=
\p_{t_{\beta ,k}}\tilde w_0(\bolds )\cdot
\tilde w_0^{-1}(\bolds )+
\tilde w_0(\bolds ){\bf c}_{\beta ,k}
({\boldsymbol t} , \bar {\boldsymbol t} )\tilde w_0^{-1}(\bolds ).
\eeq
Note that the second term 
in the right-hand side  
can be eliminated
by multiplying $\tilde w_0(\bolds )$ by an $\bolds $-independent
diagonal matrix
$
    \exp\left(
     \int^{t_{\beta,k}} 
      {\bf c}_{\beta,k}(\boldt,\boldtbar)\,
     dt_{\beta,k}
    \right)
$
from the right. But this is just the freedom 
left in the definition of the matrix $\tilde w_0(\bolds )$ by
the formulae 
(\ref{def:L,Lbar})--(\ref{def:Q,Qbar}). Therefore, 
we can put
${\bf c}_{\beta ,k}({\boldsymbol t} , \bar {\boldsymbol t} )=0$ without loss
of generality. 
Let us rewrite equation (\ref{gauge2a}) in the form
\beq\label{gauge4}
(\boldBbar_{\beta ,k})_0=((\boldLbar)^k\boldUbar_{\beta})_{0}=
-\p_{\bar t_{\beta ,k}}\tilde w_0(\bolds )\cdot \tilde w_0^{-1}(\bolds )+
\bar {\bf c}_{\beta ,k}({\boldsymbol t} , \bar {\boldsymbol t} )
\eeq
and prove that the choice ${\bf c}_{\beta ,k}({\boldsymbol t} , 
\bar {\boldsymbol t} )=0$ implies
that $\bar {\bf c}_{\beta ,k}({\boldsymbol t} , \bar {\boldsymbol t} )=0$. This fact follows from
the Zakharov-Shabat equation
(\ref{zs-BLUbar>0}). Its $(\,\, )_0$-part is
$$
\left [\p_{t_{\alpha,m}} \!- (\boldB_{\alpha,m})_0,\,
     \p_{\bar t_{\beta ,k}}\!
      + (\boldLbar^k \boldUbar_\beta )_{0}
    \right]=0
    $$
The general solution is
$$
(\boldB_{\alpha,m})_0=\p_{t_{\alpha ,m}}W \cdot W^{-1}, \quad
(\boldLbar^k \boldUbar_\beta )_{0}=
-\p_{\bar t_{\beta ,k}}W \cdot W^{-1},
$$
where $W$ is an arbitrary matrix-valued 
function\footnote{We use the fact
that the general solution of the matrix equation
$[\p_x -B_x, \p_y -B_y]=0$ is $B_x =\p_x W\cdot W^{-1}$,
$B_y =\p_y W\cdot W^{-1}$}. According to (\ref{gauge5}) with
${\bf c}_{\beta ,k}=0$, we should put 
$W=\tilde w_0(\bolds )$, which means
(after comparison with (\ref{gauge4})) that
$\bar {\bf c}_{\beta ,k}({\boldsymbol t} , \bar {\boldsymbol t} )=0$.

Finally, we have:
\beq\label{gauge6}
\boldB'_{\beta ,k}=((\boldL')^k \boldU_{\beta}')_{>0}, \quad
\boldBbar'_{\beta ,k}=((\boldLbar')^k \boldUbar_{\beta}')_{\leq 0}.
\eeq
Therefore, we see that in this gauge the roles of the two Lax
operators are exchanged. This restores the symmetry in their definitions,
which is broken in (\ref{def:L,Lbar})--(\ref{def:Q,Qbar}).

\begin{zam}
\label{rem:gauge}
The gauge transformations of the one-component Toda lattice were
discussed in \cite{tak90}.
They can be represent 
in the form
\begin{equation*}
 \begin{aligned}
    \boldL     &\mapsto \boldL^{(G)} := G^{-1} \boldL G, &
    \boldLbar  &\mapsto \boldLbar^{(G)} := G^{-1} \boldLbar G,
\\
    \boldB_n     &\mapsto \boldB_n^{(G)}
    := g^{-1} \boldB_n G - G^{-1} \frac{\der G}{\der t_n}, &
    \boldBbar_n &\mapsto \boldBbar_n^{(G)}
    := G^{-1} \boldBbar_n G - G^{-1} \frac{\der G}{\der \bar t_n},
\\
    \text{or\ }
    G^{-1} &\left(\frac{\der}{\der t_n} - \boldB_n\right) G
    = \frac{\der}{\der t_n} - \boldB^{(G)}_n, &
    G^{-1} &\left(\frac{\der}{\der \bar t_n} - \boldBbar_n\right) G
    = \frac{\der}{\der \bar t_n} - \boldBbar^{(G)}_n.
 \end{aligned}
\end{equation*}
The choice $G=\tilde w_0^{1/2}(s)$ corresponds to the {\it symmetric
gauge}, in which the leading coefficients of the two Lax operators
$\boldL^{({\rm sym})}$, $\boldLbar^{({\rm sym})}$ become equal 
up to a shift of $s$ by 1:
$$
\begin{array}{l}
\boldL^{({\rm sym})}=b_0^{(s)}(s)e^{\der_s}
+\ldots \, =
\tilde w_0^{-1/2}(s) \tilde w_0^{1/2}(s+1)e^{\der_s}
+\ldots \, ,
\\ \\
\boldLbar^{({\rm sym})}=b_0^{(s)}(s-1)e^{-\der_s}+\ldots \, =
\tilde w_0^{-1/2}(s-1) \tilde w_0^{1/2}(s)e^{-\der_s}
+\ldots 
\end{array}
$$
(here $\tilde w_0(s)$ is a scalar function, so we can
change the order in the product). The $\boldB$-operators 
in the symmetric gauge look as 
follows:
\beq\label{gauge7}
\begin{array}{l}
\boldB_k^{({\rm sym})} =((\boldL^{({\rm sym})})^k)_{>0}+
\frac{1}{2}((\boldL^{({\rm sym})})^k)_{0},
\\ \\
\boldBbar_k^{({\rm sym})} =((\boldLbar^{({\rm sym})})^k)_{<0}+
\frac{1}{2}((\boldLbar^{({\rm sym})})^k)_{0}.
\end{array}
\eeq
In the gauge with $G=\tilde w_0(s)$ the roles of $\boldL$ and 
$\boldLbar$ are exchanged.
See \cite{tak90} for details.
\end{zam}
 
An analogue of the symmetric gauge exists 
in the multi-component case, too.
Taking $G=\tilde
 w_0^{1/2}(\bolds)$, we have:
\begin{equation}
 \begin{aligned}
    \boldL^{({\rm sym})}(\bolds ) &
    = \tilde w_0^{-1/2}(\bolds) \tilde w_0^{1/2}(\bolds+\boldone)
      e^{\der_s} + \dotsc,
\\
    \boldLbar^{({\rm sym})}(\bolds ) &
    = \tilde w_0^{1/2}(\bolds) \tilde w_0^{-1/2}(\bolds-\boldone)
      e^{-\der_s} + \dotsc.
 \end{aligned}
\label{sym-gauge}
\end{equation}
Because of the non-commutativity of matrices, the leading
coefficients are not related in the same simple way as in the 
one-component case.
The $\boldB^{({\rm sym})}$- and 
$\boldBbar^{({\rm sym})}$-operators in the symmetric gauge 
are expressed as follows:
\beq\label{gauge8}
\begin{array}{l}
\boldB^{({\rm sym})}_{\beta ,k}=((\boldL^{({\rm sym})})^k
\boldU_{\beta}^{({\rm sym})})_{>0}+
\p_{t_{\beta ,k}}\tilde w_0^{1/2}\cdot \tilde w_0^{-1/2},
\\ \\
\boldBbar^{({\rm sym})}_{\beta ,k}=((\boldLbar^{({\rm sym})})^k
\boldUbar_{\beta}^{({\rm sym})})_{<0}-\tilde w_0^{-1/2}\cdot
\p_{\bar t_{\beta ,k}}\tilde w_0^{1/2}.
\end{array}
\eeq
However, since, for example, 
$
    \p_{t_{\beta ,k}}\tilde w_0^{1/2}
    \neq
    \frac{1}{2}\, \tilde w_0^{-1/2} 
     \p_{t_{\beta ,k}}\bar w_0,
$
 the $\boldB^{({\rm sym})}$- and 
 $\boldBbar^{({\rm sym})}$-operators do not have such a simple
 description as in the one-component case.

\section{Linearization: wave operators and wave functions}
\label{sec:linear}

\subsection{Wave operators}

Here we derive the linear problems associated with the multi-component
Toda lattice hierarchy defined in Section \ref{subsection:LZS}. The first
half of the derivation is almost the same as the one given 
by Ueno and Takasaki,
but some details of the proof of the existence of the wave matrix
functions are
added, correcting inaccurate arguments in \S3.2 of \cite{UT84}.

\begin{predl}
\label{predl:wave-mat}
 (i)
 Let $(\boldL, \boldLbar, \boldU_\alpha, \boldUbar_\alpha,
 \boldQ_\alpha,\boldQbar_\alpha)_{\alpha=1,\dotsc,N}$ be a
 solution of the $N$-component Toda lattice hierarchy.
 Then there exist  
 operators $\boldW(\bolds,{\boldsymbol t} ,\bar {\boldsymbol t})$ and
 $\boldWbar(\bolds,{\boldsymbol t} ,\bar {\boldsymbol t})$ of the form
\begin{equation}
 \begin{aligned}
    \boldW(\bolds,{\boldsymbol t} ,\bar {\boldsymbol t})
    &=
    \hat\boldW(\bolds,{\boldsymbol t} ,\bar {\boldsymbol t})\,
    \diag_\alpha (e^{\xi({\boldsymbol t} _\alpha,e^{\der_s})}),
\\
    \boldWbar(\bolds,{\boldsymbol t} ,\bar {\boldsymbol t})
    &=
    \hat\boldWbar(\bolds,{\boldsymbol t} ,\bar {\boldsymbol t})\,
    \diag_\alpha (e^{\xi(\bar {\boldsymbol t}_\alpha,e^{-\der_s})}),
\\
    \hat\boldW(\bolds,{\boldsymbol t} ,\bar {\boldsymbol t})
    &= \sum_{j=0}^\infty
    w_j(\bolds,{\boldsymbol t} ,\bar {\boldsymbol t})\, e^{-j\der_s},
    \qquad w_0(\bolds,{\boldsymbol t} ,\bar {\boldsymbol t}) = 1_N,
\\
    \hat\boldWbar(\bolds,{\boldsymbol t} ,\bar {\boldsymbol t})
    &= \sum_{j=0}^\infty
    \bar w_j(\bolds,{\boldsymbol t} ,\bar {\boldsymbol t})\, e^{j\der_s},
    \qquad \bar w_0(\bolds,{\boldsymbol t} ,\bar {\boldsymbol t}) 
    \in GL(N,\Comp),
 \end{aligned}
\label{wave-mat}
\end{equation}
 where
\begin{equation}
 \begin{gathered}
    \xi({\boldsymbol t}_\alpha,e^{\der_s})
    =
    \sum_{n=1}^\infty t_{\alpha,n}e^{n\der_s},
    \qquad
    \diag_\alpha (a_\alpha)
    =
    \begin{pmatrix}
    a_1 &        & \\
        & \ddots & \\
        &        & a_N
    \end{pmatrix},
\\
    w_j(\bolds,{\boldsymbol t} ,\bar {\boldsymbol t})
    =
    (w_{j,\alpha\beta}(\bolds,{\boldsymbol t} ,\bar {\boldsymbol t})
    )_{\alpha,\beta=1,\dotsc,N}
    \in \Mat(N\times N,\Comp),
\\
    \bar w_j(\bolds,{\boldsymbol t} ,\bar {\boldsymbol t})
    =
    (\bar w_{j,\alpha\beta}(\bolds,{\boldsymbol t} ,\bar {\boldsymbol t})
    )_{\alpha,\beta=1,\dotsc,N}
    \in \Mat(N\times N,\Comp),
 \end{gathered}
\label{wave-mat:notations}
\end{equation}
 satisfying the following linear equations:
\begin{align}
    \boldL(\bolds) \boldW(\bolds) &= \boldW(\bolds) e^{\der_s}, &
    \boldU_\alpha(\bolds) \boldW(\bolds) &= \boldW(\bolds) E_\alpha,
\label{lin-eq:LW,UW}
\\
    \boldLbar(\bolds) \boldWbar(\bolds)
    &= \boldWbar(\bolds) e^{-\der_s},
    &
    \boldUbar_\alpha(\bolds) \boldWbar(\bolds)
    &= \boldWbar(\bolds) E_\alpha,
\label{lin-eq:LbarWbar,UbarWbar}
\\
    \frac{\der \boldW(\bolds)}{\der t_{\alpha,n}}
    &= \boldB_{\alpha,n}(\bolds) \boldW(\bolds), &
    \frac{\der \boldW(\bolds)}{\der \bar t_{\alpha,n}}
    &= \boldBbar_{\alpha,n}(\bolds) \boldW(\bolds),
\label{lin-eq:dW}
\\
    \frac{\der \boldWbar(\bolds)}{\der t_{\alpha,n}}
    &= \boldB_{\alpha,n}(\bolds) \boldWbar(\bolds), &
    \frac{\der \boldWbar(\bolds)}{\der \bar t_{\alpha,n}}
    &= \boldBbar_{\alpha,n}(\bolds) \boldWbar(\bolds),
\label{lin-eq:dWbar}
\end{align}
 and
\begin{equation}
    \boldQ_\alpha(\bolds) \boldW(\bolds)
    = \boldW(\bolds) e^{-\der_{s_\alpha}},
    \qquad
    \boldQbar_\alpha(\bolds) \boldWbar(\bolds)
    = \boldWbar(\bolds) e^{-\der_{s_\alpha}}.
\label{lin-eq:QW}
\end{equation}
 In terms of $\hat\boldW(\bolds)$ and $\hat\boldWbar(\bolds)$, equations
 \eqref{lin-eq:LW,UW} and \eqref{lin-eq:LbarWbar,UbarWbar} are
 equivalent to the following equations:
\begin{align}
    \boldL(\bolds) \hat\boldW(\bolds)
    &= \hat\boldW(\bolds) e^{\der_s}, &
    \boldU_\alpha(\bolds) \hat\boldW(\bolds)
    &= \hat\boldW(\bolds) E_\alpha,
\label{lin-eq:LWhat,UWhat}
\\
    \boldLbar(\bolds) \hat\boldWbar(\bolds)
    &= \hat\boldWbar(\bolds) e^{-\der_s}, &
    \boldUbar_\alpha(\bolds) \hat\boldWbar(\bolds)
    &= \hat\boldWbar(\bolds) E_\alpha ,
\label{lin-eq:LbarWbarhat,UbarWbarhat}
\end{align}
 and equations \eqref{lin-eq:dW}, \eqref{lin-eq:dWbar} are equivalent
 to
\begin{gather}
    \frac{\der \hat\boldW(\bolds)}{\der t_{\alpha,n}}
    = \boldB_{\alpha,n}(\bolds) \hat\boldW(\bolds)
     - \hat\boldW(\bolds) e^{n\der_s}E_\alpha,
    \quad
    \frac{\der \hat\boldW(\bolds)}{\der \bar t_{\alpha,n}}
    = \boldBbar_{\alpha,n}(\bolds) \hat\boldW(\bolds),
\label{lin-eq:dWhat}
\\
    \frac{\der \hat\boldWbar(\bolds)}{\der t_{\alpha,n}}
    = \boldB_{\alpha,n} \hat\boldWbar(\bolds), 
    \quad
    \frac{\der \hat\boldWbar(\bolds)}{\der \bar t_{\alpha,n}}
    = \boldBbar_{\alpha,n}(\bolds) \hat\boldWbar(\bolds)
    - \hat\boldWbar(\bolds) e^{-n\der_s} E_\alpha.
\label{lin-eq:dWbarhat}
\end{gather}
 Relations \eqref{lin-eq:QW} in terms of 
 $\hat\boldW$ and $\hat\boldWbar$
 are rewritten as 
\begin{equation}
    \boldQ_\alpha(\bolds) \hat\boldW(\bolds)
    = \hat\boldW(\bolds) e^{-\der_{s_\alpha}},
    \qquad
    \boldQbar_\alpha(\bolds) \hat\boldWbar(\bolds)
    = \hat\boldWbar(\bolds) e^{-\der_{s_\alpha}},
\label{lin-eq:QWhat=What}
\end{equation}
 or, in terms of $\boldP_\alpha$,
\begin{equation}
    \boldP_\alpha(\bolds) \hat\boldW(\bolds)
    = \hat\boldW(\bolds+[1]_\alpha),
    \qquad
    \boldPbar_\alpha(\bolds) \hat\boldWbar(\bolds)
    = \hat\boldWbar(\bolds+[1]_\alpha).
\label{lin-eq:PWhat=What}
\end{equation}

 The operators $\hat\boldW(\bolds)$ and $\hat\boldWbar(\bolds)$ are
 unique up to multiplication of matrix difference operators of the form
\begin{align*}
    \hat\boldW(\bolds,{\boldsymbol t} ,\bar {\boldsymbol t})
    &\mapsto
    \hat\boldW(\bolds,{\boldsymbol t} ,\bar {\boldsymbol t})
    \sum_{j=0}^\infty c_j e^{-j\der_s}, \qquad c_0=1_N,
\\
    \hat\boldWbar(\bolds,{\boldsymbol t} ,\bar {\boldsymbol t})
    &\mapsto
    \hat\boldWbar(\bolds,{\boldsymbol t} ,\bar {\boldsymbol t})
    \sum_{j=0}^\infty \bar c_j e^{j\der_s},\qquad
    \bar c_0 \; \text{is invertible},
\end{align*}
 where $c_j$ and $\bar c_j$ are diagonal matrices\footnote{In Theorem
 3.3 of \cite{UT84} it is not stated that $c_j$ and $\bar c_j$ should be
 diagonal.} which do not depend on ${\boldsymbol t} $, 
 $\bar {\boldsymbol t}$ and $\bolds$.
 
(ii)
 Conversely, if $\hat\boldW(\bolds,\boldt,\boldtbar)$ and
 $\hat\boldWbar(\bolds,\boldt,\boldtbar)$ of the form \eqref{wave-mat}
 are solutions of differential equations \eqref{lin-eq:dWhat} and
 \eqref{lin-eq:dWbarhat} for certain difference operators
 $\boldB_{\alpha,n}$ and $\boldBbar_{\alpha,n}$, then the operators
 $\boldL$, $\boldLbar$, $\boldU_\alpha$, $\boldUbar_\alpha$,
 $\boldQ_\alpha$, $\boldQbar_\alpha$ defined by
 \eqref{lin-eq:LWhat,UWhat}, \eqref{lin-eq:LbarWbarhat,UbarWbarhat} and
 \eqref{lin-eq:QWhat=What}, or, equivalently, by
\begin{equation}
 \begin{aligned}
    \boldL(\bolds)
    &:=\hat\boldW(\bolds) e^{\der_s} \hat\boldW^{-1}(\bolds),&
    \boldU_\alpha(\bolds)
    &:=\hat\boldW(\bolds) E_\alpha \hat\boldW^{-1}(\bolds),
\\
    \boldLbar(\bolds)
    &:=\hat\boldWbar(\bolds) e^{-\der_s} \hat\boldWbar^{-1}(\bolds),&
    \boldUbar_\alpha(\bolds)
    &:=\hat\boldWbar(\bolds) E_\alpha \hat\boldWbar^{-1}(\bolds),
\\
    \boldQ_\alpha(\bolds)
    &:=\hat\boldW(\bolds)      e^{-\der_{s_\alpha}}
       \hat\boldW^{-1}(\bolds),&
    \boldQbar_\alpha(\bolds)
    &:=\hat\boldWbar(\bolds)   e^{-\der_{s_\alpha}}
       \hat\boldW^{-1}(\bolds),
 \end{aligned}
\label{W->L,U,P}
\end{equation}
 solve \eqref{lax:L,Lbar,U,Ubar,Q,Qbar} and
 also satisfy \eqref{def:BBbar} and relations (\ref{commutative:L,U,Q},
 \ref{U:alg-cond}, \ref{Q:alg-cond}, \ref{commutative:Lbar,Ubar,Qbar},
 \ref{Ubar:alg-cond}, \ref{Qbar:alg-cond}).

\end{predl}

We call $\boldW$ and $\boldWbar$ {\em wave operators}. They are
also sometimes called dressing operators because the expression 
$\boldL (\bolds )=\boldW(\bolds)e^{\p_s}\boldW^{-1}(\bolds)$
equivalent to the first equation in (\ref{lin-eq:LW,UW}) is
interpreted as a ``dressing'' of the ``bare'' shift operator 
$e^{\p_s}$ by the wave operator $\boldW$.

\begin{proof}
 (i) 
 The equivalence of equations for $(\boldW,\boldWbar)$ and
 $(\hat\boldW,\hat\boldWbar)$ follows immediately from $e^{\der_s}
 E_\alpha = E_\alpha e^{\der_s}$.

 We begin the proof assuming that
 each sequence of the $\bolds$-variables is of the form
 $\bolds=\bolds^{(0)}+s\boldone=\{s_1^{(0)}+s,\dotsc,s_N^{(0)}+s\}$ for a
 fixed $\bolds^{(0)}$ and consider all functions as functions
 of the single variable $s$. We denote
 $\boldL(\bolds,{\boldsymbol t} ,\bar{\boldsymbol t} )$,
 $\hat\boldW(\bolds,{\boldsymbol t} ,\bar{\boldsymbol t} )$, ...\ by
 $\boldL(s,{\boldsymbol t} ,\bar{\boldsymbol t} )$, $\hat\boldW(s,{\boldsymbol t} ,\bar{\boldsymbol t} )$, ...\
 respectively.
 
 The proof of statement (i) goes in a few steps.
\begin{enumerate}
 \item Find $\hat\boldW_0(s)=\hat\boldW(s,0,0)$ satisfying
       \eqref{lin-eq:LWhat,UWhat} and
       $\hat\boldWbar_0(s)=\hat\boldWbar(s,0,0)$ satisfying
       \eqref{lin-eq:LbarWbarhat,UbarWbarhat} at $\boldt=\boldtbar=0$. 
 \item Solve the differential equations 
       with initial values $\hat\boldW_0(s)$ and $\hat\boldWbar_0(s)$ to
       construct $\hat\boldW(s,\boldt, \bar {\boldsymbol t})$ and
       $\hat\boldWbar(s,{\boldsymbol t} ,\bar {\boldsymbol t})$.
 \item Show that these $\hat\boldW(s,{\boldsymbol t} ,
 \bar {\boldsymbol t})$ and
       $\hat\boldWbar(s,{\boldsymbol t} ,\bar {\boldsymbol t})$ satisfy
       \eqref{lin-eq:LWhat,UWhat},
       \eqref{lin-eq:LbarWbarhat,UbarWbarhat}, \eqref{lin-eq:dWhat} and
       \eqref{lin-eq:dWbarhat}.

 \item Show the uniqueness up to multiplication of constant diagonal
       matrix difference operators.

 \item Restore the dependence of the $\bolds$-variables 
 $\bolds=\{s_1,\dotsc,s_N\}$ to prove
       the existence of $\hat\boldW(\bolds,{\boldsymbol t} ,
       \bar{\boldsymbol t} )$ and
       $\hat\boldWbar(\bolds, {\boldsymbol t} , \bar{\boldsymbol t} )$ satisfying
       \eqref{lin-eq:QWhat=What} or, equivalently,
       \eqref{lin-eq:PWhat=What}. 
\end{enumerate}

 \underline{Step 1.} Construction of $\hat\boldW_0(s)$ and
 $\hat\boldWbar_0(s)$: Recall that $\boldL(s,0,0)$ has the form
\[
    \boldL(s,0,0)
    =
    \sum_{j=0}^\infty b_{0,j}(s) e^{(1-j)\der_s}
    :=
    \sum_{j=0}^\infty 
    b_j(s,{\boldsymbol t} =0,\bar {\boldsymbol t}=0) e^{(1-j)\der_s},
\]
 Hence, at ${\boldsymbol t} =\bar {\boldsymbol t}=0$ 
\[
     \hat\boldW_0(s)=\sum_{j=0}^\infty w_{0,j}(s) e^{-j\der_s}
\]
 should satisfy
\[
    \sum_{j=0}^\infty
    \left(\sum_{k=0}^j b_{0,k}(s)\, w_{0,j-k}(s+1-k) \right)
    e^{(1-j)\der_s}
    =
    \sum_{j=0}^\infty
    w_{0,j}(s) e^{(1-j)\der_s}
\]
 by the first equation of \eqref{lin-eq:LWhat,UWhat}, which means that
\[
    \sum_{k=0}^j b_{0,k}(s)\, w_{0,j-k}(s+1-k) = w_{0,j}(s)
\]
 for each integer $j\geq 1$.
 It is equivalent to the following system of difference
 equations for $\{w_{0,j}(s)\}_{j=0,1,\dotsc}$:
\[
    w_{0,j}(s+1) - w_{0,j}(s)
    =
    - \sum_{k=1}^j b_{0,k}(s)\, w_{0,j-k}(s+1-k),
\]
 which can be solved recursively with respect to $j$, starting from
 $w_{0,0}(s)=1_N$.

 Note that such an operator $\hat\boldW_0(s)$ is unique up to
 multiplication of an operator of the form $\sum_{j=0}^\infty c_j
 e^{-j\der_s}$ from the right, where $c_j$ is a constant matrix of size
 $N\times N$ and $c_0=1_N$. In fact, if $\hat\boldW_{0,1}(s)$ and
 $\hat\boldW_{0,2}(s)$ satisfy
 $\boldL(s,0,0)\,\hat\boldW_{0,i}(s)=\hat\boldW_{0,i}(s)\,e^{\der_s}$
 ($i=1,2$), then
 $
    \hat\boldW_{0,1}^{-1}(s) \hat\boldW_{0,2}(s)\, e^{\der_s} =
    e^{\der_s}\, \hat\boldW_{0,1}^{-1}(s) \hat\boldW_{0,2}(s)
 $.
 An operator commutes with $e^{\der_s}$, if and only if it has the form
 $\sum c_j e^{-j\der_s}$ with constant $c_j$, so
\[
    \hat\boldW_{0,2}(s)
    =
    \hat\boldW_{0,1}(s)
    \left(\sum_{j=0}^\infty c_j e^{-j\der_s}\right).
\]
 It follows $c_0=1_N$ from the condition $w_{0,0}(s)=1_N$.

 We construct $\hat\boldW_0(s)$ satisfying the condition
 $\boldU_\alpha(s,0,0)\hat\boldW_0(s)=\hat\boldW_0(s) E_\alpha$ in
 \eqref{lin-eq:LWhat,UWhat}, making use of this ambiguity\footnote{Here
 we refine the proof of Theorem 3.3 in \cite{UT84}. As the equation for
 $W^{(3)}$ in that proof is degenerate, it is not obvious that it has a
 solution.}.


 Let us take any $\hat\boldW_0(s)$ satisfying
 $\boldL(s,0,0)\hat\boldW_0(s)=\hat\boldW_0(s) e^{\der_s}$ and call it
 $\hat\boldW_{0,0}(s)=\sum_{j=0}^\infty w_{0,0,j}(s) e^{-j\der_s}$. By
 the adjoint action of $\hat\boldW_{0,0}^{-1}(s)$ to
 \eqref{commutative:L,U,Q} and \eqref{U:alg-cond} we obtain the
 following conditions for
 $
    \tilde\boldU_\alpha^{(0)}(s)
    :=
    \hat\boldW_{0,0}^{-1}(s)\, \boldU_\alpha(s,\boldt=0,\boldtbar=0)\,
    \hat\boldW_{0,0}(s)
 $: 
\begin{equation}
 \begin{gathered}
    [e^{\der_s}, \tilde\boldU_\alpha^{(0)}(s)] = 0, 
\\
    \tilde\boldU_\alpha^{(0)}(s) \tilde\boldU_\beta^{(0)}(s) 
    = 
    \delta_{\alpha\beta} \tilde\boldU_\beta^{(0)}(s),
    \qquad
    \sum_{\alpha=1}^N \tilde\boldU_\alpha^{(0)}(s)
    = 1_{N}.
 \end{gathered}
\label{alg-cond:tildeU}
\end{equation}
 Since $\tilde\boldU_\alpha^{(0)}(s)$ commutes with $e^{\der_s}$, it has
 a form
$
    \tilde\boldU_\alpha^{(0)}(s)
    =
    \sum_{j=0}^\infty \tilde u_{\alpha,j}^{(0)} e^{-j\der_s}
$,
 where $\tilde u_{\alpha,j}^{(0)}$ is a constant $N\times N$-matrix.
 Moreover, since $u_{\alpha,0}=E_\alpha$ and $w_{0,0,0}(s)=1_N$, $\tilde
 u_{\alpha,0}^{(0)}=E_{\alpha}$. Namely,
 $\tilde\boldU_\alpha^{(0)}=\tilde \boldU_\alpha^{(0)}(s)$ can be
 expanded as
\begin{equation}
    \tilde\boldU_\alpha^{(0)}
    =
    E_\alpha
    +
    \sum_{j=1}^\infty \tilde u_{\alpha,j}^{(0)} e^{-j\der_s}.
\label{tildeU} 
\end{equation}
 Assume that we have $\hat\boldW_{0,k}(s)$ ($k\geq0$) which satisfies
\begin{equation}
    \boldU_\alpha(s,0,0) \hat\boldW_{0,k}(s)
    =
    \hat\boldW_{0,k}(s) \tilde\boldU_\alpha^{(k)}, \quad
    \tilde\boldU_\alpha^{(k)}
    =   
    E_\alpha
    + \sum_{j=k+1}^\infty
       \tilde u_{\alpha,j}^{(k)} e^{-j\der_s},
\label{UW0k}
\end{equation}
 where $\tilde u_{\alpha,j}^{(k)}\in\Mat(N\times N,\Comp)$. Indeed the
 above chosen $\hat\boldW_{0,0}(s)$ satisfies this condition for $k=0$.

 By the same argument as that for $\hat\boldW_{0,0}(s)$ we can show that 
$$
    \tilde\boldU_\alpha^{(k)}
    = E_\alpha +
    \sum_{j=k+1}^\infty\tilde u_{\alpha,j}^{(k)} e^{-j\der_s}
$$
 in \eqref{UW0k} satisfies algebraic equations \eqref{alg-cond:tildeU}
 with the index $(k)$ instead of $(0)$. The equation
 $\bigl({\tilde\boldU}_\alpha^{(k)}\bigr)^2 = \tilde\boldU_\alpha^{(k)}$
 is expanded as
\[
    E_\alpha
    +
     (E_\alpha\tilde u_{\alpha,k+1}^{(k)}
     +\tilde u_{\alpha,k+1}^{(k)} E_\alpha)
     e^{-(k+1)\der_s}
    + \dotsb
    =
    E_\alpha + \tilde u_{\alpha,k+1}^{(k)} e^{-(k+1)\der_s}
    + \dotsb.
\]
 Hence, 
\[
    E_\alpha\tilde u_{\alpha,k+1}^{(k)}
    +
    \tilde u_{\alpha,k+1}^{(k)} E_\alpha
    =  \tilde u_{\alpha,k+1}^{(k)},
\]
 which means $0=(\tilde u_{\alpha,k+1}^{(k)})_{ij}$ ($i\neq\alpha$,
 $j\neq\alpha$, $i\neq j$) and 
$
    2(\tilde u_{\alpha,k+1}^{(k)})_{\alpha\alpha}
    =
    (\tilde u_{\alpha,k+1}^{(k)})_{\alpha\alpha}
$,
 namely,
\begin{equation}
    (\tilde u_{\alpha,k+1}^{(k)})_{ij}= 0
    \text{ if }(i\neq\alpha\text{ and } j\neq \alpha), \text{ or }
    i=j=\alpha.
\label{tildeu:element}
\end{equation}
 The commutativity equation
$
    \tilde\boldU_\alpha^{(k)} \tilde\boldU_\beta^{(k)}
    =
    \tilde\boldU_\beta^{(k)} \tilde\boldU_\alpha^{(k)}
$
 ($\alpha\neq\beta$) is expanded as
\[
    (\tilde u_{\alpha,k+1}^{(k)} E_{\beta}
    +E_{\alpha}\tilde u_{\beta,k+1}^{(k)})
    e^{-(k+1)\der_s}
    + \dotsb
    =
    (\tilde u_{\beta,k+1}^{(k)} E_{\alpha}
    +E_{\beta}\tilde u_{\alpha,k+1}^{(k)})
    e^{-(k+1)\der_s}
    + \dotsb,
\]
 in particular,
\[
    \tilde u_{\alpha,k+1}^{(k)} E_\beta
    +
    E_\alpha\tilde u_{\beta,k+1}^{(k)}
    =
    \tilde u_{\beta,k+1}^{(k)} E_\alpha
    +
    E_\beta\tilde u_{\alpha,k+1}^{(k)}.
\]
 Together with \eqref{tildeu:element} this implies
\begin{equation}
    (\tilde u_{\alpha,k+1}^{(k)})_{\alpha\beta}
    +
    (\tilde u_{\beta,k+1}^{(k)})_{\alpha\beta}
    =
    (\tilde u_{\alpha,k+1}^{(k)})_{\beta\alpha}
    +
    (\tilde u_{\beta,k+1}^{(k)})_{\beta\alpha}
    = 0.
\label{tildeua+tildeub=0}
\end{equation}
 (This is also a consequence of the condition
 $\sum_{\alpha=1}^N\tilde\boldU_\alpha^{(k)}=1_N$ in
 \eqref{alg-cond:tildeU}.) 

 If we can find an operator $\tilde\boldW_{k,k+1}$ such that
\begin{equation}
    \tilde\boldU^{(k)}_\alpha \tilde\boldW_{k,k+1}
    =
    \tilde\boldW_{k,k+1}
    \left(
     E_\alpha
     + \sum_{j=k+2}^\infty
       \tilde{\tilde u}_{\alpha,j} e^{-j\der_s}
    \right),
\label{tildeWk+1}
\end{equation}
 the operator $\hat\boldW_{0,k+1}(s)$ satisfying \eqref{UW0k} for $k+1$ is
 obtained as $$\hat\boldW_{0,k+1}(s)=\hat\boldW_{0,k}(s)
 \tilde\boldW_{k,k+1}$$ from $\hat\boldW_{0,k}(s)$. 

 As such an operator $\tilde\boldW_{k,k+1}$, we take
\begin{equation}
    \tilde\boldW_{k,k+1}
    =
    1_N + w'\, e^{-(k+1)\der_s}, \quad
    (w')_{\alpha\beta}
    :=
    (\tilde u_{\beta,k+1}^{(k)})_{\alpha\beta}
    =
    -(\tilde u_{\alpha,k+1}^{(k)})_{\alpha\beta}.
\label{tildew}
\end{equation}
 It is easy to see that this $\tilde\boldW_{k,k+1}$ satisfies
 \eqref{tildeWk+1} due to \eqref{tildeu:element} and
 \eqref{tildeua+tildeub=0}.  

 Note that multiplication $\hat\boldW_{0,k+1}(s)=\hat\boldW_{0,k}(s)
 \tilde\boldW_{k,k+1}$ by $\tilde\boldW_{k,k+1}$ of the form
 \eqref{tildew} does not change the coefficients $\hat w_{0,k,j}(s)$
 ($j=1,\dots,k$) in the  expansion 
\[
    \hat\boldW_{0,k}(s)
    =
    1_N +
    \sum_{j=1}^\infty \hat w_{0,k,j}(s)\, e^{-j\der_s}.
\]
 Hence the sequence $\{\hat\boldW_{0,k}(s)\}_k$ constructed in this way
 has a limit
\begin{equation}
    \hat\boldW_0(s) := \lim_{k\to\infty} \hat\boldW_{0,k}(s),
\label{What0=lim}
\end{equation}
 which satisfies the second equation in \eqref{lin-eq:LWhat,UWhat} for
 ${\boldsymbol t} =\bar {\boldsymbol t}=0$,
 $\boldU_\alpha(s,0,0)\hat\boldW_0(s)=\hat\boldW_0(s) E_\alpha$. The
 first equation $\boldL(s,0,0)\hat\boldW_0(s)=\hat\boldW_0(s)\,e^{\der_s}$
 of \eqref{lin-eq:LWhat,UWhat} is kept unchanged in the above procedure,
 as we have already mentioned.

 The operator $\hat\boldWbar_0(s)$ satisfying
 \eqref{lin-eq:LbarWbarhat,UbarWbarhat} at ${\boldsymbol t} =
 \bar {\boldsymbol t}=0$ is
 constructed in the same way. The first term $\bar w_{0,0}(s)$ in
\[
    \hat\boldWbar(s)
    =
    \sum_{j=0}^\infty \bar w_{0,j}(s)\, e^{j\der_s}
\]
 is $\tilde w_0(s,{\boldsymbol t} =0,\bar {\boldsymbol t}=0)$, 
 where $\tilde
 w_0(s,{\boldsymbol t} ,\bar {\boldsymbol t})$ is the matrix introduced
 in \propref{proposition:tildew0}.

\medskip
 \underline{Step 2.} Solving the differential equations for
 $\hat\boldW(s)$ and $\hat\boldWbar(s)$: The system
\begin{equation}
    \frac{\der\hat\boldW(s)}{\der t_{\alpha,n}}
    =
    -(\boldL^n(s)\boldU_\alpha(s))_{<0} \hat\boldW(s), 
    \qquad
    \frac{\der\hat\boldW(s)}{\der \bar t_{\alpha,n}}
    =
    \boldBbar_{\alpha,n}(s)\hat\boldW(s)
\label{lin-eq:dWhat:2}
\end{equation}
 is compatible because of \eqref{zs-LU<0LU<0}, \eqref{zs-LU<0Bbar} and
 \eqref{zs-BbarBbar}. Actually this is equivalent to the system
 \eqref{lin-eq:dWhat}, if $\hat\boldW(s)$ satisfies
 \eqref{lin-eq:LWhat,UWhat}. Let $\hat\boldW(s)$ be the unique solution
 of \eqref{lin-eq:dWhat:2} with the initial value
 $\hat\boldW(s,{\boldsymbol t} =0,\bar {\boldsymbol t}=0)=\hat\boldW_0(s)$.

 Since the coefficients in the right-hand sides of the equations in
 \eqref{lin-eq:dWhat:2} are difference operators with negative shifts,
 the solution $\hat\boldW$ is of the form
\[
    \hat\boldW(s,{\boldsymbol t} ,\bar {\boldsymbol t})
    =
    \sum_{j=0}^\infty w_j(s,{\boldsymbol t} ,
    \bar {\boldsymbol t}) e^{-j\der_s},
\] 
 with $w_0(s,{\boldsymbol t} ,\bar {\boldsymbol t})=1_N$. 

Similarly, the system 
\begin{equation}
    \frac{\der\hat\boldWbar(s)}{\der t_{\alpha,n}}
    =
    \boldB_{\alpha,n}(s) \hat\boldWbar(s), 
    \qquad
    \frac{\der\hat\boldWbar(s)}{\der \bar t_{\alpha,n}}
    =
    -(\boldLbar^n(s)\boldUbar_\alpha(s))_{\geq0}\hat\boldWbar(s),
\label{lin-eq:dWbarhat:2}
\end{equation}
 is compatible due to \eqref{zs-BB}, \eqref{zs-LUbar>0LUbar>0} and
 \eqref{zs-BLUbar>0} and equivalent to \eqref{lin-eq:dWbarhat}, if
 $\hat\boldWbar$ satisfies \eqref{lin-eq:LbarWbarhat,UbarWbarhat}. We
 take its solution $\hat\boldWbar(s)$ with the initial value
 $\hat\boldWbar(s,{\boldsymbol t} =0,\bar {\boldsymbol t}=0)=\hat\boldWbar_0(s)$. 

 Since the coefficients in the right-hand sides of the equations in
 \eqref{lin-eq:dWbarhat:2} are difference operators with non-negative
 shifts, the solution $\hat\boldWbar(s)$ is of the form
\[
    \hat\boldWbar(s,{\boldsymbol t} ,\bar {\boldsymbol t})
    =
    \sum_{j=0}^\infty \bar w_j(s,{\boldsymbol t} ,
    \bar {\boldsymbol t})\, e^{j\der_s},
\] 
 with $\bar w_0(s,{\boldsymbol t} =0,\bar {\boldsymbol t}=0)=\tilde
 w_0(s,{\boldsymbol t} =0,\bar {\boldsymbol t}=0)$.

\medskip
 \underline{Step 3.} The proof that $\hat\boldW$
 and $\hat\boldWbar$ constructed above satisfy equations 
 \eqref{lin-eq:LWhat,UWhat},
 \eqref{lin-eq:LbarWbarhat,UbarWbarhat}, \eqref{lin-eq:dWhat} and
 \eqref{lin-eq:dWbarhat}: 
 As we have mentioned in Step 2, equation \eqref{lin-eq:dWhat} is a consequence of
 \eqref{lin-eq:LWhat,UWhat} and \eqref{lin-eq:dWhat:2}, while equation
 \eqref{lin-eq:dWbarhat} is a consequence of
 \eqref{lin-eq:LbarWbarhat,UbarWbarhat} and
 \eqref{lin-eq:dWbarhat:2}. So, it is enough to show
 \eqref{lin-eq:LWhat,UWhat} and \eqref{lin-eq:LbarWbarhat,UbarWbarhat}.
 Since we do not touch $s$ here, we do not
 write it explicitly in the formulae below.

 The first Lax equation in \eqref{lax:L,Lbar,U,Ubar,Q,Qbar} for
 $\boldL$ and the first equation of \eqref{lin-eq:dWhat:2} imply
\[
 \begin{split}
    \frac{\der}{\der t_{\alpha,n}}
    (\boldL\hat\boldW-\hat\boldW e^{\der_s})
    &= [\boldB_{\alpha,n},\boldL]\hat\boldW
     - \boldL (\boldL^n\boldU_\alpha)_{<0} \hat\boldW
     - (\boldL^n\boldU_\alpha)_{<0} \hat\boldW e^{\der_s}
\\
    &= \boldB_{\alpha,n} \boldL\hat\boldW
     - \boldL (\boldL^n \boldU_\alpha) \hat\boldW
     - (\boldL^n\boldU_\alpha)_{<0} \hat\boldW e^{\der_s}
\\
    &= \boldB_{\alpha,n} \boldL\hat\boldW
     - (\boldL^n \boldU_\alpha) \boldL \hat\boldW
     - (\boldL^n\boldU_\alpha)_{<0} \hat\boldW e^{\der_s}
\\
    &= - (\boldL^n\boldU_\alpha)_{<0} 
      (\boldL \hat\boldW - \hat\boldW e^{\der_s})
 \end{split}
\]
 because $\boldL$ and $\boldU_\alpha$ commute. Similarly, the second Lax
 equation in \eqref{lax:L,Lbar,U,Ubar,Q,Qbar} and the second equation of 
 \eqref{lin-eq:dWhat:2} imply 
\[
 \begin{split}
    \frac{\der}{\der \bar t_{\alpha,n}}
    (\boldL\hat\boldW-\hat\boldW e^{\der_s})
    &= [\boldBbar_{\alpha,n},\boldL]\hat\boldW
     + \boldL \boldBbar_{\alpha,n} \hat\boldW
     - \boldBbar_{\alpha,n} \hat\boldW e^{\der_s}
\\
    &= \boldBbar_{\alpha,n}
      (\boldL \hat\boldW - \hat\boldW e^{\der_s}).
 \end{split}
\]
 Therefore, $\boldL \hat\boldW - \hat\boldW e^{\der_s}$ satisfies the
 system \eqref{lin-eq:dWhat:2} instead of $\hat\boldW$ with the initial
 value 
$
    (\boldL \hat\boldW - \hat\boldW e^{\der_s})|_{{\boldsymbol t} =
    \bar {\boldsymbol t}=0}
    =
    \boldL(0,0) \hat\boldW_0 - \hat\boldW_0 e^{\der_s} = 0
$.
 Since the solution of the Cauchy problem for the compatible system is
 unique, $\boldL \hat\boldW - \hat\boldW e^{\der_s}=0$ for all
 ${\boldsymbol t}$ and $\boldtbar$. Thus we have proved the first
 equation of \eqref{lin-eq:LWhat,UWhat}.

 The proof of the second equation of \eqref{lin-eq:LWhat,UWhat} is
 similar: the first Lax equation \eqref{lax:L,Lbar,U,Ubar,Q,Qbar} for
 $\boldU_\alpha$ and the first equation of \eqref{lin-eq:dWhat:2} imply
\[
 \begin{split}
    \frac{\der}{\der t_{\alpha,n}}
    (\boldU_\alpha\hat\boldW-\hat\boldW E_\alpha)
    &= [\boldB_{\alpha,n},\boldU_\alpha]\hat\boldW
     - \boldU_\alpha (\boldL^n\boldU_\alpha)_{<0} \hat\boldW
     - (\boldL^n\boldU_\alpha)_{<0} \hat\boldW E_\alpha
\\
    &= \boldB_{\alpha,n} \boldU_\alpha \hat\boldW
     - \boldU_\alpha (\boldL^n \boldU_\alpha) \hat\boldW
     - (\boldL^n\boldU_\alpha)_{<0} \hat\boldW E_\alpha
\\
    &= \boldB_{\alpha,n} \boldU_\alpha\hat\boldW
     - (\boldL^n \boldU_\alpha) \boldU_\alpha \hat\boldW
     - (\boldL^n\boldU_\alpha)_{<0} \hat\boldW E_\alpha
\\
    &= - (\boldL^n\boldU_\alpha)_{<0} 
      (\boldU_\alpha \hat\boldW - \hat\boldW E_\alpha),
 \end{split}
\]
 and the second Lax equation \eqref{lax:L,Lbar,U,Ubar,Q,Qbar} for
 $\boldU_\alpha$ and the second equation of \eqref{lin-eq:dWhat:2} imply
\[
 \begin{split}
    \frac{\der}{\der \bar t_{\alpha,n}}
    (\boldU_\alpha\hat\boldW-\hat\boldW E_\alpha)
    &= [\boldBbar_{\alpha,n},\boldU_\alpha]\hat\boldW
     + \boldU_\alpha \boldBbar_{\alpha,n} \hat\boldW
     - \boldBbar_{\alpha,n} \hat\boldW E_\alpha
\\
    &= \boldBbar_{\alpha,n}
      (\boldU_\alpha \hat\boldW - \hat\boldW E_\alpha).
 \end{split}
\]
 Together with
$
    ( \boldU_\alpha \hat\boldW
    - \hat\boldW E_\alpha)|_{{\boldsymbol t} =\bar {\boldsymbol t}=0}
    =
    \boldU_\alpha(0,0) \hat\boldW_0 - \hat\boldW_0 E_\alpha = 0
$
 these equations mean that 
 $\boldU_\alpha \hat\boldW - \hat\boldW E_\alpha=0$.

 Equations \eqref{lin-eq:LbarWbarhat,UbarWbarhat} for
 $\hat\boldWbar$ can be proved in the same way.

\medskip
 \underline{Step 4.} Check the uniqueness of $\hat\boldW$ and
 $\hat\boldWbar$: Assume
 that there are two operators $\hat\boldW_1$ and $\hat\boldW_2$
 satisfying the conditions \eqref{lin-eq:LWhat,UWhat} and
 \eqref{lin-eq:dWhat}, or equivalently, \eqref{lin-eq:LWhat,UWhat} and
 \eqref{lin-eq:dWhat:2}. Here again we do not write $s$ explicitly. 

 The first equation of \eqref{lin-eq:LWhat,UWhat} implies
\[
    e^{\der_s}(\hat\boldW_1^{-1}\hat\boldW_2)
    =
    (\hat\boldW_1^{-1}\hat\boldW_2)e^{\der_s},
    \text{ i.e., }
    \hat\boldW_1^{-1}\hat\boldW_2
    =
    \sum_{j=0}^\infty c_j({\boldsymbol t} ,\bar {\boldsymbol t}) 
    e^{-j\der_s},
\]
 where $c_j({\boldsymbol t} ,\bar {\boldsymbol t})$ is an $N\times N$-matrix independent of
 $s$. It follows from the second equation of \eqref{lin-eq:LWhat,UWhat}
 that 
\[
    E_\alpha(\hat\boldW_1^{-1}\hat\boldW_2)
    =
    (\hat\boldW_1^{-1}\hat\boldW_2) E_\alpha
    \text{ i.e., }
    c_j({\boldsymbol t} ,\bar {\boldsymbol t})\, E_{\alpha}
    =
    E_{\alpha}\, c_j({\boldsymbol t} ,\bar {\boldsymbol t})
    \text{ for any $\alpha$ and $j$.}
\]
 Therefore each $c_j({\boldsymbol t} ,\bar {\boldsymbol t})$ is a diagonal matrix.

 By the differential equations \eqref{lin-eq:dWhat:2}, 
\[
 \begin{split}
    \frac{\der}{\der t_{\alpha,n}} (\hat\boldW_1^{-1}\hat\boldW_2)
    &=
    -\hat\boldW_1^{-1}
     \frac{\der\hat\boldW_1}{\der t_{\alpha,n}} \hat\boldW_1^{-1}
     \hat\boldW_2
    + \hat\boldW_1^{-1} \frac{\der \hat\boldW_2}{\der t_{\alpha,n}}
\\
    &=
    -\hat\boldW_1^{-1} (-(\boldL^n\boldU_\alpha)_{<0}) \hat\boldW_2
    +\hat\boldW_1^{-1} (-(\boldL^n\boldU_\alpha)_{<0}) \hat\boldW_2
    = 0,
 \end{split}
\]
 and
\[
 \begin{split}
    \frac{\der}{\der \bar t_{\alpha,n}}
    (\hat\boldW_1^{-1}\hat\boldW_2)
    &=
    -\hat\boldW_1^{-1}
     \frac{\der\hat\boldW_1}{\der\bar t_{\alpha,n}}\hat\boldW_1^{-1}
     \hat\boldW_2
    + \hat\boldW_1^{-1}\frac{\der \hat\boldW_2}{\der\bar t_{\alpha,n}} 
\\
    &=
     \hat\boldW_1^{-1} \boldBbar_{\alpha,n} \hat\boldW_2
    -\hat\boldW_1^{-1} \boldBbar_{\alpha,n} \hat\boldW_2
    = 0,
 \end{split}
\]
 which means that $c_j({\boldsymbol t} ,\bar {\boldsymbol t})$ are constant. Thus we have
 proved that the ambiguity of $\hat\boldW({\boldsymbol t} ,
 \bar {\boldsymbol t})$ is of the
 form 
\begin{equation}
    \hat\boldW({\boldsymbol t} ,\bar {\boldsymbol t})
    \mapsto
    \hat\boldW({\boldsymbol t} ,\bar {\boldsymbol t})
    \sum_{j=0}^\infty c_j e^{-j\der_s},
\label{hatW:ambiguity}
\end{equation}
 where each $c_j$ is a constant diagonal matrix. Because of the
 normalization $w_0(s,{\boldsymbol t} ,\bar {\boldsymbol t})=1_N$, $c_0$ should be an
 identity matrix, $1_N$.

 Similarly, the ambiguity of $\hat\boldWbar({\boldsymbol t} ,
 \bar {\boldsymbol t})$ is of the
 form
\[
    \hat\boldWbar({\boldsymbol t} ,\bar {\boldsymbol t})
    \mapsto
    \hat\boldWbar({\boldsymbol t} ,\bar {\boldsymbol t})
    \sum_{j=0}^\infty \bar c_j e^{j\der_s},
\]
 where each $\bar c_j$ is a constant diagonal matrix and $\bar c_0$ is
 invertible.

\bigskip
 \underline{Step 5.} Now we restore the variables $\bolds$ and show that
 we can modify 
$
    \hat\boldW(s,{\boldsymbol t} ,\bar{\boldsymbol t} )
    =
    \hat\boldW(\bolds^{(0)}+s\boldone,{\boldsymbol t} ,\bar{\boldsymbol t} )
$
and
$
    \hat\boldWbar(s,{\boldsymbol t} ,\bar{\boldsymbol t} )
    =
    \hat\boldWbar(\bolds^{(0)}+s\boldone,{\boldsymbol t} ,\bar{\boldsymbol t} )
$
 obtained above so that they satisfy \eqref{lin-eq:QWhat=What}.

 First, note that any $\bolds\in\Integer^N$ can be uniquely decomposed
 as $\bolds=(s_1,\dotsc,s_N)=\bolds^{(0)} + s \boldone$, where
 $\bolds^{(0)}=(s^{(0)}_1,\dotsc,s^{(0)}_N)$ satisfies
 $s^{(0)}_1+\dotsb+s^{(0)}_N\in\{0,\dotsc,N-1\}$ and $s\in\Integer$. In
 fact, one has only to take $s:=$ the integer part of
 $(s_1+\dotsc+s_N)/N$ and $\bolds^{(0)}:=\bolds-s\boldone$. Thus any
 $\bolds$ is included in a uniquely defined sequence
 $\{\bolds^{(0)}+s\boldone\}_{s\in\Integer}$.

 We have already obtained
 $\hat\boldW(\bolds^{(0)}+s\boldone,{\boldsymbol t},\bar{\boldsymbol
 t})$ and $\hat\boldWbar(\bolds^{(0)}+s\boldone,{\boldsymbol
 t},\bar{\boldsymbol t})$ satisfying \eqref{lin-eq:LWhat,UWhat},
 \eqref{lin-eq:LbarWbarhat,UbarWbarhat}, \eqref{lin-eq:dWhat} and
 \eqref{lin-eq:dWbarhat} for each $\bolds^{(0)}$ and $s$.
 As we have noted above, this means that we have $\hat\boldW(\bolds)$
 and $\hat\boldWbar(\bolds)$ for each $\bolds$.
 Let us denote these temporary wave operators as
 $\hat\boldW^\temp(\bolds,{\boldsymbol t},\bar{\boldsymbol t} )$ and
 $\hat\boldWbar^\temp(\bolds,{\boldsymbol t},\bar{\boldsymbol t} )$
 respectively.

 Using the operator
 $\hat\boldW^\temp(\bolds,{\boldsymbol t} ,\bar{\boldsymbol t} )$, we
 modify the $\boldP_\bolda$-operators (in particular,
 the $\boldP_\alpha$-operators) as follows:
\begin{equation}
    \tilde\boldP_\bolda(\bolds)
    =
    (\hat\boldW^\temp(\bolds+\bolda))^{-1}
     \boldP_\bolda(\bolds)
    \hat\boldW^\temp(\bolds).
\label{tildeP}
\end{equation}
 It follows from equations \eqref{PL,PU} that
\begin{equation}
    \tilde\boldP_\bolda(\bolds) e^{\der_s}
    =
    e^{\der_s} \tilde\boldP_\bolda(\bolds),\qquad
    \tilde\boldP_\bolda(\bolds) E_\alpha,
    =
    E_\alpha \tilde\boldP_\bolda(\bolds)
\label{tildeP:invariance} 
\end{equation}
 which means that the coefficients $\tilde
 p_{\bolda,j}(\bolds,{\boldsymbol t} ,\bar{\boldsymbol t} )$ in the expansion
\begin{equation}
    \tilde\boldP_\bolda(\bolds)
    =
    \sum_{j=0}^\infty
    \tilde p_{\bolda,j}(\bolds,{\boldsymbol t} ,\bar{\boldsymbol t} ) e^{-j\der_s} 
\label{tildeP:expansion}
\end{equation}
 are diagonal matrices invariant under translation of the
 $\bolds$-variable: $\bolds\mapsto\bolds+\boldone$.  

 Differentiating \eqref{tildeP} by $t_{\alpha,n}$, we have
\begin{equation}
 \begin{split}
    \frac{\der \tilde\boldP_\bolda(\bolds)}{\der t_{\alpha,n}}
    ={}&
    e^{n\der_s} E_\alpha \tilde\boldP_\bolda(\bolds)
    -
    \tilde\boldP_\bolda(\bolds) e^{n\der_s} E_\alpha = 0,
 \end{split}
\label{dtildeP/dt} 
\end{equation}
 due to the differential equations \eqref{lin-eq:dWhat} for
 $\hat\boldW^\temp(\bolds)$, \eqref{P/t} for $\boldP_\bolda(\bolds)$ and
 commutativity \eqref{tildeP:invariance}. Similarly, the equation
\begin{equation}
    \frac{\der \tilde\boldP_\bolda(\bolds)}{\der \bar t_{\alpha,n}}
    = 0
\label{dtildeP/dtbar} 
\end{equation}
 follows from \eqref{lin-eq:dWhat}, \eqref{P/tbar} and
 \eqref{tildeP:invariance}. Namely, $\tilde\boldP_\bolda(\bolds)$ does
 not depend on ${\boldsymbol t} $ and $\bar{\boldsymbol t} $. The composition rule 
\begin{equation}
    \tilde\boldP_\boldb(\bolds+\bolda)\tilde\boldP_\bolda(\bolds)
    = \tilde\boldP_{\bolda+\boldb}(\bolds)
\label{tildePtildeP=tildeP}
\end{equation}
 is a consequence of the corresponding formula \eqref{PP=P}.

\begin{lem}
\label{lem:shift-tildeP}
 $\tilde\boldP_{\bolda+s\boldone}(\bolds^{(0)}) =
 \tilde\boldP_\bolda(\bolds^{(0)})$ for any $\bolda\in\Integer^N$,
 $s\in\Integer$ and $\bolds^{(0)}$.
\end{lem}

\begin{proof}
 By definition \eqref{Pa},
$
    \boldP_\boldone(\bolds)
    =
    e^{\der_s}\prod_{\alpha=1}^N\boldQ_\alpha(\bolds)
$.
 Because of the condition \eqref{Q:alg-cond}, the right-hand side is
 equal to $e^{\der_s}\boldL^{-1}(\bolds)$. Therefore, by
 \eqref{tildeP} 
\[
 \begin{split}
    \tilde\boldP_\boldone(\bolds) &=
    (\hat\boldW^\temp(\bolds+\boldone))^{-1}
    \boldP_\boldone(\bolds)
    \hat\boldW^\temp(\bolds)
\\
    &=
    (\hat\boldW^\temp(\bolds+\boldone))^{-1}
    e^{\der_s}\boldL^{-1}(\bolds)
    \hat\boldW^\temp(\bolds)
\\
    &=
    e^{\der_s} (\hat\boldW^\temp(\bolds))^{-1}
    \boldL^{-1}(\bolds)
    \hat\boldW^\temp(\bolds)
    = e^{\der_s} e^{-\der_s} = 1_N.
 \end{split}
\]
 The statement of the 
 lemma follows from this equation and the composition rule
 \eqref{tildePtildeP=tildeP}. 
\end{proof}

\medskip
 Let us modify the operators $\hat\boldW^\temp(\bolds)$.
 The new wave operators $\hat\boldW(\bolds)$ are defined by
\begin{equation}
    \hat\boldW(\bolds)
    :=
    \hat\boldW^\temp(\bolds)\,
    \tilde\boldP_{\bolds}(\boldzero)
    =
    \boldP_{\bolds}(\boldzero)\,
    \hat\boldW^\temp(\boldzero),
\label{W=WtentP}
\end{equation}
 where $\boldzero=\{0,\dotsc,0\}$. (The second equality follows from
 the definition \eqref{tildeP}.) This modification does not spoil
 equations \eqref{lin-eq:LWhat,UWhat} and \eqref{lin-eq:dWhat}, as
 $\tilde\boldP_{\bolds}(\boldzero)$ is an operator of
 the form \eqref{tildeP:expansion} whose coefficients are diagonal
 matrices and constant with respect to ${\boldsymbol t} $
 (\eqref{dtildeP/dt}), $\bar {\boldsymbol t}$ (\eqref{dtildeP/dtbar})
 and $\bolds$ (\lemref{lem:shift-tildeP}). 

 It follows from the definition \eqref{tildeP} of
 $\tilde\boldP_\bolda(\bolds)$ and \eqref{PP=P} that for
 $\boldb\in\Integer^N$
\[
 \begin{split}
    \hat\boldW(\bolds+\boldb)
    &=
    \boldP_{\bolds+\boldb}(\boldzero)) \hat\boldW^\temp(\boldzero)
\\
    &=
    \boldP_{\boldb}(\bolds) \boldP_{\bolds}(\boldzero)
    \hat\boldW^\temp(\boldzero)
    =
    \boldP_{\boldb}(\bolds)
    \hat\boldW(\bolds).
 \end{split}
\]
 In particular, for $\boldb=[1]_\alpha$, we have
\[
    \hat\boldW(\bolds+[1]_\alpha)
    =
    \boldP_\alpha(\bolds) \hat\boldW(\bolds),
\]
 which is the first equation in \eqref{lin-eq:PWhat=What}.

 Thus we have obtained the 
 desired operators $\hat\boldW(\bolds)$. The
 assertion about the uniqueness follows from that of
 $\hat\boldW^\temp(s)$, \eqref{hatW:ambiguity}, and
 \eqref{lin-eq:PWhat=What}.

 The existence of $\hat\boldWbar(\bolds)$ is proved in a similar way.

\bigskip
 The converse statement (ii) follows immediately from definitions
 \eqref{W->L,U,P} and their derivatives, once
 $\boldB_{\alpha,n}(\bolds)$ and $\boldBbar_{\alpha,n}(\bolds)$ are
 expressed in terms of $\boldL(\bolds)$, $\boldLbar(\bolds)$,
 $\boldU_\alpha(\bolds)$ and $\boldUbar_\alpha(\bolds)$ as in
 \eqref{def:BBbar}.

 Both of the operators $\boldB_{\alpha,n}(\bolds)$ and
 $\boldBbar_{\alpha,n}(\bolds)$ are expressed in two ways by
 $\hat\boldW$ and $\hat\boldWbar$ because of \eqref{lin-eq:dWhat} and
 \eqref{lin-eq:dWbarhat} as follows:
\begin{align*}
    \boldB_{\alpha,n}(\bolds)
    &=\frac{\der\hat\boldWbar(\bolds)}{\der t_{\alpha,n}}
    \hat\boldWbar^{-1}(\bolds)
    = 
    \hat\boldW(\bolds) e^{n\der_s} E_\alpha \hat\boldW^{-1}(\bolds)
    +
    \frac{\der \hat\boldW(\bolds)}{\der t_{\alpha,n}}
    \hat\boldW^{-1}(\bolds),
\\
    \boldBbar_{\alpha,n}(\bolds)
    &=\frac{\der\hat\boldW(\bolds)}{\der \bar t_{\alpha,n}}
    \hat\boldW^{-1}(\bolds)
    =
    \hat\boldWbar(\bolds) e^{-n\der_s} E_\alpha 
    \hat\boldWbar^{-1}(\bolds)
    +
    \frac{\der \hat\boldWbar(\bolds)}{\der \bar t_{\alpha,n}}
    \hat\boldWbar^{-1}(\bolds).
\end{align*}
 The first equality in the first equation means that
 $\boldB_{\alpha,n}(\bolds)$ is a difference operator with non-negative
 shifts. Hence the latter half of the same equation implies
$
    \boldB_{\alpha,n}(\bolds)
    =
    (\hat\boldW(\bolds) e^{n\der_s}
     E_\alpha\hat\boldW^{-1}(\bolds))_{\geq0}
$.
 Similarly,
$
    \boldBbar_{\alpha,n}(\bolds)
    =
    (\hat\boldWbar(\bolds) e^{-n\der_s}
     E_\alpha \hat\boldW^{-1}(\bolds))_{<0}
$
 follows from the second equation. Thus equations \eqref{def:BBbar}
 are proved.
\end{proof}

\begin{zam}
\label{rem:w0-tilde-neq-w0-bar}
 Comparing the coefficients of $e^{-\der_s}$ in the expansion
 of the first equation of \eqref{lin-eq:LbarWbarhat,UbarWbarhat} by
 using
$
    \bar b_0(\bolds)=\tilde w_0(\bolds)\, \tilde
 w_0^{-1}(\bolds-\boldone)
$
 (\propref{proposition:tildew0}), we have 
\[
    \tilde w_0(\bolds)\, \tilde w_0^{-1}(\bolds-\boldone)\,
    \bar w_0(\bolds-\bold1) = \bar w_0(\bolds),
    \ \text{i.e., }
    \tilde w_0^{-1}(\bolds-\boldone)\, \bar w_0(\bolds-\bold1)
    =
    \tilde w_0^{-1}(\bolds)\, \bar w_0(\bolds),
\]
 which means that the matrix $\tilde{\bar c}(\bolds):=\tilde
 w_0^{-1}(\bolds)\, \bar w_0(\bolds)$ is invariant under the shift
 $\bolds\mapsto \bolds+a\boldone$ ($a\in\Integer$). Similarly, it
 follows from the second equation in
 \eqref{lin-eq:LbarWbarhat,UbarWbarhat} and
$
    \bar u_{\alpha,0}(\bolds) 
    =
    \tilde w_0(\bolds) E_\alpha \tilde w_0^{-1}(\bolds),
$
 (\propref{proposition:tildew0})
 that
\[
     \tilde w_0^{-1}(\bolds)\, \bar w_0(\bolds)\, E_\alpha
     =
     E_\alpha\, \tilde w_0^{-1}(\bolds)\, \bar w_0(\bolds)
\]
 for any $\alpha=1,\dotsc,N$. Therefore $\tilde{\bar c}(\bolds)$ is a
 diagonal matrix.
 The second equation of \eqref{lin-eq:PWhat=What} and
$
    \bar p_{\alpha,0}(\bolds) 
    =
    \tilde w_0(\bolds+[1]_\alpha)\, \tilde w_0^{-1}(\bolds)
$
 (\propref{proposition:tildew0}) implies
\[
    \tilde w_0^{-1}(\bolds+[1]_\alpha)\, \bar w_0(\bolds+[1]_\alpha)
    =
    \tilde w_0^{-1}(\bolds)\, \bar w_0(\bolds).
\]
 Hence $\tilde{\bar c}(\bolds)$ is invariant under any shift
 $\bolds\mapsto\bolds+\bolda$ ($\bolda\in\Integer^N$), i.e., it
 does not depend on $\bolds$.
 Thus we have shown that $\bar w_0(\bolds)$ and $\tilde w_0(\bolds)$ are
 related as
 $\bar w_0(\bolds)
    =\tilde w_0(\bolds) \,
 \tilde{\bar c}(\boldt,\boldtbar)$, 
 where $\tilde{\bar c}(\boldt,\boldtbar)$ is a diagonal matrix
 independent of $\bolds$. However, it can depend on $\boldt$ and
 $\boldtbar$ non-trivially. So, 
 $\bar w_0(\bolds)$ and $\tilde
 w_0(\bolds)$ are closely related but, strictly speaking, are 
 not the same. 
 This is just the 
 ambiguity in the definition of $\tilde w_0(\bolds)$ (see
 Remark \ref{remark:freedom}). One may also say that the freedom
 in the definition of $\tilde w_0$ is partially fixed in the 
 $\bar w_0$. 
 Actually, once we have found $\bar w_0(\bolds)$, we can replace $\tilde
 w_0(\bolds)$ in all equations in 
 Section \ref{sec:def-multi-toda} 
 with $\bar w_0(\bolds)$. This does not change
 those equations at all. 
\end{zam}

A useful corollary of Proposition \ref{predl:wave-mat} is the following
statement.

\begin{predl}\label{proposition:alternative}
The following relations hold:
\beq\label{n10a}
\bar b_0(\bolds )=\bar w_0(\bolds )\bar w_0^{-1}(\bolds -1)=
\p_{\bar t_1}w_1 (\bolds ), \qquad
\p_{\bar t_1}\equiv \sum_{\mu =1}^N\p_{\bar t_{\mu ,1}},
\eeq
\beq\label{n11}
b_1(\bolds )=\p_{t_1}\bar w_0(\bolds )
\bar w_0^{-1}(\bolds ), \qquad
\p_{t_1}\equiv \sum_{\mu =1}^N\p_{ t_{\mu ,1}}.
\phantom{aaaaaaaaa}
\eeq
\end{predl}

\begin{proof}
These relations follow immediately from equations
(\ref{lin-eq:dWhat}) and (\ref{lin-eq:dWbarhat}) for the 
wave operators $\hat \boldW$, $\hat \boldWbar$. 
They are obtained 
from them by summing over $\alpha$ from $1$ to $N$ and
restriction to the highest coefficients. 
\end{proof}

\subsection{Wave functions}

The system of linear equations \eqref{lin-eq:LW,UW},
\eqref{lin-eq:LbarWbar,UbarWbar}, \eqref{lin-eq:dW},
\eqref{lin-eq:dWbar} and \eqref{lin-eq:QW} can be rewritten as a linear
system for the {\em matrix wave functions} defined by the wave matrices
as
\begin{equation}
 \begin{aligned}
    \boldPsi(\bolds,{\boldsymbol t} ,\bar {\boldsymbol t};z)
    &=
    (\Psi_{\alpha\beta}(\bolds,{\boldsymbol t} ,\bar {\boldsymbol t};z)
    )_{\alpha,\beta}
    :=
    \boldW(\bolds,{\boldsymbol t} ,\bar {\boldsymbol t}) 
    \diag_{\alpha}(z^{s_\alpha})
\\
    &=
    \sum_{j=0}^\infty w_j(\bolds,{\boldsymbol t} ,\bar {\boldsymbol t})\,
     \diag_{\alpha}\bigl(
      z^{s_\alpha-j} e^{\xi({\boldsymbol t} _\alpha,z)}
     \bigr),
\\
    \boldPsibar(\bolds,{\boldsymbol t} ,\bar {\boldsymbol t};z)
    &=
    (\bar\Psi_{\alpha\beta}(\bolds,{\boldsymbol t} ,\bar {\boldsymbol t};z)
    )_{\alpha,\beta}
    :=
    \boldWbar(\bolds,{\boldsymbol t} ,\bar {\boldsymbol t}) \diag_{\alpha}(z^{s_\alpha})
\\
    &=
    \sum_{j=0}^\infty \bar w_j(\bolds,{\boldsymbol t} ,
    \bar {\boldsymbol t})\,
    \diag_\alpha \bigl(
     z^{s_\alpha+j} e^{\xi(\bar {\boldsymbol t}_\alpha,z^{-1})}
    \bigr).
 \end{aligned}
\label{wave-func}
\end{equation}
In other words, the matrix elements of the wave functions are
\begin{equation}
 \begin{aligned}
    \Psi_{\alpha\beta}(\bolds,{\boldsymbol t} ,\bar {\boldsymbol t};z)
    &=
    w_{\alpha\beta}(\bolds,{\boldsymbol t} ,\bar {\boldsymbol t};z)\,
    z^{s_\beta} e^{\xi({\boldsymbol t} _\beta,z)},\
    w_{\alpha\beta}(\bolds,{\boldsymbol t} ,\bar {\boldsymbol t};z)
    :=
    \sum_{j=0}^\infty
    w_{j,\alpha\beta}(\bolds,{\boldsymbol t} ,\bar {\boldsymbol t})\, z^{-j},
\\
    \bar\Psi_{\alpha\beta}(\bolds,{\boldsymbol t} ,\bar {\boldsymbol t};z)
    &=
    \bar w_{\alpha\beta}(\bolds,{\boldsymbol t} ,\bar {\boldsymbol t};z)\,
    z^{s_\beta} e^{\xi(\bar {\boldsymbol t}_\beta,z^{-1})},\
    \bar w_{\alpha\beta}(\bolds,{\boldsymbol t} ,\bar {\boldsymbol t};z)
    :=
    \sum_{j=0}^\infty
    \bar w_{j,\alpha\beta}
    (\bolds,{\boldsymbol t} ,\bar {\boldsymbol t})\, z^j.
 \end{aligned}
\label{wave-func:element}
\end{equation}
The linear systems for $\boldPsi$ and $\boldPsibar$ are as follows:
\begin{align}
    \boldL(\bolds) \boldPsi(\bolds;z)
    &= z\, \boldPsi(\bolds;z),
    &
    \boldU_\alpha(\bolds) \boldPsi(\bolds;z)
    &= \boldPsi(\bolds;z) E_\alpha,
\label{lin-eq:LPsi,UPsi}
\\
    \boldLbar(\bolds) \boldPsibar(\bolds;z)
    &= z^{-1}\, \boldPsibar(\bolds;z),
    &
    \boldUbar_\alpha(\bolds) \boldPsibar(\bolds;z)
    &= \boldPsibar(\bolds) E_\alpha,
\label{lin-eq:LbarPsibar,UbarPsibar}
\\
    \frac{\der \boldPsi(\bolds;z)}{\der t_{\alpha,n}}
    &= \boldB_{\alpha,n}(\bolds) \boldPsi(\bolds;z), &
    \frac{\der \boldPsi(\bolds;z)}{\der \bar t_{\alpha,n}}
    &= \boldBbar_{\alpha,n}(\bolds) \boldPsi(\bolds;z),
\label{lin-eq:dPsi}
\\
    \frac{\der \boldPsibar(\bolds;z)}{\der t_{\alpha,n}}
    &= \boldB_{\alpha,n}(\bolds) \boldPsibar(\bolds;z), &
    \frac{\der \boldPsibar(\bolds;z)}{\der \bar t_{\alpha,n}}
    &= \boldBbar_{\alpha,n}(\bolds) \boldPsibar(\bolds;z),
\label{lin-eq:dPsibar}
\end{align}
 and
\begin{equation}
    \boldQ_\alpha(\bolds) \boldPsi(\bolds;z)
    = \boldPsi(\bolds;z) \diag_\gamma(z^{-\delta_{\alpha\gamma}})
    \qquad
    \boldQbar_\alpha(\bolds) \boldPsibar(\bolds;z)
    = \boldPsibar(\bolds;z) \diag_\gamma(z^{-\delta_{\alpha\gamma}}).
\label{lin-eq:QPsi}
\end{equation}
The {\em adjoint wave functions} are defined by the adjoint actions of
wave matrices\footnote{The formal adjoint operator $A^*$ of
$A=e^{n\der_{s_\alpha}}\circ a(\bolds)$ is defined by $A^*=a(\bolds)^T
e^{-n\der_{s_\alpha}}$, where $(\cdot)^T$ is the transposed matrix.}: 
\begin{equation}
 \begin{aligned}
    \boldPsi^*(\bolds,{\boldsymbol t} ,\bar {\boldsymbol t};z)
    &=
    (\Psi^*_{\alpha\beta}(\bolds,{\boldsymbol t} ,\bar {\boldsymbol t};z)
    )_{\alpha,\beta}
    =
    \left(
     \bigl(\boldW^{-1}(\bolds-\boldone,{\boldsymbol t} ,\bar {\boldsymbol t})\bigr)^*
     \diag_{\alpha}(z^{-s_\alpha})
    \right)^T
\\
    &=
    \sum_{j=0}^\infty
     \diag_{\alpha}\bigl(
      z^{-s_\alpha-j} e^{-\xi({\boldsymbol t} _\alpha,z)}
     \bigr)\,
     w^*_j(\bolds,{\boldsymbol t} ,\bar {\boldsymbol t}),
\\
    \boldPsibar^*(\bolds,{\boldsymbol t} ,\bar {\boldsymbol t};z)
    &=
    (\bar\Psi^*_{\alpha\beta}(\bolds,{\boldsymbol t} ,\bar {\boldsymbol t};z)
    )_{\alpha,\beta}
    =
    \left(
     \bigl(\boldWbar^{-1}(\bolds-\boldone,{\boldsymbol t} ,
     \bar {\boldsymbol t})\bigr)^*
     \diag_{\alpha}(z^{-s_\alpha})
    \right)^T
\\
    &=
    \sum_{j=0}^\infty
    \diag_\alpha \bigl(
     z^{-s_\alpha+j} e^{-\xi(\bar {\boldsymbol t}_\alpha,z^{-1})}
    \bigr)\,
    \bar w^*_j(\bolds,{\boldsymbol t} ,\bar {\boldsymbol t}).
 \end{aligned}
\label{adj-wave-func}
\end{equation}
Here we expand the inverses of $\hat\boldW$ and $\hat\boldWbar$ as
\begin{equation}
 \begin{aligned}
    \hat\boldW^{-1}(\bolds,{\boldsymbol t} ,\bar {\boldsymbol t})
    &= \sum_{j=0}^\infty
    e^{-j\der_s}\, w^*_j(\bolds+\boldone,{\boldsymbol t} ,
    \bar {\boldsymbol t}), 
    \qquad w^*_0(\bolds,{\boldsymbol t} ,\bar {\boldsymbol t}) = 1_N,
\\
    \hat\boldWbar^{-1}(\bolds,{\boldsymbol t} ,\bar {\boldsymbol t})
    &= \sum_{j=0}^\infty
    e^{j\der_s}\, \bar w^*_j(\bolds+\boldone,{\boldsymbol t} ,
    \bar {\boldsymbol t}),
    \qquad \bar w^*_0(\bolds,{\boldsymbol t} ,
    \bar {\boldsymbol t}) \in GL(N,\Comp),
 \end{aligned}
\label{inv-wave-mat}
\end{equation}
where
\begin{equation}
 \begin{gathered}
    w^*_j(\bolds,{\boldsymbol t} ,\bar {\boldsymbol t})
    =
    (w^*_{j,\alpha\beta}(\bolds,{\boldsymbol t} ,\bar {\boldsymbol t})
    )_{\alpha,\beta=1,\dotsc,N}
    \in \Mat(N\times N,\Comp),
\\
    \bar w^*_j(\bolds,{\boldsymbol t} ,\bar {\boldsymbol t})
    =
    (\bar w^*_{j,\alpha\beta}(\bolds,{\boldsymbol t} ,\bar {\boldsymbol t})
    )_{\alpha,\beta=1,\dotsc,N}
    \in \Mat(N\times N,\Comp).
 \end{gathered}
\label{wave-mat-inverse:notations}
\end{equation}
The matrix elements of the adjoint wave functions are
\begin{equation}
 \begin{aligned}
    \Psi^*_{\alpha\beta}(\bolds,{\boldsymbol t} ,\bar {\boldsymbol t};z)
    &=
    w^*_{\alpha\beta}(\bolds,{\boldsymbol t} ,\bar {\boldsymbol t};z)\,
    z^{-s_\alpha} e^{-\xi({\boldsymbol t}_\alpha,z)},\ 
    w^*_{\alpha\beta}(\bolds,{\boldsymbol t} ,\bar {\boldsymbol t};z)
    :=
    \sum_{j=0}^\infty
    w^*_{j,\alpha\beta}(\bolds,{\boldsymbol t} ,
    \bar {\boldsymbol t})\, z^{-j},
\\
    \bar\Psi^*_{\alpha\beta}(\bolds,{\boldsymbol t} ,\bar {\boldsymbol t};z)
    &=
    \bar w^*_{\alpha\beta}(\bolds,{\boldsymbol t} ,\bar {\boldsymbol t};z)\,
    z^{-s_\alpha} e^{-\xi(\bar {\boldsymbol t}_\alpha,z^{-1})},\ 
    \bar w^*_{\alpha\beta}(\bolds,{\boldsymbol t} ,\bar {\boldsymbol t};z)
    :=
    \sum_{j=0}^\infty
    \bar w^*_{j,\alpha\beta}(\bolds,{\boldsymbol t} ,
    \bar {\boldsymbol t})\, z^{j}.
 \end{aligned}
\label{adj-wave-func:element}
\end{equation}

\section{The bilinear identity}
\label{sec:bilinear}

The wave operators $\boldW$ and $\boldWbar$ are characterized by a
bilinear identity satisfied by the wave functions and the adjoint wave
functions. 

\begin{predl}
\label{predl:bil-res}
 (i)
 The wave functions $\boldPsi(\bolds,{\boldsymbol t} ,
 \bar {\boldsymbol t};z)$,
 $\bar\boldPsi(\bolds,{\boldsymbol t} ,\bar {\boldsymbol t};z)$ and the 
 adjoint wave functions
 $\bold\Psi^*(\bolds,{\boldsymbol t} ,\bar {\boldsymbol t};z)$,
 $\bar\boldPsi^*(\bolds,{\boldsymbol t} ,
 \bar {\boldsymbol t};z)$ of the multi-component Toda
 lattice hierarchy satisfy the following bilinear identity:
\begin{equation}
    \oint_{C_\infty}
    \boldPsi(\bolds,{\boldsymbol t} ,\bar {\boldsymbol t};z)\,
    \boldPsi^*(\bolds',{\boldsymbol t} ',\bar {\boldsymbol t}';z)\,
    dz
    =
    \oint_{C_0}
    \bar\boldPsi(\bolds,{\boldsymbol t} ,\bar {\boldsymbol t};z)\,
    \bar\boldPsi^*(\bolds',{\boldsymbol t} ',\bar {\boldsymbol t}';z)\,
    dz,
\label{bil-res-id:mat}
\end{equation}
 where $C_\infty$ is a circle around $\infty$ and $C_0$ is a small 
 circle around $0$. This identity holds for all $\bolds$, $\bolds'$,
 ${\boldsymbol t}$, $\boldt'$, $\boldtbar$ and $\bar {\boldsymbol t}'$.

(ii)
 Conversely, assume that matrix-valued functions of the form
\begin{equation}
 \begin{aligned}
    \boldPsi(\bolds,{\boldsymbol t} ,\bar {\boldsymbol t};z)
    &=
    \sum_{j=0}^\infty w_j(\bolds,{\boldsymbol t} ,\bar {\boldsymbol t})\,
     \diag_{\alpha}\bigl(
      z^{s_\alpha-j} e^{\xi({\boldsymbol t} _\alpha,z)}
     \bigr),
\\
    \boldPsibar(\bolds,{\boldsymbol t} ,\bar {\boldsymbol t};z)
    &=
    \sum_{j=0}^\infty \bar w_j(\bolds,{\boldsymbol t} ,
    \bar {\boldsymbol t})\,
    \diag_\alpha \bigl(
     z^{s_\alpha+j} e^{\xi(\bar {\boldsymbol t}_\alpha,z^{-1})}
    \bigr),
\\
    \boldPsi^*(\bolds,{\boldsymbol t} ,\bar {\boldsymbol t};z)
    &=
    \sum_{j=0}^\infty
     \diag_{\alpha}\bigl(
      z^{-s_\alpha-j} e^{-\xi({\boldsymbol t} _\alpha,z)}
     \bigr)\,
     w^*_j(\bolds,{\boldsymbol t} ,\bar {\boldsymbol t}),
\\
    \boldPsibar^*(\bolds,{\boldsymbol t} ,\bar {\boldsymbol t};z)
    &=
    \sum_{j=0}^\infty
    \diag_\alpha \bigl(
     z^{-s_\alpha+j} e^{-\xi(\bar {\boldsymbol t}_\alpha,z^{-1})}
    \bigr)\,
    \bar w^*_j(\bolds,{\boldsymbol t} ,\bar {\boldsymbol t})
 \end{aligned}
\label{Psi,Psibar,Psi*,Psi*bar}
\end{equation}
 satisfy the bilinear identity \eqref{bil-res-id:mat}. (In
 \eqref{Psi,Psibar,Psi*,Psi*bar} $w_j$, $\bar w_j$, $w_j^*$, $\bar
 w_j^*$ are $N\times N$ matrices and $w_0=w_0^*=1_N$, $\bar w_0$, $\bar
 w_0^* \in GL(N, \Comp)$.)  Then they are wave functions and adjoint
 wave functions of the multi-component Toda lattice hierarchy. Namely,
 the functions $w_j$, $\bar w_j$, $w_j^*$ and $\bar w_j^*$ in
 \eqref{Psi,Psibar,Psi*,Psi*bar} are the coefficients of the wave
 matrices and their inverse matrices in the expansions \eqref{wave-mat}
 and \eqref{inv-wave-mat} and the sextet $(\boldL, \boldLbar,
 \boldU_\alpha, \boldUbar_\alpha,
 \boldQ_\alpha,\boldQbar_\alpha)_{\alpha=1,\dotsc,N}$ defined by
 \eqref{W->L,U,P} is a solution of the $N$-component Toda lattice
 hierarchy.

\end{predl}

In terms of matrix elements the bilinear identity
\eqref{bil-res-id:mat} acquires the form
\commentout{
\begin{multline}
    \sum_{\gamma=1}^N
    \oint_{C_\infty}
    \Psi_{\alpha\gamma}(\bolds,{\boldsymbol t} ,\bar {\boldsymbol t};z)\,
    \Psi^*_{\gamma\beta}(\bolds',{\boldsymbol t} ',
    \bar {\boldsymbol t}';z)\,
    dz
\\
    =
    \sum_{\gamma=1}^N
    \oint_{C_0}
    \bar\Psi_{\alpha\gamma}(\bolds,{\boldsymbol t} ,
    \bar {\boldsymbol t};z)\,
    \bar\Psi^*_{\gamma\beta}(\bolds',{\boldsymbol t} ',
    \bar {\boldsymbol t}';z)\,
    dz,
\label{bil-res-id:element}
\end{multline}
or,}
\begin{multline}
    \sum_{\gamma=1}^N
    \oint_{C_\infty}
    z^{s_\gamma-s'_\gamma}
    e^{\xi({\boldsymbol t} _\gamma-{\boldsymbol t} '_\gamma,z)}\,
    w_{\alpha\gamma}(\bolds,{\boldsymbol t} ,\bar {\boldsymbol t};z)\,
    w^*_{\gamma\beta}(\bolds',{\boldsymbol t} ',\bar {\boldsymbol t}';z)\,
    dz
\\
    =
    \sum_{\gamma=1}^N
    \oint_{C_0} 
    z^{s_\gamma-s'_\gamma}
    e^{\xi(\bar {\boldsymbol t}_\gamma-
    \bar {\boldsymbol t}'_\gamma,z^{-1})}
    \bar w_{\alpha\gamma}(\bolds,{\boldsymbol t} ,
    \bar {\boldsymbol t};z)\,
    \bar w^*_{\gamma\beta}(\bolds',{\boldsymbol t} ',
    \bar {\boldsymbol t}';z)\,
    dz.
\label{bil-res-id:w}
\end{multline}

\begin{proof}
 (i)
 Assume that $\boldPsi$ is a matrix wave function of the multi-component
 Toda lattice hierarchy and $\boldPsi^*$ is a corresponding adjoint wave
 function.

 First, let us regard the left-hand side and the right-hand side of
 \eqref{bil-res-id:mat} as functions of ${\boldsymbol t} $ and $\bar
 {\boldsymbol t}$ with parameters $\bolds$, $\bolds'$, ${\boldsymbol t}
 '$ and $\bar {\boldsymbol t}'$. Then, both of them (${\boldsymbol T} =$
 the right-hand side or the left-hand side) satisfy the same linear
 system of differential equations,
\[
    \frac{\der}{\der t_{\alpha,n}}
    {\boldsymbol T} ({\boldsymbol t} ,\bar {\boldsymbol t})
    = \boldB_{\alpha,n}\, 
      {\boldsymbol T}({\boldsymbol t} ,\bar {\boldsymbol t}),
    \qquad
    \frac{\der}{\der \bar t_{\alpha,n}}
    {\boldsymbol T} ({\boldsymbol t} ,\bar {\boldsymbol t})
    = \boldBbar_{\alpha,n}\, 
      {\boldsymbol T}({\boldsymbol t} ,\bar {\boldsymbol t}),
\]
 because of \eqref{lin-eq:dPsi} (for the left-hand side) and
 \eqref{lin-eq:dPsibar} (for the right-hand side). Since this system is
 compatible by the Zakharov-Shabat equations \eqref{zs-BB},
 \eqref{zs-BbarBbar} and \eqref{zs-BBbar}, its solution
 ${\boldsymbol T} ({\boldsymbol t} ,\bar {\boldsymbol t})$ is determined 
 by the initial value,
 ${\boldsymbol T} ({\boldsymbol t} ', \bar {\boldsymbol t}')$.
 Hence, in order to prove
 \eqref{bil-res-id:mat}, i.e.\ \eqref{bil-res-id:w}, it is sufficient to
 prove the identity with ${\boldsymbol t} ={\boldsymbol t} '$ 
 and $\bar {\boldsymbol t}=\bar {\boldsymbol t}'$,
 namely,
\begin{multline}
    \sum_{\gamma=1}^N
    \oint_{C_\infty}
    z^{s_\gamma-s'_\gamma}\,
    w_{\alpha\gamma}(\bolds,{\boldsymbol t} ,\bar {\boldsymbol t};z)\,
    w^*_{\gamma\beta}(\bolds',{\boldsymbol t} ,\bar {\boldsymbol t};z)\,
    dz
\\
    =
    \sum_{\gamma=1}^N
    \oint_{C_0}
    z^{s_\gamma-s'_\gamma}\,
    \bar w_{\alpha\gamma}(\bolds,{\boldsymbol t} ,
    \bar {\boldsymbol t};z)\,
    \bar w^*_{\gamma\beta}(\bolds',{\boldsymbol t} ,
    \bar {\boldsymbol t};z)\,
    dz.
\label{bil-res-id:t=t'}
\end{multline}
 Substituting the Taylor expansions \eqref{wave-func:element} and
 \eqref{adj-wave-func:element} of
 $w_{\alpha\gamma}(\bolds,{\boldsymbol t} ,\bar {\boldsymbol t};z)$,
 $\bar w_{\gamma\beta}(\bolds,{\boldsymbol t} ,\bar {\boldsymbol t};z)$,
 $w^*_{\alpha\gamma}(\bolds,{\boldsymbol t} ,
 \bar {\boldsymbol t};z)$ and
 $\bar w^*_{\gamma\beta}(\bolds,{\boldsymbol t} ,\bar {\boldsymbol t};z)$,
 we have
\begin{multline}
    \sum_{\gamma=1}^N
    \sum_{\substack{j,k\geq0\\ j+k=s_\gamma-s_\gamma'+1}}
    w_{j,\alpha\gamma}(\bolds,{\boldsymbol t} ,\bar {\boldsymbol t})\,
    w^*_{k,\gamma\beta}(\bolds',{\boldsymbol t} ,\bar {\boldsymbol t})
\\
    =
    \sum_{\gamma=1}^N
    \sum_{\substack{j,k\geq0\\ j+k=-s_\gamma+s_\gamma'-1}}
    \bar w_{j,\alpha\gamma}(\bolds,{\boldsymbol t} ,
    \bar {\boldsymbol t})\,
    \bar w^*_{k,\gamma\beta}(\bolds',{\boldsymbol t} ,
    \bar {\boldsymbol t}).
\label{bil-res-id:t=t':w}
\end{multline}

 On the other hand, the $(\alpha,\beta)$-element of the left hand side
 of the condition \eqref{ULQ=UbarLbarQbar:general} is
 \begin{multline*}
    \left(
    \prod_{\mu=1}^N
    \boldQ_\mu(\bolds)^{a_\mu}
    \sum_{\nu=1}^N
    \boldU_\nu(\bolds) \boldL^{a_\nu}(\bolds)
    \right)_{\alpha\beta}
    =
    \left(
    \hat\boldW(\bolds) \,
    e^{-\sum_\delta a_\delta \der_{s_\delta}}\,
    \diag_\gamma (e^{a_\gamma\der_s})\,
    \hat\boldW^{-1}(\bolds) 
    \right)_{\alpha\beta}
\\
    =
    \sum_{\gamma=1}^N \sum_{j,k=0}^\infty
    w_{j,\alpha\gamma}(\bolds,{\boldsymbol t} ,\bar {\boldsymbol t})\,
    e^{-j\der_s}\,
    e^{-\sum_\delta a_\delta\der_{s_\delta}} e^{a_\gamma\der_s}\,
    e^{-k\der_s}\, 
    w^*_{k,\gamma\beta}(\bolds+\boldone,{\boldsymbol t} ,
    \bar {\boldsymbol t})
\\
    =
    \sum_{\gamma=1}^N \sum_{n=0}^\infty
    \left(\sum_{\substack{j,k\geq0\\j+k=n}}
     w_{j,\alpha\gamma}(\bolds,{\boldsymbol t} ,\bar {\boldsymbol t})\,
     w^*_{k,\gamma\beta}
      (\bolds-(n-a_\gamma-1)\boldone-\bolda,{\boldsymbol t} ,
      \bar {\boldsymbol t})
    \right)
    e^{-(n-a_\gamma)\der_s} 
    e^{-\sum_\delta a_\delta\der_{s_\delta}}
\\
    =
    \sum_{\gamma=1}^N \sum_{n=-a_\gamma}^\infty
    \left(\sum_{\substack{j,k\geq0\\j+k=n+a_\gamma}}
     w_{j,\alpha\gamma}(\bolds,{\boldsymbol t} ,\bar {\boldsymbol t})\,
     w^*_{k,\gamma\beta}
      (\bolds-(n-1)\boldone-\bolda,{\boldsymbol t} ,\bar {\boldsymbol t})
    \right)
    e^{-n\der_s} 
    e^{-\sum_\delta a_\delta\der_{s_\delta}},
\end{multline*}
where $\bolda=\{a_1,\dotsc,a_N\}\in\Integer^N$.
As there are no non-negative $j$ and $k$ such that $j+k=n+a_\gamma$, if
 $n<-a_\gamma$, we may replace the sum over $n$ in the last line by
 $\displaystyle{\sum_{n\in\Integer}}$. In this way we obtain
\begin{multline}
    \left(
    \prod_{\mu=1}^N
    \boldQ_\mu(\bolds)^{a_\mu}
    \sum_{\nu=1}^N
    \boldU_\nu(\bolds) \boldL^{a_\nu}(\bolds)
    \right)_{\alpha\beta}
\\
    =
    \sum_{n\in\Integer} \sum_{\gamma=1}^N 
    \left(\sum_{\substack{j,k\geq0\\j+k=n+a_\gamma}}
     w_{j,\alpha\gamma}(\bolds,{\boldsymbol t} ,\bar {\boldsymbol t})\,
     w^*_{k,\gamma\beta}
      (\bolds-(n-1)\boldone-\bolda,{\boldsymbol t} ,\bar {\boldsymbol t})
    \right)
    e^{-n\der_s} 
    e^{-\sum_\delta a_\delta\der_{s_\delta}}.
\label{ULQ:expansion}
\end{multline}
 Similarly, the $(\alpha,\beta)$-element of the right-hand side
 of \eqref{ULQ=UbarLbarQbar:general} is
\begin{multline}
    \left(
    \prod_{\mu=1}^N
    \boldQbar_\mu(\bolds)^{a_\mu}
    \sum_{\nu=1}^N
    \boldUbar_\nu(\bolds) \boldLbar^{-a_\nu}(\bolds)
    \right)_{\alpha\beta}
\\
    =
    \sum_{m\in\Integer} \sum_{\gamma=1}^N 
    \left(\sum_{\substack{j,k\geq0\\j+k=m-a_\gamma}}
     \bar w_{j,\alpha\gamma}(\bolds,{\boldsymbol t} ,
     \bar {\boldsymbol t})\,
     \bar w^*_{k,\gamma\beta}
      (\bolds+(m+1)\boldone-\bolda,{\boldsymbol t} ,\bar {\boldsymbol t})
    \right)
    e^{m\der_s} 
    e^{-\sum_\delta a_\delta\der_{s_\delta}}.
\label{UbarLbarQbar:expansion}
\end{multline}
 Comparing the coefficients in \eqref{ULQ:expansion} and
 \eqref{UbarLbarQbar:expansion}, we have
\begin{multline}
    \sum_{\gamma=1}^N 
    \sum_{\substack{j,k\geq0\\j+k=n+a_\gamma}}
     w_{j,\alpha\gamma}(\bolds,{\boldsymbol t} ,\bar {\boldsymbol t})\,
     w^*_{k,\gamma\beta}
      (\bolds-(n-1)\boldone-\bolda,{\boldsymbol t} ,\bar {\boldsymbol t})
\\
    =
    \sum_{\gamma=1}^N 
    \sum_{\substack{j,k\geq0\\j+k=-n-a_\gamma}}
     \bar w_{j,\alpha\gamma}(\bolds,{\boldsymbol t} ,\bar {\boldsymbol t})\,
     \bar w^*_{k,\gamma\beta}
      (\bolds-(n-1)\boldone-\bolda,{\boldsymbol t} ,\bar {\boldsymbol t}).
\label{ULQ=UbarLbarQbar:element}
\end{multline}
 If we take $n$ and $\bolda$ such that
 $\bolds'=\bolds-(n-1)\boldone-\bolda$, then
 $s_\gamma-s'_\gamma+1=n+a_\gamma$. Hence the equation
 \eqref{ULQ=UbarLbarQbar:element} gives the bilinear identity
 \eqref{bil-res-id:t=t':w}, which is equivalent to
 \eqref{bil-res-id:t=t'}. 
 Thus we have proved the bilinear identity \eqref{bil-res-id:w},
 i.e. \eqref{bil-res-id:mat}.

\medskip
 (ii)
 Assume that the bilinear identity \eqref{bil-res-id:mat}, which is
 equivalent to \eqref{bil-res-id:w}, holds. Setting ${\boldsymbol t} '={\boldsymbol t} $,
 $\bar {\boldsymbol t}'=\bar {\boldsymbol t}$ and 
 $\bolds'=\bolds+(1-n)\boldone$
 ($n\in\Integer_{>0}$) in \eqref{bil-res-id:w}, we have
\begin{multline*}
    \sum_{\gamma=1}^N
    \oint_{C_\infty}
    z^{n-1}\,
    w_{\alpha\gamma}(\bolds,{\boldsymbol t} ,\bar {\boldsymbol t};z)\,
    w^*_{\gamma\beta}(\bolds+(1-n)\boldone,{\boldsymbol t} ,
    \bar {\boldsymbol t};z)\,
    dz
\\
    =
    \sum_{\gamma=1}^N
    \oint_{C_\infty} 
    z^{n-1}\,
    \bar w_{\alpha\gamma}(\bolds,{\boldsymbol t} ,
    \bar {\boldsymbol t};z)\,
    \bar w^*_{\gamma\beta}(\bolds+(1-n)\boldone,{\boldsymbol t} ',
    \bar {\boldsymbol t}';z)\,
    dz.
\end{multline*}
 The integrand in the right-hand side is a series of $z$ with
 non-negative powers. Hence we obtain
\begin{equation}
    \sum_{\gamma=1}^N
    \sum_{\substack{j,k\geq0\\ j+k=-n}}
    w_{j,\alpha\gamma}(\bolds,{\boldsymbol t} ,\bar {\boldsymbol t})\,
    w^*_{k,\gamma\beta}(\bolds+(1-n)\boldone,{\boldsymbol t} ,
    \bar {\boldsymbol t})
    =0.
\label{bil-res-id:t'=t:s'=s+(1-n)1}
\end{equation}
 On the other hand, the $(\alpha,\beta)$-element of the product of
 matrix difference operators
$\displaystyle{
    \hat\boldW(\bolds,{\boldsymbol t} ,\bar {\boldsymbol t})
    =
    \sum_{j=0}^\infty w_j(\bolds,{\boldsymbol t} ,
    \bar {\boldsymbol t})\, e^{-j\der_s}}
$
 and
$\displaystyle{
    \hat\boldW^*(\bolds,{\boldsymbol t} ,\bar {\boldsymbol t})
    =
    \sum_{k=0}^\infty
    e^{-k\der_s}\, w^*_k(\bolds+\boldone,{\boldsymbol t} ,
    \bar{\boldsymbol t})}
$
 is
\begin{equation}
    \sum_{n=0}^\infty
    \left(
    \sum_{\gamma=1}^N
    \sum_{\substack{j,k\geq0\\ j+k=-n}}
    w_{j,\alpha\gamma}(\bolds,{\boldsymbol t} ,\bar {\boldsymbol t})\,
    w^*_{k,\gamma\beta}(\bolds+(1-n)\boldone,{\boldsymbol t} ,
    \bar {\boldsymbol t})
    \right) e^{-n\der_s}.
\label{WW*:expansion}
\end{equation}
 Each coefficient of $e^{-n\der_s}$ for non-zero $n$ vanishes because of
 \eqref{bil-res-id:t'=t:s'=s+(1-n)1}. Therefore the matrix difference
 operator $\hat\boldW(\bolds,{\boldsymbol t} ,\bar {\boldsymbol t})\,
 \hat\boldW^*(\bolds,{\boldsymbol t} ,\bar {\boldsymbol t})$ acts just by multiplication
 by a matrix. Actually, since $w_0=w_0^*=1_N$ by assumption, the product
 should be $1_N$. Namely, the inverse of the operator
 $\hat\boldW(\bolds,{\boldsymbol t} ,
 \bar {\boldsymbol t})$ defined as \eqref{wave-mat} from
 the coefficients of $\boldPsi$ has the expression as in
 \eqref{inv-wave-mat}.


 If we set ${\boldsymbol t} '={\boldsymbol t} $, 
 $\bar {\boldsymbol t}'=\bar {\boldsymbol t}$ and
 $\bolds'=\bolds+(1+n)\boldone$ ($n\in\Integer_{>0}$) in
 \eqref{bil-res-id:w}, then we obtain
\begin{equation}
    0 =
    \sum_{\gamma=1}^N
    \sum_{\substack{j,k\geq0\\ j+k=n}}
    \bar w_{j,\alpha\gamma}(\bolds,{\boldsymbol t} ,\bar {\boldsymbol t})\,
    \bar w^*_{k,\gamma\beta}(\bolds+(1+n)\boldone,{\boldsymbol t} ,
    \bar {\boldsymbol t}).
\label{bil-res-id:t'=t:s'=s+(1+n)1}
\end{equation}
 by a similar argument as above. 

 The $(\alpha,\beta)$-element of the product of
 matrix difference operators
$$
    \hat\boldWbar(\bolds,{\boldsymbol t} ,\bar {\boldsymbol t})
    =
    \sum_{j=0}^\infty
    \bar w_j(\bolds,{\boldsymbol t} ,\bar {\boldsymbol t})\, e^{j\der_s}
$$
 and
$$
    \hat\boldW^*(\bolds,{\boldsymbol t} ,\bar {\boldsymbol t})
    =
    \sum_{k=0}^\infty
    e^{k\der_s}\, \bar w^*_k(\bolds+\boldone,{\boldsymbol t} ,
    \bar {\boldsymbol t})
$$
 is
\begin{equation*}
    \sum_{n=0}^\infty
    \left(
    \sum_{\gamma=1}^N
    \sum_{\substack{j,k\geq0\\ j+k=n}}
    \bar w_{j,\alpha\gamma}(\bolds,{\boldsymbol t} ,\bar {\boldsymbol t})\,
    \bar w^*_{k,\gamma\beta}(\bolds+(1+n)\boldone,{\boldsymbol t} ,
    \bar {\boldsymbol t})
    \right) e^{n\der_s}.
\end{equation*}
 This and \eqref{bil-res-id:t'=t:s'=s+(1+n)1} imply that the
 matrix difference operator $\hat\boldWbar(\bolds,{\boldsymbol t} ,
 \bar {\boldsymbol t})\,
 \hat\boldWbar^*(\bolds,{\boldsymbol t} ,
 \bar {\boldsymbol t})$ is a multiplication operator
 by a matrix, whose $(\alpha,\beta)$-element is equal to 
\[
    \sum_{\gamma=1}^N
    \bar w_{j,\alpha\gamma}(\bolds,{\boldsymbol t} ,
    \bar {\boldsymbol t})\,
    \bar w^*_{k,\gamma\beta}(\bolds+\boldone,{\boldsymbol t} ,
    \bar {\boldsymbol t})
\]
 by the above computation. The bilinear identity
 \eqref{bil-res-id:w} with ${\boldsymbol t} '={\boldsymbol t} $,
 $\bar {\boldsymbol t}'=\bar {\boldsymbol t}$ and 
 $\bolds=\bolds+\boldone$ yields:
\[
    \sum_{\gamma=1}^N
    w_{0,\alpha\gamma}(\bolds,{\boldsymbol t} ,\bar {\boldsymbol t})\,
    w^*_{0,\gamma\beta}(\bolds+\boldone,{\boldsymbol t} ,
    \bar {\boldsymbol t})
    =
    \sum_{\gamma=1}^N
    \bar w_{0,\alpha\gamma}(\bolds,{\boldsymbol t} ,
    \bar {\boldsymbol t})\,
    \bar w^*_{0,\gamma\beta}(\bolds+\boldone,{\boldsymbol t} ,
    \bar {\boldsymbol t}).
\]
 Its right-hand side is the $(\alpha,\beta)$-element of
 $\hat\boldWbar(\bolds,{\boldsymbol t} ,\bar {\boldsymbol t})\,
 \hat\boldWbar^*(\bolds,{\boldsymbol t} ,\bar {\boldsymbol t})$, as we have just shown and
 its left-hand side is the $(\alpha,\beta)$-element of the coefficient
 of $e^{0\, \der_s}$ in \eqref{WW*:expansion}. As \eqref{WW*:expansion}
 is nothing but $\hat\boldW(\bolds,{\boldsymbol t} ,
 \bar {\boldsymbol t})\,
 \hat\boldW^{-1}(\bolds,{\boldsymbol t} ,
 \bar {\boldsymbol t}) = 1_N$, we obtain
 $\hat\boldWbar(\bolds,{\boldsymbol t} ,\bar {\boldsymbol t})\,
 \hat\boldWbar^*(\bolds,{\boldsymbol t} ,
 \bar {\boldsymbol t}) = 1_N$. Namely, the inverse of
 the operator $\hat\boldWbar(\bolds,{\boldsymbol t} ,
 \bar {\boldsymbol t})$ defined as
 \eqref{wave-mat} from the coefficients of $\boldPsibar$ has the
 expression as in \eqref{inv-wave-mat}.

\medskip
 A similar consideration with $\bolds'=\bolds+(1-n)\boldone$
 ($n\in\Integer$) in the bilinear identity \eqref{bil-res-id:w}
 leads to the bilinear identity for matrix difference operators:
\begin{equation}
    \boldW(\bolds,{\boldsymbol t} ,\bar {\boldsymbol t})\,
    \boldW^{-1}(\bolds,{\boldsymbol t} ',\bar {\boldsymbol t}')
    =
    \boldWbar(\bolds,{\boldsymbol t} ,\bar {\boldsymbol t})\,
    \boldWbar^{-1}(\bolds,{\boldsymbol t} ',\bar {\boldsymbol t}').
\label{bil-id:mat}
\end{equation}
 Differentiating this equation by $t_{\alpha,n}$ and then setting
 ${\boldsymbol t} '={\boldsymbol t} $, 
 $\bar {\boldsymbol t}'=\bar {\boldsymbol t}$, we have
\[
    \frac{\der\hat\boldW}{\der t_{\alpha,n}}(\bolds,{\boldsymbol t} ,
    \bar {\boldsymbol t})
    \hat\boldW^{-1}(\bolds,{\boldsymbol t} ,
    \bar {\boldsymbol t})
    +
    \hat\boldW(\bolds,{\boldsymbol t} ,\bar {\boldsymbol t})\, 
    e^{\der_s}\, E_\alpha
    \hat\boldW^{-1}(\bolds,{\boldsymbol t} ,\bar {\boldsymbol t})
    =
    \frac{\der\hat\boldWbar}{\der t_{\alpha,n}}
         (\bolds,{\boldsymbol t} ,\bar {\boldsymbol t})
    \hat\boldWbar^{-1}(\bolds,{\boldsymbol t} ,\bar {\boldsymbol t}).
\]
 The left-hand side is a sum of $e^{k\der_s}$ ($k\leq n$) with matrix
 coefficients and the right-hand side is a sum of $e^{k\der_s}$ ($k\geq
 0$) with matrix coefficients. Thus we obtain the first equations in
 \eqref{lin-eq:dWhat} and \eqref{lin-eq:dWbarhat} by the standard
 argument. 

 The second equations in \eqref{lin-eq:dWhat} and
 \eqref{lin-eq:dWbarhat} are proved similarly by differentiating
 \eqref{bil-id:mat} by $\bar t_{\alpha,n}$. 

 The Lax equations \eqref{lax:L,Lbar,U,Ubar,Q,Qbar} are direct
 consequences of the definitions \eqref{W->L,U,P} of $\boldL$,
 $\boldLbar$, $\boldU_\alpha$, $\boldUbar_\alpha$, $\boldQ_\alpha$,
 and $\boldQbar_\alpha$ ($\alpha=1,\dotsc,N$) and the linear equations
 \eqref{lin-eq:dWhat} and \eqref{lin-eq:dWbarhat} for $\hat\boldW$ and
 $\hat\boldWbar$.

 The algebraic conditions \eqref{commutative:L,U,Q}, \eqref{U:alg-cond},
 \eqref{Q:alg-cond}, \eqref{commutative:Lbar,Ubar,Qbar},
 \eqref{Ubar:alg-cond} and \eqref{Qbar:alg-cond} are trivially satisfied
 because of \eqref{W->L,U,P}. 

 The remaining algebraic constraint \eqref{ULQ=UbarLbarQbar:general} has
 already been shown to be equivalent to the bilinear identity
 with $\bolds'=\bolds- (n-1)\boldone-\bolda$ ($n\in\Integer$),
 ${\boldsymbol t} '={\boldsymbol t} $ and 
 $\bar {\boldsymbol t}'=\bar {\boldsymbol t}$ in the proof of the
 statement (i) of the proposition.

 Hence the sextet $(\boldL,\boldLbar,\boldU_\alpha,\boldUbar_\alpha, 
 \boldQ_\alpha,\boldQbar_\alpha)_{\alpha=1,\dotsc,N}$ is a solution of
 the multi-component Toda lattice hierarchy.
\end{proof}

\section{The tau-function}
\label{section:tau-function}

\subsection{Existence of the tau-function}
\label{subsection:existence}

The aim of this subsection is to show that the bilinear identity 
for the wave functions for the mul\-ti\-com\-po\-nent
Toda lattice hierarchy
implies existence of the tau-function.
Let $\boldPsi (\bolds , {\boldsymbol t} , \bar {\boldsymbol t} ;z)$,
$\boldPsibar (\bolds , {\boldsymbol t} , \bar {\boldsymbol t} ;z)$ be the 
$N\times N$ matrix wave functions for the $N$-component
Toda hierarchy
and $\boldPsi^* (\bolds , {\boldsymbol t} , \bar {\boldsymbol t} ;z)$,
$\boldPsibar^* (\bolds , {\boldsymbol t} , \bar {\boldsymbol t} ;z)$ 
the adjoint wave functions.
The bilinear relation for the wave functions has the form 
(\ref{bil-res-id:mat}) (see Proposition \ref{predl:bil-res}).
Changing the integration variable $z$ in the right-hand side 
as $z\to z^{-1}$, we can write it in the form
\beq\label{e1}
\sum_{\gamma =1}^N \oint_{C_{\infty}}\!\!
\Psi_{\alpha \gamma}(\bolds , {\boldsymbol t} , \bar {\boldsymbol t} ;z)
\Psi^*_{\gamma \beta} (\bolds ', {\boldsymbol t} ', 
\bar {\boldsymbol t} ';z)dz
=
\sum_{\gamma =1}^N \oint_{C_{\infty}}\!\!
\bar \Psi_{\alpha \gamma}(\bolds , {\boldsymbol t} , 
\bar {\boldsymbol t} ;z^{-1})
\bar \Psi^*_{\gamma \beta} (\bolds ', {\boldsymbol t} ', 
\bar {\boldsymbol t} ';z^{-1})
z^{-2}dz.
\eeq
This identity is valid for all $\alpha , \beta$, 
${\boldsymbol t} $, ${\boldsymbol t} '$, $\bar {\boldsymbol t} $,
$\bar {\boldsymbol t} '$, $\bolds $, $\bolds '$. 

This subsection is devoted to the proof of the following theorem:

\begin{theorem}
\label{theorem:tau-function}
The bilinear identity (\ref{e1}) implies that there exists an
$N\times N$ 
matrix-valued function 
$
    \bigl(
    \tau_{\alpha \beta}(\bolds , {\boldsymbol t} , \bar {\boldsymbol t})
    \bigr)_{\alpha,\beta=1,\dotsc,N}
$
such that the diagonal elements 
$\tau_{\alpha \alpha}(\bolds , {\boldsymbol t} , 
\bar {\boldsymbol t} )
=:\tau (\bolds , {\boldsymbol t} , \bar {\boldsymbol t} )$ are all the same
and such that the wave functions and adjoint wave
functions are expressed through it as
\beq\label{e2}
\begin{array}{l}
\displaystyle{
\Psi_{\alpha \beta}(\bolds , {\boldsymbol t} , \bar {\boldsymbol t} ;z)=
z^{s_{\beta}+\delta_{\alpha \beta}-1}
e^{\xi ({\boldsymbol t} _{\beta}, z)}\frac{\tau_{\alpha \beta}
(\bolds , {\boldsymbol t} -[z^{-1}]_{\beta}, \bar {\boldsymbol t} )}{\tau
(\bolds , {\boldsymbol t} , \bar {\boldsymbol t} )},}
\\ \\
\displaystyle{
\Psi^*_{\alpha \beta}(\bolds , {\boldsymbol t} , \bar {\boldsymbol t} ;z)=
(-1)^{\delta_{\alpha \beta}-1}
z^{-s_{\alpha}+\delta_{\alpha \beta}-1}
e^{-\xi ({\boldsymbol t} _{\alpha}, z)}\frac{\tau_{\alpha \beta}
(\bolds , {\boldsymbol t} +[z^{-1}]_{\alpha}, \bar {\boldsymbol t} )}{\tau
(\bolds , {\boldsymbol t} , \bar {\boldsymbol t} )},}
\\ \\
\displaystyle{
\bar \Psi_{\alpha \beta}(\bolds , {\boldsymbol t} , 
\bar {\boldsymbol t} ;z^{-1})=(-1)^{\delta_{\alpha \beta}-1}
z^{-s_{\beta}}
e^{\xi (\bar {\boldsymbol t} _{\beta}, z)}\frac{\tau_{\alpha \beta}
(\bolds +[1]_{\beta}, {\boldsymbol t} , 
\bar {\boldsymbol t} -[z^{-1}]_{\beta})}{\tau
(\bolds , {\boldsymbol t} , \bar {\boldsymbol t} )},}
\\ \\
\displaystyle{
\bar \Psi^*_{\alpha \beta}(\bolds , {\boldsymbol t} , 
\bar {\boldsymbol t} ;z^{-1})= 
z^{s_{\alpha}}
e^{-\xi (\bar {\boldsymbol t} _{\alpha}, z)}\frac{\tau_{\alpha \beta}
(\bolds -[1]_{\alpha}, {\boldsymbol t} , \bar {\boldsymbol t} +
[z^{-1}]_{\alpha})}{\tau
(\bolds , {\boldsymbol t} , \bar {\boldsymbol t} )},}
\end{array}
\eeq
where
\beq\label{st1}
\left ({\boldsymbol t} \pm [z^{-1}]_{\gamma}\right )_{\alpha j}=t_{\alpha , j}\pm
\delta_{\alpha \gamma} \frac{z^{-j}}{j}.
\eeq
and
\beq\label{st1a}
\bolds \pm [1]_{\alpha} =\{s_1, \ldots , s_{\alpha -1}, s_{\alpha}\pm 1,
s_{\alpha +1}, \ldots , s_N\}.
\eeq
\end{theorem}

\noindent
The matrix-valued function 
$
    \bigl(
    \tau_{\alpha \beta}(\bolds , {\boldsymbol t} , \bar {\boldsymbol t})
    \bigr)_{\alpha,\beta=1,\dotsc,N}
$
is called the {\it tau-function} of the $N$-component Toda lattice
hierarchy.

Now we proceed to the proof of Theorem \ref{theorem:tau-function}.
Let us represent the wave functions in the form
\beq\label{e3}
\begin{array}{l}
\displaystyle{
\Psi_{\alpha \beta}(\bolds , {\boldsymbol t} , \bar {\boldsymbol t} ;z)=
z^{s_{\beta}}
e^{\xi ({\boldsymbol t} _{\beta}, z)}
w_{\alpha \beta}(\bolds , {\boldsymbol t},
\bar {\boldsymbol t} ;z),}
\\ \\
\displaystyle{
\Psi^*_{\alpha \beta}(\bolds , {\boldsymbol t} , 
\bar {\boldsymbol t} ;z)=
z^{-s_{\alpha}}
e^{-\xi ({\boldsymbol t} _{\alpha}, z)}
w^*_{\alpha \beta}(\bolds , {\boldsymbol t} ,
\bar {\boldsymbol t} ;z),}
\\ \\
\displaystyle{
\bar \Psi_{\alpha \beta}(\bolds , {\boldsymbol t} , 
\bar {\boldsymbol t} ;z^{-1})=
z^{-s_{\beta}}
e^{\xi (\bar {\boldsymbol t} _{\beta}, z)}
\bar w_{\alpha \beta}(\bolds , {\boldsymbol t} ,
\bar {\boldsymbol t} ;z^{-1}),}
\\ \\
\displaystyle{
\bar \Psi^*_{\alpha \beta}(\bolds , {\boldsymbol t} , 
\bar {\boldsymbol t} ;z^{-1})= 
z^{s_{\alpha}}
e^{-\xi (\bar {\boldsymbol t} _{\alpha}, z)}
\bar w^*_{\alpha \beta}(\bolds , {\boldsymbol t} ,
\bar {\boldsymbol t} ;z^{-1}),}
\end{array}
\eeq
where all the $w$-functions are assumed to be regular 
as $z\to \infty$.
Substituting this into (\ref{e1}), we write the bilinear
identity for the wave functions in the form
\beq\label{e4}
\begin{array}{l}
\displaystyle{
\sum_{\gamma =1}^N 
\oint_{C_{\infty}}
\!\! z^{s_{\gamma}-s_{\gamma}'}
e^{\xi ({\boldsymbol t} _{\gamma}-{\boldsymbol t} _{\gamma}', z)}
w_{\alpha \gamma}(\bolds , {\boldsymbol t} , \bar {\boldsymbol t} ;z)
w^*_{\gamma \beta}(\bolds ', {\boldsymbol t} ', \bar {\boldsymbol t} ';z)dz}
\\ \\
\displaystyle{
=\, \sum_{\gamma =1}^N
\oint_{C_{\infty}}
\!\, z^{s_{\gamma}'-s_{\gamma}-2}
e^{\xi (\bar {\boldsymbol t} _{\gamma}-\bar {\boldsymbol t} _{\gamma}', z)}}
\bar w_{\alpha \gamma}(\bolds , {\boldsymbol t} , 
\bar {\boldsymbol t} ;z^{-1})
\bar w^*_{\gamma \beta}(\bolds ', {\boldsymbol t} ', 
\bar {\boldsymbol t} ';z^{-1})dz.
\end{array}
\eeq
The $w$-functions are normalized as
$$
w(\bolds , {\boldsymbol t} , 
\bar {\boldsymbol t} ;\infty )=
w^*(\bolds , {\boldsymbol t} , \bar {\boldsymbol t} ;\infty )=1_N,
$$
i.e., their non-diagonal elements vanish and 
$
w_{\alpha \alpha}(\bolds , {\boldsymbol t} , \bar {\boldsymbol t} ;\infty )=
w^*_{\alpha \alpha}(\bolds , {\boldsymbol t} , \bar {\boldsymbol t} ;\infty )=1,
\quad \alpha =1, \ldots , N
$
by \eqref{wave-mat} and \eqref{inv-wave-mat}.
Taking this into account,  
we introduce the functions $\tilde w_{\alpha \beta}$,
$\tilde w^*_{\alpha \beta}$ by extracting the vanishing 
$z$-dependent factor
explicitly\footnote{
The functions
$\tilde w_{\alpha \beta}=
\tilde w_{\alpha \beta}(\bolds , {\boldsymbol t} , 
\bar {\boldsymbol t} ;z)$ introduced here should not be mixed
with
$\tilde w_0$ from \propref{proposition:tildew0}.
}:
$$
w_{\alpha \beta}(\bolds , {\boldsymbol t} , \bar {\boldsymbol t} ;z)=
z^{\delta_{\alpha \beta}-1}
\tilde w_{\alpha \beta}(\bolds , {\boldsymbol t} , \bar {\boldsymbol t} ;z),
\quad
w^{*}_{\alpha \beta}(\bolds , {\boldsymbol t} , \bar {\boldsymbol t} ;z)=
z^{\delta_{\alpha \beta}-1}
\tilde w^{*}_{\alpha \beta}(\bolds , {\boldsymbol t} , 
\bar {\boldsymbol t} ;z).
$$

Let us first set $\bolds '=\bolds $, 
$\bar {\boldsymbol t} '=\bar {\boldsymbol t} $ in (\ref{e4}). 
Then the right-hand side
vanishes and we get
\beq\label{e5}
\sum_{\gamma =1}^N 
\oint_{C_{\infty}}
\!\! z^{\delta_{\alpha \gamma}+\delta_{\beta \gamma}-2}
e^{\xi ({\boldsymbol t} _{\gamma}-{\boldsymbol t} _{\gamma}', z)}
\tilde w_{\alpha \gamma}(\bolds , {\boldsymbol t} , \bar {\boldsymbol t} ;z)
\tilde w^*_{\gamma \beta}(\bolds , {\boldsymbol t} ', 
\bar {\boldsymbol t} ;z^{-1})dz=0.
\eeq
To simplify the notation, we temporally will not write 
the arguments $\bolds , \bar {\boldsymbol t} $ 
explicitly since they are fixed
in (\ref{e5}). 
Set ${\boldsymbol t} -{\boldsymbol t} '=[a^{-1}]_{\mu}$, then
$$
e^{\xi ({\boldsymbol t} _{\gamma}-{\boldsymbol t} _{\gamma}', z)}=
\Bigl (\frac{a}{a-z}\Bigr )^{\delta_{\gamma \mu}}.
$$
Putting $\mu =\alpha$ or
$\mu =\beta$, we have from (\ref{e5}) for $\alpha \neq \beta$:
\beq\label{m10}
\begin{array}{c}
\displaystyle{\oint_{C_{\infty}} z^{-1}\frac{a}{a-z}\,
\tilde w_{\alpha \alpha}({\boldsymbol t} , z)\tilde w^{*}_{\alpha \beta}({\boldsymbol t} -
[a^{-1}]_{\alpha}, z)dz}
\\ \\
\displaystyle{+\oint_{C_{\infty}} z^{-1}
\tilde w_{\alpha \beta}({\boldsymbol t} , z)\tilde w^{*}_{\beta \beta}({\boldsymbol t} -
[a^{-1}]_{\alpha}, z)dz}=0,
\end{array}
\eeq
\beq\label{m10a}
\begin{array}{c}
\displaystyle{\oint_{C_{\infty}} z^{-1}\frac{a}{a-z}\,
\tilde w_{\alpha \beta}({\boldsymbol t} , z)\tilde w^{*}_{\beta \beta}({\boldsymbol t} -
[a^{-1}]_{\beta}, z)dz}
\\ \\
\displaystyle{+\oint_{C_{\infty}} z^{-1}
\tilde w_{\alpha \alpha}({\boldsymbol t} , z)\tilde w^{*}_{\alpha \beta}({\boldsymbol t} -
[a^{-1}]_{\beta}, z)dz=0}
\end{array}
\eeq
and
\beq\label{m11}
\oint_{C_{\infty}}dz \frac{a}{a-z}\,
\tilde w_{\alpha \alpha}({\boldsymbol t} , z)\tilde w^{*}_{\alpha \alpha}({\boldsymbol t} -
[a^{-1}]_{\alpha}, z)dz=0
\eeq
for $\beta =\alpha$, where we write only non-vanishing terms of the sum
over $\gamma$. 
The residue calculus applied to (\ref{m10}), (\ref{m10a})
and (\ref{m11}) yields\footnote{When calculating 
the residues one should take into account that
the point $a$ lies {\it outside} of the contour $C_{\infty}$, so,
shrinking the contour to infinity, we have
$
\oint_{C_{\infty}} z^{-1}\frac{a}{a-z}f(z)dz =+2\pi i f(a)
$
rather than $-2\pi i f(a)$ (for functions $f(z)$ regular
in some neighborhood of infinity).}, respectively:
\beq\label{m12}
\tilde w_{\alpha \alpha}({\boldsymbol t} , a)
\tilde w^{*}_{\alpha \beta}({\boldsymbol t} -[a^{-1}]_{\alpha}, a)=
-\tilde w^{*}_{\alpha \beta }({\boldsymbol t} , \infty ),
\eeq
\beq\label{m12a}
\tilde w_{\alpha \beta}({\boldsymbol t} , a)\tilde w^{*}_{\beta 
\beta}({\boldsymbol t} -[a^{-1}]_{\beta}, a)=
-\tilde w^{*}_{\alpha \beta }({\boldsymbol t} -[a^{-1}]_{\beta}, \infty ),
\eeq
\beq\label{m13}
\tilde w_{\alpha \alpha}({\boldsymbol t} , a)
\tilde w^{*}_{\alpha \alpha}({\boldsymbol t} -[a^{-1}]_{\alpha}, a)=1.
\eeq
Tending $a\to \infty$ in (\ref{m12}), (\ref{m12a}), we get
\beq\label{m14}
\tilde w^{*}_{\alpha \beta}({\boldsymbol t} , \infty )=
-w_{\alpha \beta}({\boldsymbol t} , \infty ) \quad \mbox{for $\alpha \neq \beta$}.
\eeq
Using (\ref{m13}), we see from (\ref{m12}), (\ref{m12a}) that
\beq\label{m15}
\begin{array}{l}
\tilde w_{\alpha \beta}({\boldsymbol t} , a)=
\tilde w_{\alpha \beta}({\boldsymbol t} -[a^{-1}]_{\beta}, \infty )
\tilde w_{\beta \beta}({\boldsymbol t} , a),
\\ \\
\displaystyle{\tilde w^{*}_{\alpha \beta}({\boldsymbol t} , a)=-
\frac{\tilde w_{\alpha \beta}({\boldsymbol t} +
[a^{-1}]_{\alpha}, \infty )}{\tilde w_{\alpha \alpha}
({\boldsymbol t} +[a^{-1}]_{\alpha}, a)}}.
\end{array}
\eeq
Next, we set ${\boldsymbol t} -{\boldsymbol t} '=[a^{-1}]_{\alpha}+[b^{-1}]_{\alpha}$ in (\ref{e5}). 
At $\beta =\alpha$ the residue calculus yields
$$
\tilde w_{\alpha \alpha}({\boldsymbol t} , a)\tilde w^{*}_{\alpha \alpha}({\boldsymbol t} -
[a^{-1}]_{\alpha}-[b^{-1}]_{\alpha}, a)=
\tilde w_{\alpha \alpha}({\boldsymbol t} , b)\tilde w^{*}_{\alpha \alpha}({\boldsymbol t} -
[a^{-1}]_{\alpha}-[b^{-1}]_{\alpha}, b).
$$
Taking into account (\ref{m13}), we can write this relation in the form
\beq\label{m16}
\frac{\tilde w_{\alpha \alpha}({\boldsymbol t} , a)}{\tilde w_{\alpha \alpha}
({\boldsymbol t} -[b^{-1}]_{\alpha}, a)}=
\frac{\tilde w_{\alpha \alpha}({\boldsymbol t} , b)}{\tilde w_{\alpha \alpha}
({\boldsymbol t} -[a^{-1}]_{\alpha}, b)}.
\eeq
As is proven in \cite{DJKM83}, it follows from this relation that 
there exists a function $\tau_{\alpha \alpha}({\boldsymbol t} )$ such that
\beq\label{m17}
\tilde w_{\alpha \alpha}({\boldsymbol t} , a)=w_{\alpha \alpha}({\boldsymbol t} , a)=
\frac{\tau_{\alpha \alpha}
({\boldsymbol t} -[a^{-1}]_{\alpha})}{\tau_{\alpha \alpha}({\boldsymbol t} )}.
\eeq
In the next step, we 
set ${\boldsymbol t} -{\boldsymbol t} '=[a^{-1}]_{\alpha}+[b^{-1}]_{\beta}$ 
with $\alpha \neq \beta$ in (\ref{e5}).
The residue calculus yields
$$
\tilde w_{\alpha \alpha}({\boldsymbol t} , a)\tilde w^{*}_{\alpha \beta}({\boldsymbol t} -
[a^{-1}]_{\alpha}-[b^{-1}]_{\beta}, a)=-
\tilde w_{\alpha \beta}({\boldsymbol t} , b)\tilde w^{*}_{\beta \beta}({\boldsymbol t} -
[a^{-1}]_{\alpha}-[b^{-1}]_{\beta}, b),
$$
or, taking into account (\ref{m13}), (\ref{m15}),
\beq\label{m18}
\frac{\tilde w_{\alpha \alpha}({\boldsymbol t} , a)}{\tilde w_{\alpha \alpha }
({\boldsymbol t} -[b^{-1}]_{\beta}, a)}=
\frac{\tilde w_{\beta \beta}({\boldsymbol t} , b)}{\tilde w_{\beta \beta }
({\boldsymbol t} -[a^{-1}]_{\alpha}, b)}.
\eeq
Note that at $\beta =\alpha$ we get (\ref{m16}).

\begin{lem} 
The condition (\ref{m18}) implies that we can choose
$\tau_{\alpha \alpha}$ in (\ref{m17}) 
which does not depend on the 
index $\alpha$:
$\tau_{\alpha \alpha}({\boldsymbol t} )=:\tau ({\boldsymbol t} )$.
\end{lem}

\begin{proof}
In the proof we follow \cite{D97}. 
Denote $f_{\alpha}({\boldsymbol t} , z)=
\log w_{\alpha \alpha}({\boldsymbol t} , z)$, then (\ref{m18}) acquires the form
\beq\label{m18a}
f_{\alpha}({\boldsymbol t} -[b^{-1}]_{\beta}, a)-
f_{\alpha}({\boldsymbol t} , a)
=f_{\beta}({\boldsymbol t} -[a^{-1}]_{\alpha}, b)-
f_{\beta}({\boldsymbol t} , b).
\eeq
Let us introduce the differential operator
$$
\p_{\alpha}(z)=\p_z-\sum_{k\geq 0}z^{-k-1}\p_{t_{\alpha , k}}
$$
and apply $\p_{\alpha}(a)$ to the both sides of (\ref{m18a}).
We get
$$
\p_{\alpha}(a)f_{\alpha}({\boldsymbol t} -[b^{-1}]_{\beta}, a)-
\p_{\alpha}(a)f_{\alpha}({\boldsymbol t} , a)
=\sum_{k\geq 0}a^{-k-1}\p_{t_{\alpha , k}}
f_{\beta}({\boldsymbol t} ,b).  $$ 
Multiply both sides of this equality by $a^i$ and take the
residue: 
\beq\label{m18b} r_{\alpha, i}({\boldsymbol t} )=
r_{\alpha,i}({\boldsymbol t} -[b^{-1}]_{\beta})- 
\p_{t_{\alpha ,i}}f_{\beta}({\boldsymbol t} , b), 
\eeq 
where $r_{\alpha , i}({\boldsymbol t} )=
\mbox{res}_z (z^i \p_{\alpha} (z)
f_{\alpha}({\boldsymbol t} , z))$ and the residue is defined as
coefficient at $z^{-1}$. Writing the same equality with 
indices $\gamma, j$ instead of $\alpha , i$, 
applying $\p_{t_{\gamma, j}}$ to the
former and $\p_{t_{\alpha, i}}$ to the latter and subtracting one from
the other, we get: 
$$ \p_{t_{\gamma, j}}r_{\alpha , i}({\boldsymbol t}
-[b^{-1}]_{\beta})- \p_{t_{\alpha, i}}r_{\gamma , j}({\boldsymbol t}
-[b^{-1}]_{\beta})= \p_{t_{\gamma, j}}r_{\alpha , i}({\boldsymbol t} )-
\p_{t_{\alpha, i}}r_{\gamma , j}({\boldsymbol t} ).  
$$ 
This means that
$\p_{t_{\gamma, j}}r_{\alpha , i}({\boldsymbol t} )- 
\p_{t_{\alpha, i}}r_{\gamma , j}({\boldsymbol t} )=
\mbox{constant}$ by the following lemma.

\begin{lem}
\label{lemma:g=g}
A function $G$ of ${\boldsymbol t} =\{t_1, t_2, t_3, \ldots \}$ obeying
the relation $G({\boldsymbol t} -[a^{-1}])=G({\boldsymbol t} )$
identically in $a$ does not depend on ${\boldsymbol t} $.
\end{lem}

\begin{proof}
 In the proof we follow \cite{DJKM83}.
We have:
$$
G({\boldsymbol t} -[a^{-1}])=\exp 
\Bigl (-\sum_{k\geq 1}a^{-k}\tilde \p_k\Bigr )
G({\boldsymbol t} )=
\sum_{k\geq 0} a^{-k}p_k(-\tilde \p_{\boldsymbol t} )G({\boldsymbol t} ),
$$
where the polynomials $p_k({\bf x})$ (${\bf x}=\{x_1, x_2, \ldots \}
=:(x_n)_n$) 
are defined by the generating function
$$
\exp \Bigl (\sum_{k\geq 1}x_ka^{-k}\Bigr )=
\sum_{k\geq 0}p_k ({\bf x})a^{-k}
$$
and $\tilde \p_{\boldsymbol t} =(\tilde \p_n)_n=(n^{-1}\p_{t_n})_n$.
(The polynomials $p_k({\bf x})$ are the Schur polynomials associated
with one-row Young diagrams, see for example \cite{DJKM83}, 
\S\S 2.3--2.4.) Note that $p_0({\bf x})=1$. Therefore, we have the
condition
$$
\sum_{k\geq 1} a^{-k}p_k(-\tilde \p_{\boldsymbol t} )G({\boldsymbol t} )=0
$$
for all $a$, hence
\beq\label{pk}
p_k(-\tilde \p_{\boldsymbol t} )G({\boldsymbol t} )=0 \quad \mbox{for all $k\geq 1$}.
\eeq
It follows from this condition that 
$\p_{t_k}G({\boldsymbol t} )=0$ for all $k\geq 1$. This can be shown by induction.
At $k=1$ we have $\p_{t_1}G({\boldsymbol t} )=0$ because $p_1({\bf x})=x_1$.
Assume that (\ref{pk}) is true for $k=1, \ldots , k_0$. 
Note that 
$$
p_{k_0+1}({\bf x})=x_{k_0+1}+(\mbox{polynomial of $x_1, \ldots , x_{k_0}$}).
$$
Hence equation
(\ref{pk}) for $k=k_0+1$ has the form
$$
\Bigl (\frac{1}{k_0+1}\, \p_{k_0+1}+
(\mbox{polynomial of $\p_{t_1}, \ldots , \p_{t_{k_0}}$})\Bigr )
G({\boldsymbol t} )=0.
$$
By the induction assumption only the first term in the left-hand side
survives and gives $\p_{t_{k_0+1}}G({\boldsymbol t} )=0$. 
\end{proof}

From the definition of $r_{\alpha , i}$ it follows that the constant 
$\p_{t_{\gamma, j}}r_{\alpha , i}({\boldsymbol t} )- 
\p_{t_{\alpha, i}}r_{\gamma , j}({\boldsymbol t} )$ is
zero.  Therefore, 
$$
    \p_{t_{\gamma, j}}r_{\alpha , i}({\boldsymbol t} )
    =
    \p_{t_{\alpha, i}}r_{\gamma , j}({\boldsymbol t} ),
$$ 
which implies the
existence of a function $\tau ({\boldsymbol t} )$ such that $r_{\alpha ,
i}({\boldsymbol t} )=\p_{t_{\alpha , i}}\log \tau ({\boldsymbol t} )$.
From (\ref{m18b}) we then see that $$ \p_{t_{\alpha ,
i}}f_{\beta}({\boldsymbol t} , b)= \p_{t_{\alpha , i}}\Bigl (\log \tau
({\boldsymbol t} -[b^{-1}]_{\beta})- \log \tau 
({\boldsymbol t} )\Bigr ).  $$ 
Integrating, we get $$ w_{\alpha \alpha}({\boldsymbol t} ,
z)=c(z) \frac{\tau ({\boldsymbol t} -[z^{-1}]_{\alpha})}{\tau
({\boldsymbol t} )}, $$ where the constant $c(z)$ can be eliminated by
multiplying the tau-function by exponent of a linear form in the times.
\end{proof}
As it follows from the lemma, 
we can write (\ref{m17}) in the form
\beq\label{m20}
\tilde w_{\alpha \alpha}({\boldsymbol t} , a)=\frac{\tau
({\boldsymbol t} -[a^{-1}]_{\alpha})}{\tau ({\boldsymbol t} )},
\eeq
where the function $\tau ({\boldsymbol t} )$ already does not depend on 
the index $\alpha$. Plugging (\ref{m20}) into (\ref{m15}), we get
the equations
\beq\label{m21}
\begin{array}{l}
\displaystyle{
\tilde w_{\alpha \beta}({\boldsymbol t} , a)=\tilde 
w_{\alpha \beta}({\boldsymbol t} -[a^{-1}]_{\beta}, \infty )
\, \frac{\tau ({\boldsymbol t} -[a^{-1}]_{\beta})}{\tau ({\boldsymbol t} )},}
\\ \\
\displaystyle{
\tilde w^{*}_{\alpha \beta}({\boldsymbol t} , a)=(-1)^{\delta_{\alpha \beta}-1}
\tilde w_{\alpha \beta}({\boldsymbol t} +[a^{-1}]_{\alpha})
\, \frac{\tau ({\boldsymbol t} +[a^{-1}]_{\alpha})}{\tau ({\boldsymbol t} )}}.
\end{array}
\eeq
In the original notation, they read:
\beq\label{m22}
\begin{array}{l}
\displaystyle{
\tilde w_{\alpha \beta}(\bolds , {\boldsymbol t} , 
\bar {\boldsymbol t} ;a)=
\tilde 
w_{\alpha \beta}(\bolds , {\boldsymbol t} -
[a^{-1}]_{\beta}, \bar {\boldsymbol t} ;\infty )
\frac{\tau (\bolds , {\boldsymbol t} -
[a^{-1}]_{\beta}, \bar {\boldsymbol t} )}{\tau
(\bolds , {\boldsymbol t} , \bar {\boldsymbol t} )},}
\\ \\
\displaystyle{
\tilde w_{\alpha \beta}^*(\bolds , {\boldsymbol t} , 
\bar {\boldsymbol t} ;a)=
(-1)^{\delta_{\alpha \beta}-1}
\tilde w_{\alpha \beta}(\bolds , 
{\boldsymbol t} +[a^{-1}]_{\alpha}, 
\bar {\boldsymbol t} ;\infty )
\frac{\tau (\bolds , {\boldsymbol t} +[a^{-1}]_{\alpha}, 
\bar {\boldsymbol t} )}{\tau
(\bolds , {\boldsymbol t} , \bar {\boldsymbol t} )},}
\end{array}
\eeq
These formulae are of the form (\ref{e2}), i.e.,
\beq\label{m23}
\begin{array}{l}
\displaystyle{
\tilde w_{\alpha \beta}(\bolds , {\boldsymbol t} , \bar {\boldsymbol t} ;a)=
\frac{\tau_{\alpha \beta}(\bolds , {\boldsymbol t} -[a^{-1}]_{\beta}, 
\bar {\boldsymbol t} )}{\tau (\bolds , {\boldsymbol t} , \bar {\boldsymbol t} )}},
\\ \\
\displaystyle{
\tilde w^*_{\alpha \beta}(\bolds , {\boldsymbol t} , \bar {\boldsymbol t} ;a)=
(-1)^{\delta_{\alpha \beta}-1}
\frac{\tau_{\alpha \beta}(\bolds , {\boldsymbol t} +[a^{-1}]_{\alpha}, 
\bar {\boldsymbol t} )}{\tau (\bolds , {\boldsymbol t} , \bar {\boldsymbol t} )}}
\end{array}
\eeq
with
\beq\label{m24}
\tau_{\alpha \beta}(\bolds , {\boldsymbol t} , \bar {\boldsymbol t} )=
\tilde w_{\alpha \beta}(\bolds , {\boldsymbol t} , \bar {\boldsymbol t} , \infty )
\tau (\bolds , {\boldsymbol t} , \bar {\boldsymbol t} ).
\eeq
Note that the matrix-valued function
$
    \bigl(
    \tau_{\alpha \beta}(\bolds,\boldt, \boldtbar)
    \bigr)_{\alpha,\beta=1,\dotsc,N}
$
is defined
up to multiplication by an arbitrary scalar function of $\bolds$ and 
$\bar {\boldsymbol t} $. 

Let us now put $\bolds '=\bolds +[1]_{\alpha}+[1]_{\beta}$, 
${\boldsymbol t} '={\boldsymbol t} $ in
(\ref{e4}). In this case the left-hand side of (\ref{e4}) 
vanishes and we arrive at the relation
\beq\label{e6}
\sum_{\gamma =1}^N 
\oint_{C_{\infty}}
\!\! z^{\delta_{\alpha \gamma}+\delta_{\beta \gamma}-2}
e^{\xi (\bar {\boldsymbol t} _{\gamma}-\bar {\boldsymbol t} _{\gamma}', z)}
\bar w_{\alpha \gamma}(\bolds , {\boldsymbol t} , 
\bar {\boldsymbol t} ;z^{-1})
\bar w^*_{\gamma \beta}(\bolds +[1]_{\alpha}+[1]_{\beta}, 
{\boldsymbol t} , \bar {\boldsymbol t} ';z^{-1})dz=0.
\eeq
The form of 
this equation is very similar to (\ref{e5}). The only difference 
is the normalization of the $\bar w$-functions. 
To make 
equation (\ref{e6}) closer to (\ref{e5}), 
 we rewrite it in an equivalent form:
\beq\label{e6a}
\sum_{\gamma =1}^N 
\oint_{C_{\infty}}
\!\! z^{\delta_{\alpha \gamma}+\delta_{\beta \gamma}-2}
e^{\xi (\bar {\boldsymbol t} _{\gamma}-\bar {\boldsymbol t} _{\gamma}', z)}
\frac{\bar w_{\alpha \gamma}(\bolds -[1]_{\alpha}, 
{\boldsymbol t} , \bar {\boldsymbol t} ;
z^{-1})}{\bar w_{\alpha \alpha}(\bolds -[1]_{\alpha}, 
{\boldsymbol t} , \bar {\boldsymbol t} ;0 )}\,
\frac{\bar w^*_{\gamma \beta}(\bolds +[1]_{\beta}, 
{\boldsymbol t} , \bar {\boldsymbol t} ';z^{-1})}{\bar w^*_{\beta \beta}(\bolds +[1]_{\beta}, 
{\boldsymbol t} , \bar {\boldsymbol t} ';0 )}\, dz=0,
\eeq
in which it can be analyzed 
in the same way as (\ref{e5}). (The
functions
$$
\tilde {\bar w}_{\alpha \alpha}(\bolds , 
{\boldsymbol t} , \bar {\boldsymbol t} ;z^{-1} )=
\frac{\bar w_{\alpha \alpha}(\bolds -[1]_{\alpha}, 
{\boldsymbol t} , \bar {\boldsymbol t} ;
z^{-1})}{\bar w_{\alpha \alpha}(\bolds -[1]_{\alpha}, 
{\boldsymbol t} , \bar {\boldsymbol t} ;0 )}, 
$$
$$
\tilde {\bar w}^{*}_{\beta \beta}(\bolds , 
{\boldsymbol t} , \bar {\boldsymbol t} ';z^{-1})=
\frac{\bar w^*_{\beta \beta}(\bolds +[1]_{\beta}, 
{\boldsymbol t} , \bar {\boldsymbol t} ';z^{-1})}{\bar w^*_{\beta \beta}(\bolds +[1]_{\beta}, 
{\boldsymbol t} , \bar {\boldsymbol t} ';0 )}
$$
are normalized in the same way as the diagonal
elements of the $\tilde w$-functions in 
(\ref{e5}).)
The conclusion is that there exists a function
$\bar \tau (\bolds , {\boldsymbol t} , \bar {\boldsymbol t} )$ such that
\beq\label{e7}
\begin{array}{l}
\displaystyle{
\frac{\bar w_{\alpha \beta}(\bolds , {\boldsymbol t} , 
\bar {\boldsymbol t} ;a^{-1})}{\bar w_{\alpha \alpha}(\bolds , 
{\boldsymbol t} , 
\bar {\boldsymbol t} ;0 )}=
\frac{\bar \tau_{\alpha \beta}(\bolds +[1]_{\alpha}, {\boldsymbol t} ,
\bar {\boldsymbol t} -[a^{-1}]_{\beta})}{\bar \tau 
(\bolds +[1]_{\alpha}, {\boldsymbol t} ,\bar {\boldsymbol t} )}},
\\ \\
\displaystyle{
\frac{\bar w^*_{\alpha \beta}(\bolds , {\boldsymbol t} , 
\bar {\boldsymbol t} ;a^{-1})}{\bar w_{\beta \beta}(\bolds , {\boldsymbol t} , 
\bar {\boldsymbol t} ;0 )}=(-1)^{\delta_{\alpha \beta}-1}
\frac{\bar \tau_{\alpha \beta}(\bolds -[1]_{\beta}, {\boldsymbol t} ,
\bar {\boldsymbol t} +[a^{-1}]_{\alpha})}{\bar \tau 
(\bolds -[1]_{\beta}, {\boldsymbol t} ,\bar {\boldsymbol t} )}}
\end{array}
\eeq
with
\beq\label{e8}
\bar \tau_{\alpha \beta}(\bolds , {\boldsymbol t} , \bar {\boldsymbol t} )=
\frac{\bar w_{\alpha \beta}(\bolds -[1]_{\alpha}, {\boldsymbol t} , 
\bar {\boldsymbol t} ;0 )}{\bar w_{\alpha \beta}
(\bolds -[1]_{\alpha}, {\boldsymbol t} , 
\bar {\boldsymbol t} ;0 )}
\, \bar \tau 
(\bolds , {\boldsymbol t} , \bar {\boldsymbol t} ).
\eeq
These equations do not yet fix the dependence of $\bar \tau$ on 
$\bolds $. We can fix this dependence by setting
\beq\label{e8a}
\begin{array}{l}
\displaystyle{
\bar w_{\alpha \alpha}(\bolds , {\boldsymbol t} , \bar {\boldsymbol t} ; 0 )=
\frac{\bar \tau 
(\bolds +[1]_{\alpha}, {\boldsymbol t} , \bar {\boldsymbol t} )}{\bar \tau 
(\bolds , {\boldsymbol t} , \bar {\boldsymbol t} )},}
\\ \\
\displaystyle{
\bar w^*_{\alpha \alpha}(\bolds , {\boldsymbol t} , \bar {\boldsymbol t} ; 0 )=
\frac{\bar \tau 
(\bolds -[1]_{\alpha}, {\boldsymbol t} , \bar {\boldsymbol t} )}{\bar \tau 
(\bolds , {\boldsymbol t} , \bar {\boldsymbol t} )}.}
\end{array}
\eeq
The consistency of this definition follows from the relation
\beq\label{e8b}
\bar w_{\alpha \alpha}(\bolds , {\boldsymbol t} , \bar {\boldsymbol t} ; 0 )
\bar w^*_{\alpha \alpha}(\bolds +[1]_{\alpha}, {\boldsymbol t} , 
\bar {\boldsymbol t} ; 0 )=1
\eeq
which is a corollary of (\ref{e4}) (one should put 
$\beta =\alpha$, ${\boldsymbol t} '={\boldsymbol t} $, $\bar {\boldsymbol t} '=\bar {\boldsymbol t} $,
$\bolds '=\bolds +[1]_{\alpha}$ and calculate the residues 
at infinity).
Note that the matrix-valued function
$
    \bigl(
    \bar\tau_{\alpha \beta}(\bolds ,\boldt, \boldtbar)
    \bigr)_{\alpha,\beta=1,\dotsc,N}
$
is defined up to 
multiplication by an arbitrary scalar function of ${\boldsymbol t}$. 

It remains to connect the functions $\tau$ and $\bar \tau$. 
To this end, we set $\bolds '=\bolds +[1]_{\beta}$, 
${\boldsymbol t} -{\boldsymbol t} '=[a^{-1}]_{\alpha}$,
$\bar {\boldsymbol t} -\bar {\boldsymbol t} '=[b^{-1}]_{\beta}$ in (\ref{e4}).
Calculating the residues, we obtain:
\beq\label{e9}
\begin{array}{l}
\tilde w_{\alpha \alpha}(\bolds , {\boldsymbol t} , \bar {\boldsymbol t} , a)
\tilde w^*_{\alpha \beta}(\bolds +[1]_{\beta}, {\boldsymbol t} 
-[a^{-1}]_{\alpha}, \bar {\boldsymbol t} -[b^{-1}]_{\beta}, a)
\\ \\
\phantom{aaaaaaaaaaaaa}=
\bar w_{\alpha \beta}(\bolds , {\boldsymbol t} , \bar {\boldsymbol t} , b^{-1})
\bar w^*_{\beta \beta}(\bolds +[1]_{\beta}, {\boldsymbol t} 
-[a^{-1}]_{\alpha}, \bar {\boldsymbol t} -[b^{-1}]_{\beta}, b^{-1})
\end{array}
\eeq
or, in terms of the tau-functions,
\beq\label{e10}
\begin{array}{lll}
&& \displaystyle{
\frac{\tau (\bolds , 
{\boldsymbol t} -[a^{-1}]_{\alpha}, \bar {\boldsymbol t} )}{\tau (\bolds , 
{\boldsymbol t} , \bar {\boldsymbol t} )}\,
\frac{\tau_{\alpha \beta}(\bolds +[1]_{\beta}, {\boldsymbol t} , 
\bar {\boldsymbol t} -[b^{-1}]_{\beta})}{\tau (\bolds +[1]_{\beta},
{\boldsymbol t} -[a^{-1}]_{\alpha}, \bar {\boldsymbol t} -[b^{-1}]_{\beta})}}
\\ \\
&=& (-1)^{\delta_{\alpha \beta}-1}
\displaystyle{
\frac{\bar \tau_{\alpha \beta}(\bolds +[1]_{\alpha}, {\boldsymbol t} , 
\bar {\boldsymbol t} -[b^{-1}]_{\beta})}{\bar \tau (\bolds , 
{\boldsymbol t} , \bar {\boldsymbol t} )}\,
\frac{\bar \tau (\bolds , 
{\boldsymbol t} -[a^{-1}]_{\alpha}, 
\bar {\boldsymbol t} )}{\bar \tau (\bolds +[1]_{\beta},
{\boldsymbol t} -[a^{-1}]_{\alpha}, \bar {\boldsymbol t} -[b^{-1}]_{\beta})}}.
\end{array}
\eeq
Putting here $a=b=\infty$, we arrive at the relation
\beq\label{n1}
(-1)^{\delta_{\alpha \beta}-1}\frac{\bar \tau_{\alpha \beta}
(\bolds +[1]_{\alpha \beta}, {\boldsymbol t} , \bar {\boldsymbol t} )}{\tau_{\alpha \beta}
(\bolds , {\boldsymbol t} , \bar {\boldsymbol t} )}=
\frac{\bar \tau (\bolds , {\boldsymbol t} , \bar {\boldsymbol t} )}{\tau (\bolds , 
{\boldsymbol t} , \bar {\boldsymbol t} )}=: f(\bolds , 
{\boldsymbol t} , \bar {\boldsymbol t} ),
\eeq
where $[1]_{\alpha \beta}=[1]_{\alpha} -[1]_{\beta}$ and a function 
$f$ is the same for all $\alpha , \beta$. In terms of the function $f$
relation (\ref{e10}) acquires the form
\beq\label{n2}
\frac{f(\bolds +[1]_{\beta}, {\boldsymbol t} , \bar {\boldsymbol t} -[b^{-1}]_{\beta})
f(\bolds , {\boldsymbol t} -[a^{-1}]_{\alpha}, 
\bar {\boldsymbol t} )}{f(\bolds +[1]_{\beta}, {\boldsymbol t} -[a^{-1}]_{\alpha}, 
\bar {\boldsymbol t} -[b^{-1}]_{\beta})f(\bolds , {\boldsymbol t} , \bar {\boldsymbol t} )}=1.
\eeq
Putting here $b=\infty$, we have the relation
\beq\label{n3}
\frac{f(\bolds +[1]_{\beta}, {\boldsymbol t} , 
\bar {\boldsymbol t} )}{f(\bolds , {\boldsymbol t} , \bar {\boldsymbol t} )}=
\frac{f(\bolds +[1]_{\beta}, {\boldsymbol t} -[a^{-1}]_{\alpha}, 
\bar {\boldsymbol t} )}{f(\bolds , {\boldsymbol t} -[a^{-1}]_{\alpha}, 
\bar {\boldsymbol t} )}=: g_{\beta}(\bolds , {\boldsymbol t} , \bar {\boldsymbol t} )
\eeq
which means the following condition for the function  
$g_{\beta}(\bolds , {\boldsymbol t} , \bar {\boldsymbol t} )$:
\beq\label{g1}
g_{\beta}(\bolds , {\boldsymbol t} -[a^{-1}]_{\alpha}, \bar {\boldsymbol t} )=
g_{\beta}(\bolds , {\boldsymbol t} , \bar {\boldsymbol t} )
\quad \mbox{for all $\alpha$ and $a$}.
\eeq
The condition (\ref{g1}) makes it possible to apply \lemref{lemma:g=g}
to the function $g_{\beta}(\bolds , {\boldsymbol t}, \bar{\boldsymbol
t})$, namely, $g_{\beta}(\bolds , {\boldsymbol t}, \bar{\boldsymbol
t})$ does not depend on ${\boldsymbol t}$.

Substituting equation (\ref{n3}) back to (\ref{n2}), we get:
\beq\label{n31}
\frac{f(\bolds , {\boldsymbol t} , \bar {\boldsymbol t} -[b^{-1}]_{\beta})
f(\bolds , {\boldsymbol t} -[a^{-1}]_{\alpha}, 
\bar {\boldsymbol t} )}{f(\bolds , {\boldsymbol t} -[a^{-1}]_{\alpha}, 
\bar {\boldsymbol t} -[b^{-1}]_{\beta})f(\bolds , {\boldsymbol t} , \bar {\boldsymbol t} )}=1.
\eeq

\begin{lem}
\label{lemma:ff/ff=1}
If $f(\bolds , {\boldsymbol t} , \bar {\boldsymbol t} )$ satisfies (\ref{n31}) for any
$a,b,\alpha , \beta$, then there exist functions $F(\bolds , {\boldsymbol t} )$,
$\bar F(\bolds , \bar {\boldsymbol t} )$ such that 
$f(\bolds , {\boldsymbol t} , \bar {\boldsymbol t} )=
F(\bolds , {\boldsymbol t} )\bar F(\bolds , \bar {\boldsymbol t} )$.
\end{lem}

\begin{proof}
Since the following argument does not depend on the number of
components, we take $N=1$ in this proof for simplicity of the notation.
We also omit the argument $\bolds $ of the functions since it is the same
for all of them. 

Taking logarithm of (\ref{n31}),
we have:
\begin{multline}
    \log f({\boldsymbol t} -[a^{-1}], \bar {\boldsymbol t}               ) +
    \log f({\boldsymbol t}                 
    ,\bar {\boldsymbol t}-[b^{-1}])
\\
    -
    \log f({\boldsymbol t},
    \bar {\boldsymbol t}               ) -
    \log f({\boldsymbol t} -
    [a^{-1}],\bar {\boldsymbol t}-[b^{-1}]) =0.
\label{log(ff/ff)=0}
\end{multline}
Let us expand the different terms in this equation in a series.
For example, the first term in this equation is expressed as
\[
    \log f({\boldsymbol t} -[a^{-1}],\bar {\boldsymbol t})
    =
    e^{-\sum_{i=1}^\infty a^{-i}\tilde\partial_i}
     \log f({\boldsymbol t} ,\bar {\boldsymbol t})
    =
    \sum_{k=0}^\infty a^{-k} p_k(-\tilde\partial_{\boldsymbol t} )
    \log f({\boldsymbol t} ,\bar {\boldsymbol t}),
\]
where the polynomials $p_k({\bf x})$ (the Schur polynomials)
and the symbol $\tilde\der_{\boldsymbol t} $ are
defined in the proof of Lemma \ref{lemma:g=g}.
Hence, \eqref{log(ff/ff)=0} is expanded as
\[
    \sum_{k=1}^\infty \sum_{l=1}^\infty a^{-k} b^{-l}
    p_k(-\tilde\der_{\boldsymbol t} )\, p_l(-\tilde
    \der_{\bar {\boldsymbol t}})
    \log f({\boldsymbol t} , \bar {\boldsymbol t}) = 0, 
\]
which means that
\begin{equation}
    p_k(-\tilde\der_{\boldsymbol t} )\, p_l(-\tilde
    \der_{\bar {\boldsymbol t}})
    \log f({\boldsymbol t} , \bar {\boldsymbol t}) = 0
\label{p(d)p(d)logf=0}
\end{equation}
for all $k\geq1$, $l\geq1$. 

From this condition we will show that
\begin{equation}
    \frac{\der}{\der t_k}\frac{\der}{\der\bar t_l}\log f = 0
\label{ddlogf=0}
\end{equation}
for any $k\geq1$, $l\geq1$ by induction as follows.
The equation \eqref{p(d)p(d)logf=0} for $k=l=1$ is 
$\der_{t_1}\der_{\bar t_1}\log f = 0$, as $p_1({\bf x})=x_1$.
Assume that \eqref{ddlogf=0} is true for $k=1$, $l=1,\dotsc,l_0$. Note
that 
\[
    p_{l_0+1}({\bf x})
    =
    x_{l_0+1} + (\text{polynomial of }x_1,\dotsc,x_{l_0}).
\]
Hence, equation \eqref{p(d)p(d)logf=0} for $k=1$, $l=l_0+1$ has the
form 
\[
    \frac{\der}{\der t_1}
    \left(
     \frac{1}{l_0+1}\frac{\der}{\der \bar t_{l_0+1}}
     +
     \left(\text{polynomial of }
           \frac{\der}{\der \bar t_1}, \dotsc, 
           \frac{\der}{\der \bar t_{l_0}}
     \right)
    \right)
    \log f
    = 0.
\]
By the induction assumption only the first term in the left hand side
survives and gives $\der_{t_1}\der_{\bar t_{l_0+1}}\log f = 0$.
Thus we have \eqref{ddlogf=0} for $k=1$ and all $l$.
Fixing $l$ and applying the induction for $k$ in the same way, we can prove
\eqref{ddlogf=0} for any $k$.

Assume that $\log f({\boldsymbol t} , 
\bar {\boldsymbol t})$ is expanded into a power
series:
\[
    \log f({\boldsymbol t} , \bar {\boldsymbol t})
    =
    \sum_{{\bf n}, {\bf \bar n}\in\{0,1,2,\dotsc\}^\infty}
    c_{{\bf n}{\bf \bar n}} {\boldsymbol t} ^{\bf n} 
    \bar {\boldsymbol t}^{\bf \bar n}.
\]
Then equation \eqref{ddlogf=0} for $k$ and $l$ means
\[
    \sum_{{\bf n}, {\bf \bar n}\in\{0,1,2,\dotsc\}^\infty}
    n_k {\bar n}_l
    c_{{\bf n}{\bf \bar n}} {\boldsymbol t} ^{\bf n- e_k}
                            \bar {\boldsymbol t}^{\bf \bar n - e_l}
    =
    0,
\]
where ${\bf e_k}=(\delta_{nk})_n$. Hence $c_{\bf n \bar n}=0$ unless
either $n_k$ or $\bar n_l$ is zero. 
Therefore, $c_{\bf n \bar n}=0$ unless ${\bf n=0}$ or ${\bf \bar n=0}$:
\[
 \begin{split}
    \log f({\boldsymbol t} , \bar {\boldsymbol t})
    &=
    c_{\bf 0 0} +
    \sum_{{\bf n}\in\{0,1,2,\dotsc\}^\infty}
    c_{{\bf n  0}} {\boldsymbol t} ^{\bf n}
    +
    \sum_{{\bf \bar n}\in\{0,1,2,\dotsc\}^\infty}
    c_{{\bf 0 \bar n}} \bar {\boldsymbol t}^{\bf \bar n}
\\
    &=
    \log f({\boldsymbol t} , {\bf 0}) + 
    \log f({\bf 0}, \bar {\boldsymbol t}) 
    - \log f({\bf 0}, {\bf 0}),
 \end{split}
\]
 so the factorization of $f$ follows by
 setting, for example,
 $F(\boldt)=f(\boldt,\boldzero)/f(\boldzero,\boldzero)$ and
 $\bar F(\boldtbar)=f(\boldzero,\boldtbar)$. 
\end{proof}

\noindent
Recall that we have 
$g_{\beta}(\bolds , {\boldsymbol t} , \bar {\boldsymbol t} )=
g_{\beta}(\bolds , \bar {\boldsymbol t} )$. 
According to \lemref{lemma:ff/ff=1}, we can write
$$
\frac{F(\bolds +[1]_{\beta}, \boldt )}{F(\bolds , \boldt )}=
\frac{\bar F(\bolds , \bar \boldt )g_{\beta}(\bolds , \bar {\boldsymbol t})}{\bar F(\bolds +[1]_{\beta}, 
\bar \boldt )}
$$
where $F$ is the function 
introduced by \lemref{lemma:ff/ff=1}.
Since the right-hand side does not depend on $\boldt$, 
the same must be true for the left-hand side.
Let us denote the function in the left-hand side by 
$h_\beta(\bolds)$. Then we obtain a
recursion relation 
$
    F(\bolds+[1]_\beta,\boldt) = h_\beta(\bolds) F(\bolds,\boldt)
$.
Therefore, $F(\bolds,\boldt)$ can be expressed as a product of several
factors of the form $h_\beta(\bolds')$ (they depend on $\bolds$) and
$F(\boldzero,\boldt)$ for any $\bolds$.
Thus
we conclude that
the function $F$ factorizes into a product of a function of $\bolds$
and a function of $\boldt$.
Since the function of $\bolds$ can be included into $\bar F(\bolds ,
\bar \boldt )$, we are free to consider $F$ as a function 
only of $\boldt$. Therefore,
$f(\bolds , {\boldsymbol t} , \bar {\boldsymbol t} )$ 
can be represented as a product of a function 
of ${\boldsymbol t} $ and a function of $\bolds , \bar {\boldsymbol t} $:
$f(\bolds , {\boldsymbol t} , \bar {\boldsymbol t} )=F({\boldsymbol t} )\bar F(\bolds , 
\bar {\boldsymbol t} )$. 
Thus the functions $\tau$ and $\bar \tau$ are connected
by the relation
$$
\bar \tau (\bolds , {\boldsymbol t} , \bar {\boldsymbol t} )=
F({\boldsymbol t} )\bar F(\bolds , \bar {\boldsymbol t} )
\tau (\bolds , {\boldsymbol t} , \bar {\boldsymbol t} ).
$$
However, the possible 
factors $F({\boldsymbol t} )\bar F(\bolds , \bar {\boldsymbol t} )$
just reflect the freedom in the choice of the functions 
$\tau$ and $\bar \tau$ mentioned above. That is why we can put
$F({\boldsymbol t} )=\bar F(\bolds , \bar {\boldsymbol t} )=1$ without loss of
generality. 
Then from (\ref{n1}) we see that 
the relation between the tau-functions looks as follows:
\beq\label{n4}
\bar \tau_{\alpha \beta}(\bolds , {\boldsymbol t} , \bar {\boldsymbol t} )=
(-1)^{\delta_{\alpha \beta}-1}
\tau_{\alpha \beta}(\bolds +[1]_{\beta \alpha}, {\boldsymbol t} , \bar {\boldsymbol t} ).
\eeq
Therefore, we have:
\beq\label{n5}
\begin{array}{l}
\displaystyle{
\bar w_{\alpha \beta}(\bolds , {\boldsymbol t} , 
\bar {\boldsymbol t} ;z^{-1})=(-1)^{\delta_{\alpha \beta}-1}\,
\frac{\tau_{\alpha \beta}(\bolds +[1]_{\beta}, {\boldsymbol t} , \bar {\boldsymbol t} -
[z^{-1}]_{\beta})}{\tau (\bolds , {\boldsymbol t} , \bar {\boldsymbol t} )},}
\\ \\
\displaystyle{
\bar w^*_{\alpha \beta}(\bolds , {\boldsymbol t} , 
\bar {\boldsymbol t} ;z^{-1})=
\frac{\tau_{\alpha \beta}(\bolds -[1]_{\alpha}, {\boldsymbol t} , \bar {\boldsymbol t} +
[z^{-1}]_{\alpha})}{\tau (\bolds , {\boldsymbol t} , 
\bar {\boldsymbol t} )},}
\end{array}
\eeq
which agrees with (\ref{e2}), (\ref{e3}). 
This concludes the proof of existence
of the tau-function for the multi-component Toda hierarchy.

\begin{theorem}
\label{theorem:tautau}
The tau-function of the multi-component Toda lattice hierarchy
satisfies the matrix bilinear equation
\beq\label{n6}
\begin{array}{l}
\displaystyle{
\sum_{\gamma =1}^N (-1)^{\delta_{\beta \gamma}}\oint_{C_{\infty}}
z^{s_{\gamma}-s'_{\gamma}+\delta_{\alpha \gamma}+\delta_{\beta \gamma}
-2}e^{\xi ({\boldsymbol t} _{\gamma}-{\boldsymbol t} '_{\gamma},z)}
\tau_{\alpha \gamma}(\bolds , {\boldsymbol t} -[z^{-1}]_{\gamma}, \bar {\boldsymbol t} )
\tau_{\gamma \beta}(\bolds ', {\boldsymbol t} '+[z^{-1}]_{\gamma}, \bar {\boldsymbol t} ')
\, dz}
\\ \\
\displaystyle{
=\sum_{\gamma =1}^N (-1)^{\delta_{\alpha \gamma}}\oint_{C_{\infty}}
z^{s_{\gamma}'-s_{\gamma}-2}e^{\xi (\bar 
{\boldsymbol t} _{\gamma}-\bar {\boldsymbol t} '_{\gamma},z)}}
\\ \\
\displaystyle{\phantom{aaaaaaaaaaaaaaaa}
\times \tau_{\alpha \gamma}(\bolds +[1]_{\gamma}, 
{\boldsymbol t} , \bar {\boldsymbol t} -[z^{-1}]_{\gamma})
\tau_{\gamma \beta}(\bolds '-[1]_{\gamma}, 
{\boldsymbol t} ', \bar {\boldsymbol t} '+[z^{-1}]_{\gamma})
\, dz}
\end{array}
\eeq
valid for all $\bolds , {\boldsymbol t} , \bar {\boldsymbol t} , 
\bolds ', {\boldsymbol t} ', \bar {\boldsymbol t} '$. 
\end{theorem} 

\begin{proof}
This equation is obtained by substitution of the relations (\ref{e2})
into the bilinear identity (\ref{e1}) for the wave functions. 
\end{proof}

\noindent
The simplest solution of equation (\ref{n6}) is
$\displaystyle{\tau_{\alpha \beta}(\bolds , \boldt , \bar \boldt )=
\delta_{\alpha \beta}\exp \Bigl (\sum_{k\geq 1}\sum_{\gamma =1}^N 
k t_{\gamma ,k}\bar t_{\gamma ,k}\Bigr )}$.

\subsection{Tau-function as a universal dependent variable}
\label{subsection:dependent}

The tau-function plays the role of the universal dependent variable
of the hierarchy meaning that
the coefficients of the Lax operators and of the operators
$\boldU_{\alpha}$, $\boldUbar_{\alpha}$, 
$\boldP_{\alpha}$, $\boldPbar_{\alpha}$ can be expressed through
the tau-function. Indeed, the coefficients of the wave functions
(and of the wave operators)
can be expressed through the tau-function by expanding the formulae 
(\ref{e2}) in inverse powers of $z$, then the coefficients of all
the operators that participate in the Lax formalism are obtained 
after ``dressing'' with the help of the wave operators.

In particular, we have from (\ref{e2}):
\beq\label{n7}
(\bar w_0(\bolds , \boldt , \boldtbar ))_{\alpha \beta}=
(-1)^{\delta_{\alpha \beta}-1}
\frac{\tau_{\alpha \beta}
(\bolds +[1]_{\beta}, \boldt , \boldtbar )}{\tau (\bolds ,
\boldt , \boldtbar )}.
\eeq
It can be also deduced from (\ref{e2}) that the inverse matrix is
expressed through the tau-function as follows:
\beq\label{n7a}
(\bar w_0^{-1}(\bolds , \boldt , \boldtbar ))_{\alpha \beta}=
\frac{\tau_{\alpha \beta}
(\bolds +\boldone -[1]_{\alpha}, \boldt , 
\boldtbar )}{\tau (\bolds +\boldone , \boldt , \boldtbar )}.
\eeq
To see this, let us choose $\bolds '=\bolds +\boldone$, $\boldt '=\boldt$,
$\boldtbar '=\boldtbar$ in (\ref{n6}), then it becomes the identity
$$
\sum_{\gamma =1}^N (-1)^{\delta_{\alpha \gamma}-1}
\tau_{\alpha \gamma}(\bolds +[1]_{\gamma}, \boldt , \boldtbar )
\tau_{\gamma \beta}(\bolds +\boldone -[1]_{\gamma}, 
\boldt , \boldtbar)=
\delta_{\alpha \beta}\tau (\bolds , \boldt , \boldtbar )
\tau (\bolds +\boldone , \boldt , \boldtbar )
$$
which means that for solutions of the Toda lattice hierarchy 
the right-hand side of (\ref{n7a}) is indeed 
the inverse of (\ref{n7}).
Therefore, the leading coefficient $\bar b_0(\bolds )$ 
of $\boldLbar (\bolds )$ is expressed
through the tau-function by the
formula
\beq\label{n10}
(\bar b_0(\bolds, \boldt , \boldtbar  ))_{\alpha \beta}=
\sum_{\gamma =1}^N (-1)^{\delta_{\alpha \gamma}-1}\, 
\frac{\tau_{\alpha \gamma}(\bolds +[1]_{\gamma}, \boldt , \boldtbar )
\tau_{\gamma \beta}(\bolds -[1]_{\gamma}, \boldt , \boldtbar )}{\tau^2
(\bolds , \boldt , \boldtbar )}.
\eeq

From the ``dressing relation'' 
$\boldL (\bolds )=\hat \boldW (\bolds )e^{\p_s}
\hat \boldW^{-1}(\bolds )$ we have
$$
\boldL (\bolds )=1_N e^{\p_s}+w_1(\bolds )+w^*_1(\bolds +\boldone )+
\ldots ,
$$
whence the coefficient $b_1(\bolds )$ is given by
$$
b_1(\bolds )=w_1(\bolds )+w^*_1(\bolds +\boldone )=
w_1(\bolds )-w_1(\bolds +\boldone )
$$
since from $\hat \boldW (\bolds )\hat 
\boldW^{-1}(\bolds )=1_N$ it follows that
$w^*_1(\bolds )=-w_1(\bolds )$. The coefficient $w_1(\bolds )$
can be found by expanding the first equation in 
(\ref{e2}) in inverse powers of $z$
up to $z^{-1}$:
\beq\label{n8}
(w_1(\bolds , \boldt , \boldtbar ))_{\alpha \beta}=
\frac{\tau_{\alpha \beta}
(\bolds , \boldt , \boldtbar )}{\tau
(\bolds , \boldt , \boldtbar )}-
\delta_{\alpha \beta}\Bigl (\p_{t_{\alpha , 1}}\log \tau
(\bolds , \boldt , \boldtbar ) \, +1\Bigr ).
\eeq
For $b_1(\bolds )$ we then obtain:
\beq\label{n9}
(b_1(\bolds , \boldt , \boldtbar))_{\alpha \beta}=
\delta_{\alpha \beta}\p_{t_{\alpha , 1}}\log \frac{\tau
(\bolds +\boldone , \boldt , \boldtbar)}{\tau
(\bolds, \boldt , \boldtbar)}+
\frac{\tau_{\alpha \beta}
(\bolds , \boldt , \boldtbar )}{\tau
(\bolds , \boldt , \boldtbar )}-
\frac{\tau_{\alpha \beta}
(\bolds +\boldone , \boldt , \boldtbar )}{\tau
(\bolds +\boldone , \boldt , \boldtbar )}.
\eeq
In the diagonal elements the last two terms cancel 
and we are left with the expression similar to the one 
known in the one-component case.

Equations (\ref{n10a}),
(\ref{n11}) from Proposition \ref{proposition:alternative} 
provide alternative expressions
for the $\bar b_0(\bolds )$ and $b_1 (\bolds )$ which are often
more convenient than the ones obtained above.
It is instructive to give here another proof of
Proposition \ref{proposition:alternative} which is based
on the bilinear equation
(\ref{n6}).

To show (\ref{n10a}), we differentiate both sides
of (\ref{n6}) 
with respect to $\bar t_1$ (i.e., apply the differential operator
$\displaystyle{\sum_{\mu =1}^N\p_{ \bar t_{\mu ,1}}}$) and put
$\bolds '=\bolds$,
$\boldt '=\boldt$, $\bar \boldt '=\bar \boldt$ after that.
In this way we get the relations
$$
\sum_{\gamma}(-1)^{\delta_{\alpha \gamma}-1}
\tau_{\alpha \gamma}(\bolds +[1]_{\gamma})
\tau_{\gamma \beta }(\bolds -[1]_{\gamma})=
\p_{\bar t_1}\tau_{\alpha \beta}(\bolds )
\tau (\bolds )-\p_{\bar t_1}\tau (\bolds )
\tau_{\alpha \beta}(\bolds )
$$
for $\alpha \neq \beta$ and
$$
\sum_{\gamma}(-1)^{\delta_{\alpha \gamma}-1}
\tau_{\alpha \gamma}(\bolds +[1]_{\gamma})
\tau_{\gamma \alpha }(\bolds -[1]_{\gamma})=
\p_{\bar t_1}\tau (\bolds )\, \p_{t_{\alpha ,1}}
\tau (\bolds )-\p_{\bar t_1}\p_{t_{\alpha ,1}}\tau (\bolds )\, 
\tau (\bolds )
$$
for $\beta =\alpha$ (we do not indicate the
dependence on $\boldt , \bar \boldt$ explicitly). 
Together with (\ref{n7}), (\ref{n7a}) and
(\ref{n8}) they 
are equivalent to (\ref{n10a}).

To show (\ref{n11}), we express its right-hand side through
the tau-function:
\beq\label{n12}
\begin{array}{c}
\displaystyle{
\sum_{\gamma}(\p_{t_1}\bar w_0(\bolds ))_{\alpha \gamma}
(\bar w_0^{-1}(\bolds ))_{\gamma \beta}=
\sum_{\gamma}(-1)^{\delta_{\alpha \gamma}-1}
\p_{t_1}\left (
\frac{\tau_{\alpha \gamma}(\bolds +[1]_{\gamma})}{\tau (\bolds )}
\right )
\frac{\tau_{\gamma \beta}(\bolds +\boldone -
[1]_{\gamma})}{\tau (\bolds +\boldone )}}
\\ \\
\displaystyle{
=\sum_{\gamma}(-1)^{\delta_{\alpha \gamma}-1}
\frac{\p_{t_1}\tau_{\alpha \gamma}(\bolds +[1]_{\gamma})\,
\tau_{\gamma \beta}(\bolds +\boldone -
[1]_{\gamma})}{\tau (\bolds )\tau (\bolds +\boldone )}-
\delta_{\alpha \beta}\, \frac{\p_{t_1}\tau (\bolds )}{\tau (\bolds )}.}
\end{array}
\eeq
To transform the sum,
we differentiate both sides
of (\ref{n6}) 
with respect to $t_1$ (i.e., apply the differential operator
$\displaystyle{\sum_{\mu =1}^N\p_{t_{\mu ,1}}}$) and put
$\bolds '=\bolds+\boldone$,
$\boldt '=\boldt$, $\bar \boldt '=\bar \boldt$ after that.
We get the relations
$$
\sum_{\gamma}(-1)^{\delta_{\alpha \gamma}-1}
\p_{t_1}\tau_{\alpha \gamma}(\bolds +[1]_{\gamma})\,
\tau_{\gamma \beta }(\bolds +\boldone -[1]_{\gamma})=
\tau_{\alpha \beta}(\bolds )\tau (\bolds +\boldone )
-\tau (\bolds )
\tau_{\alpha \beta}(\bolds +\boldone )
$$
for $\alpha \neq \beta$ and
$$
\begin{array}{l}
\displaystyle{
\sum_{\gamma}(-1)^{\delta_{\alpha \gamma}-1}
\tau_{\alpha \gamma}(\bolds +[1]_{\gamma})
\tau_{\gamma \alpha }(\bolds +\boldone -[1]_{\gamma})}
\\ \\
\phantom{aaaaaaaaaaaaa}=\, \displaystyle{
\p_{\bar t_{\alpha ,1}}\tau (\bolds +\boldone )\, 
\tau (\bolds )-
\p_{\bar t_{\alpha ,1}}\tau (\bolds )\, 
\tau (\bolds +\boldone )+
\p_{t_1}\tau (\bolds )\, \tau (\bolds +\boldone )}
\end{array}
$$
for $\beta =\alpha$. 
Together with (\ref{n9}) they mean that the right-hand side of
(\ref{n12}) is indeed equal to $b_1(\bolds )$.

Other coefficients of the operators which 
participate in the Lax formalism can be found in a similar way.
However, for higher coefficients
the explicit formulae become rather bulky.

\subsection{Bilinear equations of the Hirota-Miwa type}
\label{subsection:Hirota-Miwa}

In this subsection we derive some bilinear equations for the tau-function
which follow from the generating bilinear equation (\ref{n6}).
Choosing $\bolds -\bolds '$, ${\boldsymbol t} -{\boldsymbol t} '$, 
$\bar {\boldsymbol t} -\bar {\boldsymbol t} '$ in some special ways, one can obtain
corollaries of (\ref{n6}) which are known
as equations of the Hirota-Miwa type. 

First of all we note that
if $\bolds '\leq \bolds $, which means $s'_\gamma\leq s_\gamma$ for all
$\gamma$, and $\bar {\boldsymbol t} '=\bar {\boldsymbol t} $, then
the right-hand side of (\ref{n6}) vanishes while the left-hand side 
becomes the integral bilinear equation for the tau-function 
of the multi-component modified KP hierarchy in the independent variables
${\boldsymbol t} $. In this way the latter is realized as a subhierarchy 
of the multi-component Toda lattice. Some bilinear equations
which are corollaries of the integral one are listed in \cite{Teo11,Z19}.

We proceed with the choice $\beta =\alpha$, $\bolds '=\bolds $,
${\boldsymbol t} -{\boldsymbol t} '=[a^{-1}]_{\alpha}$, 
$\bar {\boldsymbol t} -\bar {\boldsymbol t} '=[b^{-1}]_{\alpha}$. In this case
(\ref{n6}) takes the form
$$
\oint_{C_{\infty}}dz \frac{a}{a-z}\tau (\bolds , {\boldsymbol t} -
[z^{-1}]_{\alpha}, \bar {\boldsymbol t} )
\tau (\bolds , {\boldsymbol t} -[a^{-1}]_{\alpha}+
[z^{-1}]_{\alpha}, \bar {\boldsymbol t} -[b^{-1}]_{\alpha})
$$
$$
=
\oint_{C_{\infty}}dz z^{-2}\frac{b}{b-z}
\tau (\bolds +[1]_{\alpha}, {\boldsymbol t} , \bar {\boldsymbol t} -[z^{-1}]_{\alpha})
\tau (\bolds -[1]_{\alpha}, {\boldsymbol t} -[a^{-1}]_{\alpha},
\bar {\boldsymbol t} -[b^{-1}]_{\alpha}+[z^{-1}]_{\alpha}),
$$
where we write only non-vanishing terms in the sum over $\gamma$.
The integrals can be calculated by means of taking residues 
at the poles outside the contour $C_{\infty}$, including the non-zero
residue at infinity in the left-hand side. As a result,
after the shifts $\boldt \to \boldt +[a^{-1}]_{\alpha}$,
$\bar \boldt \to \bar \boldt +[b^{-1}]_{\alpha}$
we obtain the equation
\beq\label{hm1}
\begin{array}{c}
\tau (\bolds , {\boldsymbol t} , \bar {\boldsymbol t} +[b^{-1}]_{\alpha})
\tau (\bolds , {\boldsymbol t} +[a^{-1}]_{\alpha}, \bar {\boldsymbol t} )
-\tau (\bolds , {\boldsymbol t} , \bar {\boldsymbol t} )
\tau (\bolds , {\boldsymbol t} +[a^{-1}]_{\alpha}, \bar {\boldsymbol t} +[b^{-1}]_{\alpha})
\\ \\
=(ab)^{-1}
\tau (\bolds +[1]_{\alpha}, 
{\boldsymbol t} +[a^{-1}]_{\alpha}, \bar {\boldsymbol t} )
\tau (\bolds -[1]_{\alpha}, 
{\boldsymbol t} , \bar {\boldsymbol t} +[b^{-1}]_{\alpha}),
\end{array}
\eeq
which is the bilinear equation of the one-component Toda
lattice for each $\alpha$-th component. The other equations given
below mix different components.
If $\beta \neq \alpha$, we obtain in the same way:
\beq\label{hm2}
\begin{array}{c}
\tau_{\alpha \beta}
(\bolds , {\boldsymbol t} +[a^{-1}]_{\alpha}, \bar {\boldsymbol t} +[b^{-1}]_{\alpha})
\tau (\bolds , {\boldsymbol t} , \bar {\boldsymbol t} )-
\tau_{\alpha \beta}
(\bolds , {\boldsymbol t} +[a^{-1}]_{\alpha}, \bar {\boldsymbol t} )
\tau (\bolds , {\boldsymbol t} , \bar {\boldsymbol t} +[b^{-1}]_{\alpha})
\\ \\
=
b^{-1}\tau_{\alpha \beta}
(\bolds -[1]_{\alpha}, {\boldsymbol t} , \bar {\boldsymbol t} +[b^{-1}]_{\alpha})
\tau (\bolds +[1]_{\alpha}, {\boldsymbol t} +[a^{-1}]_{\alpha}, \bar {\boldsymbol t} )
\quad (\beta \neq \alpha ).
\end{array}
\eeq

Our next choice is
$\beta \neq \alpha$, $\bolds '=\bolds $,
${\boldsymbol t} -{\boldsymbol t} '=[a^{-1}]_{\alpha}$, 
$\bar {\boldsymbol t} -\bar {\boldsymbol t} '=[b^{-1}]_{\beta}$. In this case
the residue calculus yields:
\beq\label{hm3}
\begin{array}{c}
\tau_{\alpha \beta}
(\bolds , {\boldsymbol t} +[a^{-1}]_{\alpha}, \bar {\boldsymbol t} )
\tau (\bolds , {\boldsymbol t} , \bar {\boldsymbol t} +[b^{-1}]_{\beta})-
\tau_{\alpha \beta}
(\bolds , {\boldsymbol t} +[a^{-1}]_{\alpha}, \bar {\boldsymbol t} +[b^{-1}]_{\beta})
\tau (\bolds , {\boldsymbol t} , \bar {\boldsymbol t} )
\\ \\
=
b^{-1}\tau_{\alpha \beta}
(\bolds +[1]_{\beta}, {\boldsymbol t} +[a^{-1}]_{\alpha}, \bar {\boldsymbol t} )
\tau (\bolds -[1]_{\beta}, {\boldsymbol t} , \bar {\boldsymbol t} +[b^{-1}]_{\beta})
\quad (\beta \neq \alpha ).
\end{array}
\eeq

Next, we choose
$\beta \neq \alpha$, $\bolds '=\bolds +[1]_{\alpha}+[1]_{\beta}$,
${\boldsymbol t} -{\boldsymbol t} '=[a^{-1}]_{\alpha}$, 
$\bar {\boldsymbol t} -\bar {\boldsymbol t} '=[b^{-1}]_{\beta}$ and obtain the equation
\beq\label{hm4}
\begin{array}{l}
\tau_{\alpha \beta}
(\bolds +[1]_{\beta}, {\boldsymbol t} +[a^{-1}]_{\alpha}, \bar {\boldsymbol t} )
\tau (\bolds +[1]_{\alpha}, {\boldsymbol t} , \bar {\boldsymbol t} +[b^{-1}]_{\beta})
\\ \\
\phantom{aaaaaaaa}
-
\tau_{\alpha \beta}
(\bolds +[1]_{\beta}, {\boldsymbol t} , \bar {\boldsymbol t} )
\tau (\bolds +[1]_{\alpha}, {\boldsymbol t} +[a^{-1}]_{\alpha}, 
\bar {\boldsymbol t} +[b^{-1}]_{\beta})
\\ \\
\phantom{aaaaaaaaaaaaaa}
=
a^{-1}\tau_{\alpha \beta}
(\bolds +[1]_{\alpha}+[1]_{\beta}, {\boldsymbol t} +[a^{-1}]_{\alpha}, 
\bar {\boldsymbol t} )
\tau (\bolds , {\boldsymbol t} , \bar {\boldsymbol t} +[b^{-1}]_{\beta})
\quad (\beta \neq \alpha ).
\end{array}
\eeq
In a similar way, the choice 
$\beta \neq \alpha$, $\bolds '=\bolds +[1]_{\alpha}$,
${\boldsymbol t} -{\boldsymbol t} '=[a^{-1}]_{\alpha}$, 
$\bar {\boldsymbol t} -\bar {\boldsymbol t} '=[b^{-1}]_{\beta}$ leads to the equation
\beq\label{hm5}
\begin{array}{l}
a^{-1}
\tau_{\alpha \beta}(\bolds +[1]_{\alpha}, {\boldsymbol t} 
+[a^{-1}]_{\alpha}, \bar {\boldsymbol t} )
\tau (\bolds , {\boldsymbol t} , \bar {\boldsymbol t} +[b^{-1}]_{\beta})
\\ \\
\phantom{aaaaaaaaaaa}
-
\tau_{\alpha \beta}(\bolds , {\boldsymbol t} 
+[a^{-1}]_{\alpha}, \bar {\boldsymbol t} +[b^{-1}]_{\beta})
\tau (\bolds +[1]_{\alpha}, {\boldsymbol t} , \bar {\boldsymbol t} )
\\ \\
=b^{-1}
\tau_{\alpha \beta}(\bolds +[1]_{\beta}, {\boldsymbol t} 
+[a^{-1}]_{\alpha}, \bar {\boldsymbol t} )
\tau (\bolds +[1]_{\alpha}-[1]_{\beta}, 
{\boldsymbol t} , \bar {\boldsymbol t} +[b^{-1}]_{\beta})
\\ \\
\phantom{aaaaaaaaaaa}
-
\tau_{\alpha \beta} (\bolds , {\boldsymbol t} , \bar {\boldsymbol t} )
\tau (\bolds +[1]_{\alpha}, {\boldsymbol t} 
+[a^{-1}]_{\alpha}, \bar {\boldsymbol t} +[b^{-1}]_{\beta})
\quad (\beta \neq \alpha ).
\end{array}
\eeq

At last, we put $\beta \neq \alpha$, 
$\bolds '=\bolds +[1]_{\mu}$,
${\boldsymbol t}' ={\boldsymbol t}$, 
$\bar {\boldsymbol t}' =\bar {\boldsymbol t}$, 
$\mu \neq \alpha , \beta$.
The residue calculus yields:
\beq\label{hm6}
\begin{array}{c}
\tau_{\alpha \beta}(\bolds +[1]_{\mu}, {\boldsymbol t} , \bar {\boldsymbol t} )
\tau (\bolds , {\boldsymbol t} , \bar {\boldsymbol t} )-
\tau_{\alpha \beta}(\bolds , {\boldsymbol t} , \bar {\boldsymbol t} )
\tau (\bolds +[1]_{\mu}, {\boldsymbol t} , \bar {\boldsymbol t} )
\\ \\
=\tau_{\alpha \mu}(\bolds +[1]_{\mu}, {\boldsymbol t} , \bar {\boldsymbol t} )
\tau_{\mu \beta}(\bolds , {\boldsymbol t} , \bar {\boldsymbol t} ).
\end{array}
\eeq

\section{The multi-component Toda lattice from the universal hierarchy}
\label{sec:toda-from-universal-hierarchy}

\subsection{The universal hierarchy}

Let us start from the $M$-component KP hierarchy
\cite{DJKM81,KL93,TT07,Teo11}
in which the 
additional integrable flows are allowed to take arbitrary complex
values.
In \cite{KZ23} it is called the 
universal hierarchy.
The independent variables are $M$ infinite sets of (in general complex) 
``times''
$$
{\boldsymbol t} =\{{\boldsymbol t} _1, {\boldsymbol t} _2, \ldots , {\boldsymbol t} _M\}, \qquad
{\boldsymbol t} _{\alpha}=\{t_{\alpha , 1}, t_{\alpha , 2}, 
t_{\alpha , 3}, \ldots \, \},
\qquad \alpha = 1, \ldots , M
$$
and $M$ additional variables $r_1, \ldots , r_M$
such that 
\beq\label{s1}
\sum_{\alpha =1}^M r_{\alpha}=0.
\eeq
We denote by
${\boldsymbol r}$ the set $\{r_1, \ldots , r_M\}$
and use the already introduced notation 
\beq\label{s2}
{\boldsymbol r}+[1]_{\alpha} =\{r_1, \ldots , 
r_{\alpha}\! +\! 1, \ldots , r_M\}, 
\quad {\boldsymbol r} +[1]_{\alpha \beta}=
{\boldsymbol r} +[1]_{\alpha}-[1]_{\beta}.
\eeq
In general we treat $r_1, \ldots , r_M$ as complex variables,
as in \cite{KZ23}.
If they are restricted to be integers, the hierarchy coincides with the 
one considered in
\cite{DJKM81}--\cite{Teo11}.

In the bilinear formalism, 
the dependent variable is the tau-function 
$\tau ({\boldsymbol r}, {\boldsymbol t})$.
The universal hierarchy is the infinite set of bilinear equations
for the tau-function which are encoded in the basic bilinear
relation
\beq\label{s3}
\begin{array}{l}
\displaystyle{
\sum_{\gamma =1}^M \epsilon_{\alpha \gamma}({\boldsymbol r})
\epsilon^{-1}_{\beta \gamma}({\boldsymbol r}')
\oint_{C_{\infty}}\! dz \, 
z^{r_{\gamma}-r_{\gamma}'+\delta_{\alpha \gamma}+\delta_{\beta \gamma}-2}
e^{\xi ({\boldsymbol t} _{\gamma}-{\boldsymbol t} _{\gamma}', \, z)}}
\\ \\
\displaystyle{\phantom{aaaaaaaaaaaaaaaaa}
\times \tau \left ({\boldsymbol r}+[1]_{\alpha \gamma}, 
{\boldsymbol t} -[z^{-1}]_{\gamma}\right )
\tau \left ({\boldsymbol r}'+[1]_{\gamma \beta}, 
{\boldsymbol t} '+[z^{-1}]_{\gamma}\right )=0}
\end{array}
\eeq 
valid for any $\alpha$, $\beta$, ${\boldsymbol t}$, ${\boldsymbol t}'$
and ${\boldsymbol r}$, ${\boldsymbol r}'$
such that ${\boldsymbol r}-{\boldsymbol r}' \in \ZZ ^M$
(and subject to the constraint (\ref{s1})). 
In (\ref{s3}) 
\beq\label{s3a}
\epsilon_{\alpha \gamma}({\boldsymbol r})=\left \{
\begin{array}{cl} 
\displaystyle{\exp \, \Bigl (-i\pi \!\! \sum_{\alpha <\mu \leq \gamma}
\!\! r_{\mu}
\Bigr )}, 
&\quad \alpha < \gamma
\\ 
\hspace{-1.5cm}1, &\quad \alpha =\gamma
\\ 
\displaystyle{-\, \vphantom{\sum^{\alpha \leq }}
\exp \, \Bigl (i\pi \!\! \sum_{\gamma <\mu \leq \alpha}\!\! r_{\mu}
\Bigr )}, 
&\quad \alpha > \gamma 
\end{array}
\right.
\eeq
(see \cite{KZ23}).
The contour $C_{\infty}$ is a big circle around infinity.
It is easy to see that the equation (\ref{s3}) depends only on the
differences $r_{\alpha}-r'_{\alpha}$ which are integers. 
Different bilinear relations for the tau-function which follow from
(\ref{s3}) for special choices of 
${\boldsymbol r}-{\boldsymbol r}'$ and 
${\boldsymbol t} -{\boldsymbol t} '$
are given in \cite{Teo11}.

\subsection{Specification to the multi-component Toda lattice hierarchy}

Let us show that $M=2N$-component universal hierarchy contains 
the $N$-component Toda hierarchy \cite{UT84}. 

Let the set of $2N$ indices in the universal hierarchy be
$\{1, 2, \ldots , N, \bar 1, \bar 2, \ldots , \bar N\}$. We identify
the time variables ${\boldsymbol t} _{\alpha}$ of the Toda lattice 
with ${\boldsymbol t} _{\alpha}$ of the 
universal hierarchy and
$\bar {\boldsymbol t} _{\alpha}$ with ${\boldsymbol t} _{\bar \alpha}$ 
for $\alpha =1, \ldots , N$. Besides, we set 
$r_{\alpha}=-r_{\bar \alpha}=:s_{\alpha}$. 
In this section we consider the $s_{\alpha}'s$ as complex numbers
such that $s_{\alpha}-s_{\alpha}'$ are integers. 
By ${\boldsymbol t} $ we denote
the set of times ${\boldsymbol t} =
\{{\boldsymbol t} _1, {\boldsymbol t} _2, \ldots , {\boldsymbol t} _N,
\bar {\boldsymbol t} _1, 
\bar {\boldsymbol t} _2, \ldots , \bar {\boldsymbol t} _N\}$ and by
$\bolds $ we denote
the $2N$-component set $\bolds =\{s_1, s_2, \ldots , s_N, -s_1, -s_2,
\ldots , -s_N\}$. Consider the tau-function 
$\tau (\bolds +{\boldsymbol r}, {\boldsymbol t} )$ 
of the $2N$-component universal
hierarchy, where the variables ${\boldsymbol r}$ 
are subject to the constraint
(\ref{s1}). The tau-function of the $N$-component Toda hierarchy
is the $N\times N$ matrix with matrix elements
\beq\label{mt1}
\tau_{\alpha \gamma} (\bolds , {\boldsymbol t} , 
\bar {\boldsymbol t} )=
\tau (\bolds +[1]_{\alpha \gamma}, {\boldsymbol t} ), 
\quad \alpha , \gamma =
1, \ldots , N.
\eeq
It is easy to see that
\beq\label{mt2}
\begin{array}{l}
\tau (\bolds +[1]_{\alpha \bar \gamma}, {\boldsymbol t} )=
\tau_{\alpha \gamma} (\bolds +[1]_{\gamma}, {\boldsymbol t} , \bar {\boldsymbol t} ),
\\ \\
\tau (\bolds +[1]_{\bar \alpha \gamma}, {\boldsymbol t} )=
\tau_{\alpha \gamma} (\bolds -[1]_{\alpha}, {\boldsymbol t} , \bar {\boldsymbol t} ),
\\ \\
\tau (\bolds +[1]_{\bar \alpha \bar \gamma}, {\boldsymbol t} )=
\tau_{\alpha \gamma} (\bolds +[1]_{\gamma \alpha}, {\boldsymbol t} , \bar {\boldsymbol t} ).
\end{array}
\eeq

Taking into account (\ref{mt2}), we can represent
the bilinear equation (\ref{s3}) for $\alpha \in \{1, \ldots , N\}$,
$\beta \in \{\bar 1, \ldots , \bar N\}$ in the form
\beq\label{mt3}
\begin{array}{l}
\displaystyle{
\sum_{\gamma =1}^N\epsilon_{\alpha \gamma}(\bolds )
\epsilon^{-1}_{\bar \beta \gamma}(\bolds ')\oint_{C_{\infty}}
\! dz \, z^{s_{\gamma}-s_{\gamma}' +\delta_{\alpha \gamma}-2}
e^{\xi ({\boldsymbol t} _{\gamma}-{\boldsymbol t} _{\gamma}', z)}}
\\ \\
\phantom{aaaaaaaaaaaaaaaa}
\displaystyle{
\times \tau_{\alpha \gamma} (\bolds , 
{\boldsymbol t} -[z^{-1}]_{\gamma}, \bar {\boldsymbol t} )
\tau_{\gamma \beta}(\bolds '+[1]_{\beta}, 
{\boldsymbol t} '+[z^{-1}]_{\gamma}, \bar {\boldsymbol t} ')}
\\ \\
\displaystyle{
+\, \sum_{\gamma =1}^N\epsilon_{\alpha \bar \gamma}(\bolds )
\epsilon^{-1}_{\bar \beta \bar \gamma}(\bolds ')\oint_{C_{\infty}}
\! dz \, z^{s_{\gamma}'-s_{\gamma} +\delta_{\beta \gamma}-2}
e^{\xi (\bar {\boldsymbol t} _{\gamma}-\bar {\boldsymbol t} _{\gamma}', z)}}
\\ \\
\phantom{aaaaaaaaaaaa}
\displaystyle{
\times \tau_{\alpha \gamma} (\bolds +[1]_{\gamma}, 
{\boldsymbol t} , \bar {\boldsymbol t} -[z^{-1}]_{\gamma})
\tau_{\gamma \beta}(\bolds '+[1]_{\beta \gamma}, 
{\boldsymbol t} ', \bar {\boldsymbol t} '+[z^{-1}]_{\gamma})}=0.
\end{array}
\eeq
This equation holds for all ${\boldsymbol t} $, ${\boldsymbol t} '$, $\bar {\boldsymbol t} $,
$\bar {\boldsymbol t} '$, $\bolds $, $\bolds '$ such that 
$\bolds -\bolds '\in \ZZ^N$. 

Regarding the set
$\{1, 2, \ldots , N, \bar 1, \bar 2, \ldots , \bar N\}$ as the ordered 
set, it is not difficult to express the $\epsilon$-factors as 
functions of $\bolds $:
\beq\label{mt4}
\epsilon_{\alpha \gamma}(\bolds )=\left \{
\begin{array}{c} 
\displaystyle{\exp \Bigl (-i\pi \sum_{\alpha <\mu \leq \gamma}
s_{\mu}\Bigr ), \quad \alpha \leq \gamma ,}
\\ \\
\displaystyle{-\exp \Bigl (i\pi \sum_{\gamma <\mu \leq \alpha}
s_{\mu}\Bigr ), \quad \alpha > \gamma ,}
\end{array} \right.
\eeq
\beq\label{mt4a}
\epsilon_{\bar \alpha \bar \gamma}(\bolds )=\left \{
\begin{array}{c} 
\displaystyle{\exp \Bigl (i\pi \sum_{\alpha <\mu \leq \gamma}
s_{\mu}\Bigr ), \quad \alpha \leq \gamma ,}
\\ \\
\displaystyle{-\exp \Bigl (-i\pi \sum_{\gamma <\mu \leq \alpha}
s_{\mu}\Bigr ), \quad \alpha > \gamma ,}
\end{array} \right.
\eeq
\beq\label{mt4b}
\epsilon_{\alpha \bar \gamma}(\bolds )=\left \{
\begin{array}{c} 
\displaystyle{\exp \Bigl (-i\pi \sum_{\gamma <\mu \leq N}
s_{\mu}+i\pi \sum_{1\leq \mu \leq \alpha}s_{\mu}
\Bigr ), \quad \alpha \leq \gamma ,}
\\ \\
\displaystyle{\exp \Bigl (-i\pi \sum_{\alpha <\mu \leq N}
s_{\mu}+i\pi \sum_{1\leq \mu \leq \gamma}s_{\mu}
\Bigr ), \quad \alpha > \gamma ,}
\end{array} \right.
\eeq
\beq\label{mt4c}
\epsilon_{\bar \alpha \gamma}(\bolds )=\left \{
\begin{array}{c} 
\displaystyle{-\exp \Bigl (i\pi \sum_{\gamma <\mu \leq N}
s_{\mu}-i\pi \sum_{1\leq \mu \leq \alpha}s_{\mu}
\Bigr ), \quad \alpha \leq \gamma ,}
\\ \\
\displaystyle{-\exp \Bigl (i\pi \sum_{\alpha <\mu \leq N}
s_{\mu}-i\pi \sum_{1\leq \mu \leq \gamma}s_{\mu}
\Bigr ), \quad \alpha > \gamma .}
\end{array} \right.
\eeq

\begin{predl}
Equation (\ref{mt3})
is equivalent to the following equation:
\beq\label{mt3a}
\begin{array}{l}
\displaystyle{
\sum_{\gamma =1}^N\epsilon_{\alpha \gamma}
\epsilon_{0N}(\bolds )\epsilon_{\beta \gamma}(\bolds )
\epsilon^{-1}_{\beta \gamma}(\bolds ')
\oint_{C_{\infty}}
\!  z^{s_{\gamma}-s_{\gamma}' +\delta_{\alpha \gamma}-2}
e^{\xi ({\boldsymbol t} _{\gamma}-{\boldsymbol t} _{\gamma}', z)}}
\\ \\
\phantom{aaaaaaaaaaaaaaaa}
\displaystyle{
\times \tau_{\alpha \gamma} (\bolds , 
{\boldsymbol t} -[z^{-1}]_{\gamma}, \bar {\boldsymbol t} )
\tau_{\gamma \beta}(\bolds '+[1]_{\beta}, 
{\boldsymbol t} '+[z^{-1}]_{\gamma}, \bar {\boldsymbol t} ')\, dz}
\\ \\
\displaystyle{
=\, \sum_{\gamma =1}^N\epsilon_{\beta \gamma}
\epsilon_{0N}(\bolds ')
\epsilon_{\beta \gamma}(\bolds )
\epsilon^{-1}_{\beta \gamma}(\bolds ')
\oint_{C_{\infty}}
\! z^{s_{\gamma}'-s_{\gamma} +\delta_{\beta \gamma}-2}
e^{\xi (\bar {\boldsymbol t} _{\gamma}-\bar {\boldsymbol t} _{\gamma}', z)}}
\\ \\
\phantom{aaaaaaaaaaaa}
\displaystyle{
\times \tau_{\alpha \gamma} (\bolds +[1]_{\gamma}, 
{\boldsymbol t} , \bar {\boldsymbol t} -[z^{-1}]_{\gamma})
\tau_{\gamma \beta}(\bolds '+[1]_{\beta \gamma}, 
{\boldsymbol t} ', \bar {\boldsymbol t} '+[z^{-1}]_{\gamma})\, dz},
\end{array}
\eeq
where $\epsilon_{\alpha \gamma}$ are sign factors such that
$\epsilon_{\alpha \gamma}=1$ for $\alpha \leq \gamma$ and
$\epsilon_{\alpha \gamma}=-1$ for $\alpha >\gamma$ and
$$
\epsilon_{0N}(\bolds )=\exp \Bigl (-i\pi 
\sum_{1\leq \mu \leq N}s_{\mu}\Bigr ).
$$
\end{predl}

\noindent
Clearly, the products of $\epsilon$-factors here depend 
only on $\bolds -\bolds '$
and are just signs $\pm 1$ for $\bolds -\bolds '\in \ZZ^N$. 

\begin{proof}
Equation (\ref{mt3a}) is 
obtained by plugging (\ref{mt4})--(\ref{mt4c}) into (\ref{mt3}). 
Let us present some details of the calculation which transforms
the $\epsilon$-factors in (\ref{mt3}) to those in (\ref{mt3a}).
For brevity, we write the $\gamma$-th terms of the sums 
in (\ref{mt3}) as
$$
\epsilon_{\alpha \gamma}(\bolds )
\epsilon^{-1}_{\bar \beta \gamma}(\bolds ')I_{\alpha \beta \gamma}+
\epsilon_{\alpha \bar \gamma}(\bolds )
\epsilon^{-1}_{\bar \beta \bar \gamma}(\bolds ')
I'_{\alpha \beta \gamma}.
$$
Suppose first that $\alpha \leq \beta \leq \gamma$. Plugging here
(\ref{mt4})--(\ref{mt4c}), we represent this as
$$
\begin{array}{l}
\displaystyle{
-\epsilon_{\alpha \beta}\exp \Bigl (-i\pi\sum_{\alpha <\mu \leq \gamma}
s_{\mu} -i\pi\sum_{\gamma <\mu \leq N}
s_{\mu}' +i\pi\sum_{1 \leq \mu \leq \beta}
s_{\mu}'\Bigr )I_{\alpha \beta \gamma}}
\\ \\
\displaystyle{
+\epsilon_{\beta \gamma}\exp \Bigl (-i\pi\sum_{\gamma <\mu \leq N}
s_{\mu} +i\pi\sum_{1 \leq \mu \leq \alpha}
s_{\mu} -i\pi\sum_{\beta < \mu \leq \gamma}
s_{\mu}'\Bigr )I_{\alpha \beta \gamma}'},
\end{array}
$$
or
\beq\label{or}
\begin{array}{c}
\displaystyle{
-\epsilon_{\alpha \beta}\exp \Bigl (-i\pi\sum_{\alpha <\mu \leq N}
s_{\mu} +i\pi\sum_{1 \leq \mu \leq \beta}s_{\mu}+
i\pi\sum_{\gamma <\mu \leq N}(s_{\mu}-s'_{\mu})-
i\pi\sum_{1 \leq \mu \leq \beta}(s_{\mu}-s'_{\mu})\Bigr )
I_{\alpha \beta \gamma}}
\\ \\
\displaystyle{
+\epsilon_{\beta \gamma}\exp \Bigl (-i\pi\sum_{\beta <\mu \leq N}
s_{\mu} +i\pi\sum_{1 \leq \mu \leq \alpha}
s_{\mu} 
+i\pi\sum_{\beta < \mu \leq \gamma}(s_{\mu}-
s_{\mu}')\Bigr )I_{\alpha \beta \gamma}'}.
\end{array}
\eeq
We note that
$$
\exp \Bigl (-i\pi\sum_{\alpha <\mu \leq N}
s_{\mu} +i\pi\sum_{1 \leq \mu \leq \beta}s_{\mu}\Bigr )=
\exp \Bigl (-i\pi\sum_{\beta <\mu \leq N}
s_{\mu} +i\pi\sum_{1 \leq \mu \leq \alpha}
s_{\mu} \Bigr )=: A_{\alpha \beta}(\bolds ).
$$
Recall also that
$s_{\mu}-s'_{\mu}$ are integers, 
so (\ref{or}) can be written as
$$
\begin{array}{l}
\displaystyle{
A_{\alpha \beta}(\bolds )\Bigl [
-\epsilon_{\alpha \beta}\exp \Bigl (
i\pi\sum_{\gamma <\mu \leq N}(s_{\mu}-s'_{\mu})+
i\pi\sum_{1 \leq \mu \leq \beta}(s_{\mu}-s'_{\mu})\Bigr )
I_{\alpha \beta \gamma}}
\\ \\
\displaystyle{ \phantom{aaaaaaaaaaaaaaaaaaaaaaaaaaaa}
+\epsilon_{\beta \gamma}\exp \Bigl (
i\pi\sum_{\beta < \mu \leq \gamma}(s_{\mu}-
s_{\mu}')\Bigr )\Bigr )I_{\alpha \beta \gamma}'\Bigr ]},
\end{array}
$$
which is
$$
A_{\alpha \beta}(\bolds )\Bigl [-\epsilon_{\alpha \beta}
\epsilon_{0N}(\bolds )\epsilon_{0N}^{-1}(\bolds ')\epsilon_{\beta \gamma}
(\bolds )\epsilon_{\beta \gamma}^{-1}(\bolds ')I_{\alpha \beta \gamma}
+\epsilon_{\beta \gamma}\epsilon_{\beta \gamma}
(\bolds )\epsilon_{\beta \gamma}^{-1}(\bolds ')I'_{\alpha \beta \gamma}
\Bigr ].
$$
Similar calculations show that in the other cases, $\alpha \leq \gamma
\leq \beta$ and $\gamma \leq \alpha \leq \beta$, one obtains the same
result. Therefore, $A_{\alpha \beta}(\bolds )$ is an inessential 
common multiplier and we arrive at (\ref{mt3a}). 
\end{proof}

There is an equivalent form of equation (\ref{mt3a}) which
is more symmetric with respect to $\alpha , \beta$:
\beq\label{mt3b1}
\begin{array}{l}
\displaystyle{
\sum_{\gamma =1}^N 
\epsilon_{\beta \gamma}\epsilon_{0N}(\bolds )
\epsilon_{\alpha \gamma}(\bolds )
\epsilon^{-1}_{\beta \gamma}(\bolds ')
\oint_{C_{\infty}}
\! z^{s_{\gamma}-s_{\gamma}' +\delta_{\alpha \gamma}-2}
e^{\xi ({\boldsymbol t} _{\gamma}-{\boldsymbol t} _{\gamma}', z)}}
\\ \\
\phantom{aaaaaaaaaaaaaaaa}
\displaystyle{
\times \tau_{\alpha \gamma} (\bolds , 
{\boldsymbol t} -[z^{-1}]_{\gamma}, \bar {\boldsymbol t} )
\tau_{\gamma \beta}(\bolds '+[1]_{\beta}, 
{\boldsymbol t} '+[z^{-1}]_{\gamma}, \bar {\boldsymbol t} ') \, dz }
\\ \\
\displaystyle{
=\, \sum_{\gamma =1}^N
\epsilon_{\alpha \gamma}
\epsilon_{\alpha \gamma}(\bolds )
\epsilon^{-1}_{\beta \gamma}(\bolds ')
\epsilon_{0N}(\bolds ')
\oint_{C_{\infty}}
\! z^{s_{\gamma}'-s_{\gamma} +\delta_{\beta \gamma}-2}
e^{\xi (\bar {\boldsymbol t} _{\gamma}-\bar {\boldsymbol t} _{\gamma}', z)}}
\\ \\
\phantom{aaaaaaaaaaaa}
\displaystyle{
\times \tau_{\alpha \gamma} (\bolds +[1]_{\gamma}, 
{\boldsymbol t} , \bar {\boldsymbol t} -[z^{-1}]_{\gamma})
\tau_{\gamma \beta}(\bolds '+[1]_{\beta \gamma}, 
{\boldsymbol t} ', \bar {\boldsymbol t} '+[z^{-1}]_{\gamma}) \, dz}.
\end{array}
\eeq

The bilinear equation (\ref{s3}) contains also three other types 
of equations which correspond to the choices
$\alpha , \beta \in \{1, \ldots , N\}$,
$\alpha ,\beta \in \{\bar 1, \ldots , \bar N\}$ and
$\alpha \in \{\bar 1, \ldots , \bar N\}$,
$\beta \in \{1, \ldots , N\}$. A thorough analysis shows that they
are equivalent to (\ref{mt3b1}) and can be obtained from it by 
shifting some of $\bolds $- and $\bolds '$-variables by $\pm 1$.
For example, the equation which corresponds to the choice
$\alpha , \beta \in \{1, \ldots , N\}$ has the form
\beq\label{main}
\begin{array}{l}
\displaystyle{
\sum_{\gamma =1}^N 
\epsilon_{\alpha \gamma}(\bolds )
\epsilon^{-1}_{\beta \gamma}(\bolds ')
\oint_{C_{\infty}}
\!  z^{s_{\gamma}-s_{\gamma}' +\delta_{\alpha \gamma}+
\delta_{\beta \gamma}-2}
e^{\xi ({\boldsymbol t} _{\gamma}-{\boldsymbol t} _{\gamma}', z)}}
\\ \\
\phantom{aaaaaaaaaaaaaaaa}
\displaystyle{
\times \tau_{\alpha \gamma} (\bolds , 
{\boldsymbol t} -[z^{-1}]_{\gamma}, \bar {\boldsymbol t} )
\tau_{\gamma \beta}(\bolds ', 
{\boldsymbol t} '+[z^{-1}]_{\gamma}, \bar {\boldsymbol t} ') \, dz }
\\ \\
\displaystyle{
=\, - \sum_{\gamma =1}^N
\epsilon_{\alpha \gamma}
\epsilon_{\beta \gamma}
\epsilon_{\alpha \gamma}(\bolds )
\epsilon_{0N}(\bolds )
\epsilon^{-1}_{\beta \gamma}(\bolds ')
\epsilon_{0N}^{-1}(\bolds ')
\oint_{C_{\infty}}
\!  z^{s_{\gamma}'-s_{\gamma} -2}
e^{\xi (\bar {\boldsymbol t} _{\gamma}-\bar {\boldsymbol t} _{\gamma}', z)}}
\\ \\
\phantom{aaaaaaaaaaaa}
\displaystyle{
\times \tau_{\alpha \gamma} (\bolds +[1]_{\gamma}, 
{\boldsymbol t} , \bar {\boldsymbol t} -[z^{-1}]_{\gamma})
\tau_{\gamma \beta}(\bolds '-[1]_{\gamma}, 
{\boldsymbol t} ', \bar {\boldsymbol t} '+[z^{-1}]_{\gamma}) \, dz}.
\end{array}
\eeq
It can be obtained from (\ref{mt3b1}) by the substitution
$\bolds '\to \bolds '-[1]_{\beta}$. 
If we set $\bar {\boldsymbol t} '=\bar {\boldsymbol t} $,
the right-hand side vanishes for $s_{\mu}'\leq s_{\mu}$ and we get 
the bilinear equation for the $N$-component modified KP hierarchy
given in \cite{Z19}.

\begin{zam}
For $N=1$ equation (\ref{mt3b1}) becomes the standard bilinear
equation
\beq\label{toda}
\begin{array}{l}
\displaystyle{
\oint_{C_{\infty}}
\!  z^{s-s'-1}
e^{\xi ({\boldsymbol t} -{\boldsymbol t} ', z)}
\tau^{\rm Toda}_{s} ( 
{\boldsymbol t} -[z^{-1}], \bar {\boldsymbol t} )
\tau^{\rm Toda}_{s'+1}( 
{\boldsymbol t} '+[z^{-1}], \bar {\boldsymbol t} ') \, dz }
\\ \\
\displaystyle{
=\oint_{C_{\infty}}
\! z^{s'-s-1}
e^{\xi (\bar {\boldsymbol t} -\bar {\boldsymbol t} ', z)}
\tau^{\rm Toda}_{s+1} ( 
{\boldsymbol t} , \bar {\boldsymbol t} -[z^{-1}])
\tau^{\rm Toda}_{s'}( 
{\boldsymbol t} ', \bar {\boldsymbol t} '+[z^{-1}]) \, dz }
\end{array}
\eeq
for the tau-function $\tau^{\rm Toda}_{s} ( 
{\boldsymbol t} , \bar {\boldsymbol t} )=e^{\frac{\pi i}{2}\, s(s-1)}\tau_{11}(s, {\boldsymbol t} , 
\bar {\boldsymbol t} )$ of the Toda lattice hierarchy \cite{UT84}. 
\end{zam}

It remains to clarify how the tau-function from this section
is related to the tau-function introduced in section 
\ref{section:tau-function}, because the bilinear equations
for them look different. This matter is clarified by the
following proposition.

\begin{predl}
The tau-functions $\tau_{\alpha \beta}^{(\rm sec.5)}(\bolds , 
{\boldsymbol t} , \bar {\boldsymbol t} )$ and $\tau_{\alpha \beta}^{(\rm sec.6)}(\bolds , 
{\boldsymbol t} , \bar {\boldsymbol t} )$ differ by a simple sign factor:
\beq\label{56}
\tau_{\alpha \beta}^{(\rm sec.5)}(\bolds , {\boldsymbol t} , \bar {\boldsymbol t} )=
(-1)^{|\bolds |(|\bolds |-1)/2}\, \epsilon_{\alpha \beta}(\bolds )\,
\tau_{\alpha \beta}^{(\rm sec.6)}(\bolds , {\boldsymbol t} , \bar {\boldsymbol t} ),
\eeq
where $\displaystyle{|\bolds |=\sum_{\mu =1}^N s_{\mu}.}$
After this substitution the bilinear equations (\ref{n6}) 
and (\ref{main}) become identical. 
\end{predl}

\noindent
The proof consists in a direct verification.

\section{The multi-component Toda lattice hierarchy from
free fermions}
\label{section:fermions}

In this section we show how the multi-component Toda lattice hierarchy
can be obtained in the framework of the free fermion technique developed
by the Kyoto school. In this approach the variables $s_{\alpha}$ are
assumed to be integers again, as before
\secref{sec:toda-from-universal-hierarchy}.

\subsection{The multi-component fermions}

The Clifford algebra $\calA$ is generated by the creation-annihilation
$N$-component free fermionic operators:
\[
    \calA :=
    \left<\left.
    \psi_{j}^{(\alpha )}, \psi_{j}^{*(\alpha )} \, \right|\, 
    \alpha =1,\dotsc, N,\ j\in \ZZ
    \right>.
\]
These operators satisfy the standard 
anti-commutation relations:
\[
    [\psi_{j}^{(\alpha )}, \psi_{k}^{*(\beta )}]_+
    =
    \delta_{\alpha \beta}\delta_{jk},
    \qquad
    [\psi_{j}^{(\alpha )}, \psi_{k}^{(\beta )}]_+=
    [\psi_{j}^{*(\alpha )}, \psi_{k}^{*(\beta )}]_+=0.
\]
The generating functions of $\psi_j^{(\alpha)}$ and
$\psi_j^{*(\alpha)}$,
\[
\psi^{(\alpha )}(z)=\sum_{j\in \z}\psi^{(\alpha )}_j z^j,
\qquad
\psi^{*(\alpha )}(z)=\sum_{j\in \z}\psi^{*(\alpha )}_j z^{-j},
\]
are called the free fermionic fields.

\begin{zam}
The algebra of $N$-component fermionic operators is in fact
isomorphic to the algebra of one-component operators $\psi_j$,
$\psi^*_j$. The isomorphism is established by the map
$\psi_j^{(\alpha )}\mapsto \psi_{Nj +\alpha}$,
$\psi_j^{*(\alpha )}\mapsto \psi_{Nj +\alpha}^*$.
 See, for example, \cite{JM83}, \S4. 
However,
in practice it is more convenient to deal with the operators
$\psi_j^{(\alpha )}$, $\psi_j^{*(\alpha )}$ rather than
$\psi_j$, $\psi^*_j$.
\end{zam}

The Fock space $\calF$ and the dual Fock space $\calF^*$ are generated
as $\calA$-modules by the vacuum state $\left | {\bf 0}\rbr$
and the dual vacuum state $\lbr {\bf 0} \right |$ 
that satisfy the conditions
\begin{equation}
\begin{aligned}
    \psi_{j}^{(\alpha )}\left | {\bf 0}\rbr &=0 \quad (j<0),
    &
    \psi_{j}^{*(\alpha )}\left | {\bf 0}\rbr &= 0 \quad (j\geq 0),
\\
    \lbr {\bf 0}\right | \psi_{j}^{(\alpha )} &=0 \quad (j\geq 0),
    &
    \lbr {\bf 0}\right | \psi_{j}^{*(\alpha )} &=0 \quad (j< 0),
\end{aligned} 
\label{condition:<0|psi=psi|0>=0}
\end{equation}
so 
$\psi_{j}^{(\alpha )}$ with $j<0$ and $\psi_{j}^{*(\alpha )}$ with
$j\geq 0$ are annihilation operators while 
$\psi_{j}^{(\alpha )}$ with $j\geq 0$ and $\psi_{j}^{*(\alpha )}$ with
$j<0$ are creation operators. In other words, 
\begin{align}
    \calF = \calA \left|{\bf 0}\rbr &\cong 
    \calA/\calA \calW_{\text{ann}}, &
    \left| {\bf 0} \right> &= 1 \bmod \calA \calW_{\text{ann}},
\\
    \calF^* = \lbr{\bf 0}\right|\calA &\cong 
    \calW_{\text{cre}}\calA \backslash \calA, &
    \lbr{\bf 0}\right| &= 1 \bmod \calW_{\text{cre}}\calA,
\end{align}
where
\begin{align*}
     \calW_{\text{ann}} :=
     \bigoplus_{j<0,\alpha} \CC \psi_j^{(\alpha)} \oplus
     \bigoplus_{j\geq0,\alpha} \CC \psi_j^{*(\alpha)},\qquad
     \calW_{\text{cre}} :=
     \bigoplus_{j\geq0,\alpha} \CC \psi_j^{(\alpha)} \oplus
     \bigoplus_{j<0,\alpha} \CC \psi_j^{*(\alpha)} .
\end{align*}

A pairing $\calF^*\otimes_\calA \calF \longrightarrow \CC$,
$\lbr u|\otimes|v\rbr \mapsto \lbr u|v \rbr$ which satisfies
\[
    (\lbr u| a)| v \rbr = \lbr u| (|a v \rbr), \qquad
    (\langle u| \in \calF^*,\ a\in\calA,\ |v\rangle \in \calF)
\]
is determined by the normalization condition $\lbr {\bf 0}|{\bf 0}\rbr =
1$. To simplify the notation, we denote
$\lbr {\bf 0}|a|{\bf 0}\rbr = \lbr a \rbr$
for $a\in\calA$. In particular,
 \begin{gather*}
    \lbr 1 \rbr = 1, \qquad 
    \lbr \psi_j^{(\alpha)} \rbr = 
    \lbr \psi_j^{*(\alpha)}\rbr = 0,
\\
    \lbr \psi_j^{(\alpha)}\psi_k^{(\beta)}\rbr =
    \lbr \psi_j^{*(\alpha)}\psi_k^{*(\beta)}\rbr = 0, \quad
    \lbr \psi_j^{(\alpha)}\psi_k^{*(\beta)}\rbr =
    \delta_{\alpha\beta}\, \delta_{jk}\, \theta(j<0),
 \end{gather*}
where $\theta(P)$ is the boolean characteristic function: $\theta(P)=1$
when $P$ is true and $\theta(P)=0$ when $P$ is false.

The normal ordered product of two fermion operators $\phi_1$, $\phi_2$
($=\psi_j^{(\alpha)}$ or $\psi_k^{*(\beta)}$) is defined by 
\begin{equation}
   \normord \phi_1 \phi_2 \normord
   =
   \phi_1 \phi_2 - \lbr \phi_1 \phi_2 \rbr.
\end{equation}
In other words, the normal ordering means moving 
annihilation operators to the right and creation
operators to the left, changing the sign any time when two fermion operators
are permuted.

For an $N$-tuple of integers $\bolds =\{s_1, s_2, \ldots , s_N\} \in
\ZZ^N$, a vector
 $\left | \bolds \rbr$ in $\calF$ 
and a vector
 $\lbr \bolds  \right |$ in $\calF^*$
are defined by
\[
    \left | \bolds \rbr 
    =
    \Psi_{s_N}^{*(N)}\dotsb \Psi_{s_2}^{*(2)}\Psi_{s_1}^{*(1)}
    \left | {\bf 0}\rbr , \qquad
    \lbr \bolds  \right |
    =
    \lbr {\bf 0} \right |
    \Psi_{s_1}^{(1)}\Psi_{s_2}^{(2)}\dotsb \Psi_{s_N}^{(N)},
\]
where
\[
    \Psi_{s}^{*(\alpha )}
    =
    \begin{cases}
     \psi^{(\alpha )}_{s-1}\ldots \psi^{(\alpha )}_{0} & (s >0),
    \\
     1 & (s=0),
    \\
     \psi^{*(\alpha )}_{s}\ldots \psi^{*(\alpha )}_{-1} & (s <0),
    \end{cases}
\qquad
    \Psi_{s}^{(\alpha )}
    =
    \begin{cases}
     \psi^{*(\alpha )}_{0}\ldots \psi^{*(\alpha )}_{s-1} &(s >0),
    \\
     1 & (s=0),
    \\
     \psi^{(\alpha )}_{-1}\ldots \psi^{(\alpha )}_{s} & (s <0).
    \end{cases}
\]
The commutation relations of 
these operators with $\psi^{(\beta)}_j$ and $\psi^{*(\beta)}_j$
($\alpha\neq\beta$) are
\begin{equation}
\begin{aligned}
    \Psi_s^{*(\alpha)} \psi^{(\beta)}_j
    &= (-1)^{s}
    \psi^{(\beta)}_j \Psi_s^{*(\alpha)},
    &
    \Psi_s^{(\alpha)} \psi^{(\beta)}_j
    &= (-1)^{s}
    \psi^{(\beta)}_j \Psi_s^{*(\alpha)},
\\
    \Psi_s^{*(\alpha)} \psi^{*(\beta)}_j
    &= (-1)^{s}
    \psi^{*(\beta)}_j \Psi_s^{*(\alpha)},
    &
    \Psi_s^{(\alpha)} \psi^{*(\beta)}_j
    &= (-1)^{s}
    \psi^{*(\beta)}_j \Psi_s^{*(\alpha)}.
\end{aligned}
\label{Psipsi=-psiPsi} 
\end{equation} 
Therefore, if $\alpha\neq\beta$,
\begin{equation}
    \Psi_{s_\alpha}^{*(\alpha)} \Psi_{s_\beta}^{*(\beta)}
    =
    (-1)^{s_\alpha s_\beta}
    \Psi_{s_\beta}^{*(\beta)} \Psi_{s_\alpha}^{*(\alpha)}, \qquad
    \Psi_{s_\alpha}^{(\alpha)} \Psi_{s_\beta}^{(\beta)}
    =
    (-1)^{s_\alpha s_\beta}
    \Psi_{s_\beta}^{(\beta)} \Psi_{s_\alpha}^{(\alpha)},
\label{PsiPsi=(-1)PsiPsi}
\end{equation}
and (cf.\ \eqref{condition:<0|psi=psi|0>=0})
\beq\label{condition:<s|psi=psi|s>=0} 
\begin{array}{lll}
\lbr \bolds \right| \psi^{(\alpha)}_{s'} &=&
\displaystyle{\left \{
     \begin{array}{ll}
       0 & (s'\geq s_\alpha), \\ & \\
       \epsilon_\alpha(\bolds ) \lbr \bolds  - [1]_{\alpha}\right|
         &(s'=s_\alpha-1),
         \end{array} \right.}
         \\ && \\
\lbr \bolds \right| \psi^{*(\alpha)}_{s'} &=&
\displaystyle{\left \{
     \begin{array}{ll}
       0 & (s'< s_\alpha), \\ & \\
       \epsilon_\alpha(\bolds ) \lbr \bolds  + [1]_{\alpha}\right|
         &(s'=s_\alpha),
        \end{array}\right. }
        \\ && \\
        \psi^{(\alpha)}_{s'} \left|\bolds \rbr &=&
        \displaystyle{\left \{
        \begin{array}{ll}
       0 & (s'< s_\alpha), \\ & \\
       \epsilon_\alpha(\bolds ) \left| \bolds  + [1]_{\alpha}\rbr
         &(s'=s_\alpha),
         \end{array} \right. }
         \\ && \\
         \psi^{*(\alpha)}_{s'} \left|\bolds \rbr &=&
         \displaystyle{\left \{
     \begin{array}{ll}
       0 & (s'\geq s_\alpha), \\ &\\ 
       \epsilon_\alpha(\bolds ) \left| \bolds  - [1]_{\alpha}\rbr
         &(s'=s_\alpha-1),
         \end{array}\right. }
 \end{array}
\eeq
where the notation $\bolds  \pm [1]_{\alpha}$ 
is introduced in (\ref{st1a}) and the sign factor $\epsilon_\alpha(\bolds )$
is
\begin{equation}
    \epsilon_\alpha(\bolds ) = (-1)^{s_{\alpha+1}+\dotsb+s_N}.
\label{def:epsilon_alpha}
\end{equation}

The current operator $J^{(\alpha)}(z)$ is defined by
\[
    J^{(\alpha)}(z)
    =
    \normord \psi^{(\alpha)}(z) \psi^{*(\alpha)} (z) \normord
    =\sum_{k\in\ZZ} J_k^{(\alpha)} z^{-k},
\]
which is a generating function of the operators
\begin{equation}
    J_{k}^{(\alpha )}
    =
    \sum_{j\in \z}
    \normord \psi^{(\alpha )}_{j} \psi^{*(\alpha )}_{j+k}\normord.
\label{def:J} 
\end{equation}
(The normal ordering here is essential only for
$J_0^{(\alpha)}$.)

The commutation relations for the 
operators $J_k^{(\alpha)}$ are
as follows:
\begin{gather*}
    [J^{(\alpha)}_k,\psi_j^{(\beta)}] 
    =
    \delta_{\alpha\beta}\, \psi_{j-k}^{(\alpha)},
    \qquad
    [J^{(\alpha)}_k,\psi_j^{*(\beta)}] 
    =
    - \delta_{\alpha\beta}\, \psi_{j+k}^{*(\alpha)},
\\
    [J^{(\alpha)}_k, J^{(\beta)}_l]
    =
    k \delta_{\alpha\beta}\, \delta_{k+l,0}.
\end{gather*}
Hence,
\begin{equation}
    [J^{(\alpha)}_k, \psi^{(\beta)}(z)]
    =
    z^k \delta_{\alpha\beta}\, \psi^{(\alpha)}(z), \qquad
    [J^{(\alpha)}_k, \psi^{*(\beta)}(z)]
    =
    -z^k \delta_{\alpha\beta}\, \psi^{*(\alpha)}(z),
\label{[J,psi]} 
\end{equation}
and
\begin{equation}
    [J^{(\alpha)}_k, \Psi^{(\beta)}_s]
    =
    [J^{(\alpha)}_k, \Psi^{*(\beta)}_s]
    = 0 \qquad
    (\alpha\neq\beta).
\label{[J,Psi]} 
\end{equation}
As a consequence of \eqref{condition:<s|psi=psi|s>=0} and \eqref{def:J}
we have
\begin{equation}
    \lbr\bolds \right| J^{(\alpha)}_{-n} = 0,
    \qquad
    J^{(\alpha)}_{n} \left|\bolds \rbr = 0 \qquad
    (n\geq1).
\label{<s|J=J|s>=0}
\end{equation}
The operators $J_{0}^{(\alpha )}=:Q_{\alpha}$ are charge operators:
\begin{equation}
    \lbr \bolds \right| Q_\alpha = \lbr \bolds \right| s_{\alpha},
    \qquad
    Q_\alpha \left|\bolds \rbr = s_{\alpha} \left|\bolds \rbr.
\label{Q=charge} 
\end{equation}

In order to describe the multi-component Toda lattice
in the fermionic language, we introduce $2N$ infinite sets of
independent continuous variables
\[
 \begin{aligned}
    {\boldsymbol t} &=\{{\boldsymbol t} _1, {\boldsymbol t} _2, \ldots , {\boldsymbol t} _N\}, &
    {\boldsymbol t} _{\alpha}
    &=
    \{t_{\alpha , 1}, t_{\alpha , 2}, t_{\alpha , 3}, \dotsc \,\},
\\
    \bar {\boldsymbol t}&=\{\bar {\boldsymbol t}_1, 
    \bar {\boldsymbol t}_2, \ldots,
    \bar {\boldsymbol t}_N\, \},
    &
    \bar {\boldsymbol t}_{\alpha}
    &=
    \{\bar t_{\alpha , 1}, \bar t_{\alpha , 2}, \bar t_{\alpha , 3}, 
     \dotsc \,\}
 \end{aligned}
\]
which will be the time variables of the Toda lattice as before.
The evolution is induced by the following operators:
\beq
    J({\boldsymbol t} )
    =
    \sum_{\alpha =1}^N \sum_{k\geq 1} 
    t_{\alpha , k}J_k^{(\alpha)},
 \qquad
    \bar J(\bar {\boldsymbol t})
    =
    \sum_{\alpha =1}^N \sum_{k\geq 1} 
    \bar t_{\alpha , k}J_{-k}^{(\alpha)}.
    \label{JJ}
\eeq
The commutation relations for $J({\boldsymbol t} )$ and 
$\bar J(\bar {\boldsymbol t})$
are as follows:
\begin{align*}
    [J({\boldsymbol t} ), \psi^{(\alpha)}(z)]
    &=
    \xi({\boldsymbol t} _\alpha,z) \psi^{(\alpha)}(z), &
    [J({\boldsymbol t} ), \psi^{*(\alpha)}(z)]
    &=
    -\xi({\boldsymbol t} _\alpha,z) \psi^{*(\alpha)}(z),
\\
    [\bar J({\bf \bar t}), \psi^{(\alpha)}(z)]
    &=
    \xi(\bar {\boldsymbol t}_\alpha,z^{-1}) \psi^{(\alpha)}(z), &
    [\bar J(\bar {\boldsymbol t}), \psi^{*(\alpha)}(z)]
    &=
    -\xi(\bar {\boldsymbol t}_\alpha,z^{-1}) \psi^{*(\alpha)}(z).
\end{align*}
Therefore,
\begin{equation}
\begin{aligned}
    e^{J({\boldsymbol t} )} \psi^{(\alpha)}(z)
    &=
    e^{\xi({\boldsymbol t} _\alpha,z)} 
    \psi^{(\alpha)}(z)\, e^{J({\boldsymbol t} )}, &
    e^{J({\boldsymbol t} )} \psi^{*(\alpha)}(z)
    &=
    e^{-\xi({\boldsymbol t} _\alpha,z)} \psi^{*(\alpha)}(z)\, e^{J({\boldsymbol t} )},
\\
    e^{\bar J(\bar {\boldsymbol t})} \psi^{(\alpha)}(z)
    &=
    e^{\xi(\bar {\boldsymbol t}_\alpha,z^{-1})} 
    \psi^{(\alpha)}(z)\, e^{\bar J(\bar {\boldsymbol t})} , &
    e^{\bar J(\bar {\boldsymbol t})} \psi^{*(\alpha)}(z)
    &=
    e^{-\xi(\bar {\boldsymbol t}_\alpha,z^{-1})}
    \psi^{*(\alpha)}(z)\, e^{\bar J(\bar {\boldsymbol t})}.
\end{aligned}
\label{eJpsi=psieJ} 
\end{equation}

As is shown in \cite{DJKM83} (equation (2.6.5)) and \cite{JM83} 
(equation (1.21)), we have the following formulae
for the one-component case:
\[
    \lbr s \right| \psi(z)
    =
    z^{s-1} \lbr s-1 \right| e^{-J([z^{-1}])},
    \qquad
    \lbr s \right| \psi^{*}(z)
    =
    z^{-s} \lbr s+1 \right| e^{J([z^{-1}])}
\]
(for details of
the proof, see \cite{DJM93}, Lemma 5.3), which are often referred to
as bosonization rules.
In a similar manner, one can prove the formulae
\[
    \psi^{*}(z) \left| s \rbr
    =
    z^{-s+1} e^{-\bar J([z])} \left| s-1 \rbr,
    \qquad
    \psi(z) \left| s \rbr
    =
    z^{s} e^{\bar J([z])} \left| s+1 \rbr. 
\]
Using commutation relations \eqref{Psipsi=-psiPsi},
\eqref{PsiPsi=(-1)PsiPsi} and \eqref{[J,Psi]}, we can deduce
multi-component analogues of the bosonization rules:
\beq\label{psi(z)-on-|s>-<s|}
\begin{array}{lll}
\lbr \bolds \right| \psi^{(\alpha)}(z)&=&
\epsilon_\alpha(\bolds ) z^{s_\alpha-1}
    \lbr \bolds -[1]_\alpha \right|
    e^{-J([z^{-1}]_\alpha)},
    \\ && \\
    \lbr \bolds \right| \psi^{*(\alpha)}(z)
    &=&
    \epsilon_\alpha(\bolds ) z^{-s_\alpha}
    \lbr \bolds +[1]_\alpha \right|
    e^{J([z^{-1}]_\alpha)},
    \\ && \\
    \psi^{(\alpha)}(z) \left| \bar \bolds \rbr
    &=&
    \epsilon_\alpha(\bar \bolds ) z^{s_\alpha}
    e^{\bar J([z]_\alpha)}
    \left| \bar \bolds +[1]_\alpha \rbr ,
    \\ && \\
    \psi^{*(\alpha)}(z) \left| \bar \bolds \rbr
    &=&
    \epsilon_\alpha(\bar \bolds ) z^{-s_\alpha + 1}
    e^{-\bar J([z]_\alpha)}
    \left|  \bar \bolds -[1]_\alpha \rbr ,
    \end{array}
\eeq
where the sign factor $\epsilon_\alpha(\bolds )$ is defined in
(\ref{def:epsilon_alpha}). 

Let $g$ be a general element of the Clifford group whose typical form
is
\beq\label{g}
    g
    =
    \exp
     \left( \sum_{\alpha , \beta} \sum_{j,k}
      A_{jk}^{(\alpha \beta )}
     \psi^{*(\alpha )}_{j}\psi^{(\beta )}_{k}
     \right)
\eeq
with some infinite matrix $A_{jk}^{(\alpha \beta )}$.
The tau-function $\tau (\bolds , \bar \bolds , {\boldsymbol t} , \bar {\boldsymbol t} ;g)$ 
is defined as the expectation value
\begin{equation}
    \tau (\bolds , \bar \bolds , {\boldsymbol t} , \bar {\boldsymbol t} ;g)
    =
    \lbr \bolds  \right | 
     e^{J({\boldsymbol t} )} g e^{-\bar J(\bar {\boldsymbol t} )}
    \left |\bar \bolds \rbr.
\label{f1}
\end{equation}
Note that it is non-zero only if $|\bolds |=|\bar \bolds |$.

\subsection{The bilinear identity}
\label{subsection:bilinear}

An important property of the Clifford group elements is the following 
operator bilinear identity. (See, for example, \S2 and \S4 of
\cite{JM83}.) 

\begin{predl}
Let $g$ be a Clifford group element of the form (\ref{g}).
Then it satisfies the operator bilinear identity
\begin{equation}
    \sum_{\gamma =1}^N \sum_{j\in \z}
     \psi_{j}^{(\gamma )}g \otimes \psi_{j}^{*(\gamma )}g
    =
    \sum_{\gamma =1}^N \sum_{j\in \z}
     g\psi_{j}^{(\gamma )}\otimes g \psi_{j}^{*(\gamma )}.
\label{f2}
\end{equation}
\end{predl}

\noindent
Both sides of this identity should be understood as acting to 
arbitrary states $\bigl |U\bigr >\otimes \bigl |V\bigr >$
from the fermionic Fock space to the right and 
$\bigl <U'\bigr |\otimes \bigl <V'\bigr |$ to the left:
\beq\label{f2a}
\sum_{\gamma =1}^N \sum_{j\in \z}\bigl <U'\bigr |
\psi_{j}^{(\gamma )}g\bigl |U\bigr > \bigl <V'\bigr |
     \psi_{j}^{*(\gamma )}g\bigl |V\bigr >=
     \sum_{\gamma =1}^N \sum_{j\in \z}\bigl <U'\bigr |
g\psi_{j}^{(\gamma )}\bigl |U\bigr > \bigl <V'\bigr |
     g\psi_{j}^{*(\gamma )}\bigl |V\bigr >.
\eeq
Using the free fermionic fields $\psi^{(\alpha)}(z)$ and
$\psi^{*(\alpha)}(z)$, we can rewrite the operator bilinear identity in
the following form: 
\begin{equation}
    \sum_{\gamma =1}^N \mbox{res} \left [\frac{dz}{z} \, 
    \psi^{(\gamma )}(z)g\otimes \psi^{*(\gamma )}(z)g\right ]
    =
    \sum_{\gamma =1}^N \mbox{res} \left [\frac{dz}{z} \, 
    g\psi^{(\gamma )}(z)\otimes g\psi^{*(\gamma )}(z)\right ]
\label{f3}
\end{equation}
(the operation res  
is defined as $\displaystyle{\mbox{res}\Bigl (
\sum_k a_k z^k \, dz\Bigr )=a_{-1}}$).

\begin{theorem}
\label{theorem:tau}
Let $g$ be a Clifford group element of the form 
(\ref{g}).
Then the function
\begin{equation}
\begin{array}{lll}
    \tau_{\alpha\beta}(\bolds ,{\boldsymbol t} , \bar {\boldsymbol t} )
    &=&(-1)^{|\bolds |(|\bolds |-1)/2}
    \tau(\bolds +[1]_\alpha - [1]_\beta, \bolds ,
         {\boldsymbol t} , \bar {\boldsymbol t} ;g)
         \\ \\
         &=& (-1)^{|\bolds |(|\bolds |-1)/2}
    \lbr \bolds +[1]_\alpha - [1]_\beta \right | 
     e^{J({\boldsymbol t} )} g e^{-\bar J(\bar {\boldsymbol t} )}
    \left |\bolds \rbr
    \end{array}
\label{def:tau-by-fermion}
\end{equation}
is the tau-function of the multi-component Toda lattice hierarchy,
i.e., it satisfies the integral bilinear equation (\ref{main}).
\end{theorem}

\begin{proof}
Putting the identity (\ref{f3}) between the states
$
    \lbr \bolds   \right |e^{J({\boldsymbol t} )}
    \otimes 
    \lbr \bolds ' \right |e^{J({\boldsymbol t} ')}
$
and
$
     e^{-\bar J(\bar {\boldsymbol t} )}\left| \bar \bolds   \rbr
    \otimes 
     e^{-\bar J(\bar {\boldsymbol t} ')}\left| \bar \bolds ' \rbr
$,
we obtain, using \eqref{eJpsi=psieJ}:
\beq\label{x1}
\begin{array}{l}
\displaystyle{
\sum_{\gamma =1}^N \oint_{C_{\infty}}\frac{dz}{z}
e^{\xi (t_{\gamma}-t'_{\gamma}, z)}
\bigl < \bolds \bigl | \psi^{(\gamma )}(z)e^{J({\boldsymbol t} )}g 
e^{-\bar J(\bar {\boldsymbol t} )}\bigr | \bar \bolds \bigr >
\bigl < \bolds '\bigl | \psi^{*(\gamma )}(z)e^{J({\boldsymbol t} ')}g 
e^{-\bar J(\bar {\boldsymbol t} ')}\bigr | \bar \bolds '\bigr > }
\\ \\
\displaystyle{
=\sum_{\gamma =1}^N \oint_{C_{0}}\frac{dz}{z}
e^{\xi (\bar t_{\gamma}-\bar t'_{\gamma}, z^{-1})}
\bigl < \bolds \bigl | e^{J({\boldsymbol t} )}g 
e^{-\bar J(\bar {\boldsymbol t} )}\psi^{(\gamma )}(z)\bigr | \bar \bolds \bigr >
\bigl < \bolds '\bigl | e^{J({\boldsymbol t} ')}g 
e^{-\bar J(\bar {\boldsymbol t} ')}\psi^{*(\gamma )}(z)
\bigr | \bar \bolds '\bigr > }
\end{array}
\eeq
Application of the bosonization rules 
(\ref{psi(z)-on-|s>-<s|}) yields:
\beq\label{x2}
\begin{array}{l}
\displaystyle{
\sum_{\gamma =1}^N \epsilon_{\gamma}(\bolds )
\epsilon_{\gamma}(\bolds ')
\oint_{C_{\infty}}\frac{dz}{z} z^{s_{\gamma}-s'_{\gamma}-1}
e^{\xi (t_{\gamma}-t'_{\gamma}, z)}}
\\ \\
\displaystyle{\phantom{aaaaaa}
\times \bigl < \bolds -[1]_{\gamma}
\bigl | e^{J({\boldsymbol t} -[z^{-1}]_{\gamma})}g 
e^{-\bar J(\bar {\boldsymbol t} )}\bigr | \bar \bolds \bigr >
\bigl < \bolds '+[1]_{\gamma}\bigl | 
e^{J({\boldsymbol t} '+[z^{-1}]_{\gamma})}g 
e^{-\bar J(\bar {\boldsymbol t} ')}\bigr | \bar \bolds '\bigr > }
\\ \\
\displaystyle{
=\sum_{\gamma =1}^N \epsilon_{\gamma}(\bar \bolds )
\epsilon_{\gamma}(\bar \bolds ')
\oint_{C_{0}}\frac{dz}{z} z^{\bar s_{\gamma}-\bar s'_{\gamma}+1}
e^{\xi (\bar t_{\gamma}-\bar t'_{\gamma}, z^{-1})}}
\\ \\
\displaystyle{\phantom{aaaaaa}
\times \bigl < \bolds \bigl | e^{J({\boldsymbol t} )}g 
e^{-\bar J(\bar {\boldsymbol t} -[z]_{\gamma})}
\bigr | \bar \bolds +[1]_{\gamma}\bigr >
\bigl < \bolds '\bigl | e^{J({\boldsymbol t} ')}g 
e^{-\bar J(\bar {\boldsymbol t} '+[z]_{\gamma})}
\bigr | \bar \bolds '-[1]_{\gamma}\bigr > }
\end{array}
\eeq
Using
the definition of the tau-function (\ref{f1}), we can rewrite this 
in the form
\begin{equation}
 \begin{aligned}
    \sum_{\gamma =1}^N
    \epsilon_\gamma(\bolds ) \epsilon_\gamma(\bolds ')
    &\oint_{C_{\infty}} \frac{dz}{z} \, 
    z^{s_\gamma-s'_\gamma-1}\, 
    e^{\xi({\boldsymbol t} _\gamma-{\boldsymbol t} '_\gamma,z)}
\\
    &\times
    \tau(\bolds -[1]_\gamma, \bar \bolds ,
         {\boldsymbol t} -[z^{-1}]_\gamma, \bar {\boldsymbol t} ;g)\, 
    \tau(\bolds '+[1]_\gamma,{\bf \bar s}',
         {\boldsymbol t} '+[z^{-1}]_\gamma, \bar {\boldsymbol t} ';g)
\\
    =
    \sum_{\gamma =1}^N
    \epsilon_\gamma(\bar \bolds ) \epsilon_\gamma(\bar \bolds ')
    &\oint_{C_{0}} \frac{dz}{z} \, 
    z^{\bar s_\gamma-\bar s'_\gamma+1}\,     
    e^{\xi(\bar {\boldsymbol t}_\gamma - 
    \bar {\boldsymbol t} '_\gamma,z^{-1})}
\\
    &\times
    \tau(\bolds , \bar \bolds +[1]_\gamma,
         {\boldsymbol t} , \bar {\boldsymbol t} -[z]_\gamma;g)\, 
    \tau(\bolds ',\bar \bolds '-[1]_\gamma,
         {\boldsymbol t} ', \bar {\boldsymbol t} '+[z]_\gamma;g).
 \end{aligned}
\label{bil-general-tau:fermion}
\end{equation}
Now let us shift
$\bolds \to \bolds +[1]_{\alpha}$, 
$\bolds '\to \bolds '-[1]_{\beta}$ in (\ref{bil-general-tau:fermion}) 
and put $\bar \bolds =\bolds $, $\bar \bolds' =\bolds'$
after that.
In the notation defined by \eqref{def:tau-by-fermion}
equation
\eqref{bil-general-tau:fermion} acquires the form
\begin{equation}
 \begin{array}{l}
 \displaystyle{
    \sum_{\gamma =1}^N
    \epsilon_\gamma(\bolds +[1]_{\alpha}) 
    \epsilon_\gamma(\bolds '-[1]_{\beta})\epsilon_{0N}(\bolds )
    \oint_{C_{\infty}} 
    z^{s_\gamma-s'_\gamma+\delta_{\alpha\gamma}+
    \delta_{\beta \gamma}-2}\, 
    e^{\xi({\boldsymbol t} _\gamma-{\boldsymbol t} '_\gamma,z)}}
\\ \\
\displaystyle{\phantom{aaaaaaaaaaaaaaaaa}
    \times
    \tau_{\alpha\gamma}(\bolds ,
                        {\boldsymbol t} -[z^{-1}]_\gamma,  \bar {\boldsymbol t} )\, 
    \tau_{\gamma\beta} (\bolds ',
                        {\boldsymbol t} '+[z^{-1}]_\gamma, \bar {\boldsymbol t} ')\, dz }
\\ \\
\displaystyle{
    =-
    \sum_{\gamma =1}^N
    \epsilon_\gamma(\bolds ) \epsilon_\gamma(\bolds ')
    \epsilon_{0N}(\bolds ')
    \oint_{C_{\infty}}  
    z^{s'_\gamma-s_\gamma -2}\, 
    e^{\xi(\bar {\boldsymbol t} _\gamma- \bar {\boldsymbol t} '_\gamma,z)}}
\\ \\
\displaystyle{\phantom{aaaaaaaaaaaaa}
    \times
    \tau_{\alpha\gamma}(\bolds +[1]_\gamma,
                        {\boldsymbol t} , \bar {\boldsymbol t} -[z^{-1}]_\gamma)\, 
    \tau_{\gamma\beta} (\bolds '-[1]_\gamma,
                        {\boldsymbol t} ', \bar {\boldsymbol t} '+[z^{-1}]_\gamma)\, dz,}
 \end{array}
\label{bil-tau:fermion:1}
\end{equation}
where we have changed the integration variable $z\to z^{-1}$ in the
right-hand side (the orientation of the contour has been changed
accordingly).

Let us compare this equation with (\ref{main}). The only difference is
in the sign factors in front of the integrals. However, a simple 
verification shows that the sign factors in the left-hand side are
$$
\begin{array}{l}
\epsilon_{\gamma}(\bolds +[1]_{\alpha})=\epsilon_{\alpha \gamma}(\bolds )
\epsilon_{\alpha}(\bolds ),
\qquad
\epsilon_{\gamma}(\bolds '-[1]_{\beta})=\epsilon_{\beta \gamma}(\bolds ')
\epsilon_{\beta}(\bolds '),
\end{array}
$$
while the sign factors in the right-hand side are
$$
\begin{array}{l}
\epsilon_{\gamma}(\bolds )=\epsilon_{\alpha \gamma}
\epsilon_{\alpha \gamma}(\bolds )\epsilon_{\alpha}(\bolds ),
\qquad
\epsilon_{\gamma}(\bolds ')=\epsilon_{\beta \gamma}
\epsilon_{\beta \gamma}(\bolds ')\epsilon_{\beta}(\bolds ').
\end{array}
$$
Extracting the common multipliers $\epsilon_{\alpha}({\bf
s})\epsilon_{\beta}(\bolds ')$ and taking into account that $s_{\gamma},
s_{\gamma}'$ are integers, we see that equation
(\ref{bil-tau:fermion:1}) is identical to (\ref{main}).  Therefore, the
function $\tau_{\alpha \beta} (\bolds , 
{\boldsymbol t} , \bar {\boldsymbol t} )$ is
indeed the tau-function of the multi-component Toda lattice hierarchy.
\end{proof}

Let us note that the tau-function introduced in Section
\ref{section:tau-function} and the one constructed here from fermions
differ by a sign factor:
\beq\label{differ}
\tau^{({\rm sec.5})}(\bolds , \boldt , \bar \boldt )=
\epsilon_{\alpha \beta}(\bolds )
\tau^{({\rm sec.7})}(\bolds , \boldt , \bar \boldt ).
\eeq
This relation should be taken into account in dealing with 
an example of exact solution of the non-abelian Toda lattice 
equation in the appendix.

\subsection{An example of exact solution}
\label{subsection:example}

Here we present a simple example of exact solution which is
analogous to the well-known 
one-soliton solution of the one-component Toda lattice.

Let us take the Clifford group element of the form
\beq\label{ex1}
g=\normordbare
\exp \Bigl (\sum_{\mu , \nu =1}^N A_{\mu \nu}\psi^{*(\mu )}(q)
\psi^{(\nu )}(p)\Bigr ) \normordbare ,
\eeq
where $A_{\mu \nu}$ is some $N\times N$ matrix
and the normal ordering $\normordbare (\ldots )\normordbare$
means that the $\psi$-operators are moved to the right while 
$\psi^{*}$-operators are moved to the left (with the minus sign
factor appearing each time when two neighboring operators are permuted).
Expanding the exponential
function, we have:
\beq\label{ex2}
\begin{array}{l}
\displaystyle{
g=1+\sum_{\mu , \nu } A_{\mu \nu}\psi^{*(\mu )}(q)\psi^{(\nu )}(p)}
\\ \\
\displaystyle{\phantom{aaaaaaaaaa}
+\, \frac{1}{2!}
\sum_{\mu_1, \nu_1}\sum_{\mu_2 , \nu_2}
A_{\mu_1 , \nu_1}A_{\mu_2 , \nu_2}\psi^{*(\mu_1 )}(q)
\psi^{*(\mu_2 )}(q)\psi^{(\nu_2 )}(p)\psi^{(\nu_1 )}(p)+\ldots \, .}
\end{array}
\eeq
In order to find the tau-function
$$
\tau_{\alpha \beta}(\bolds +[1]_{\beta}, \boldt , \boldtbar )=
\bigl < \bolds +[1]_{\alpha}\bigl |
e^{J(\boldt )}ge^{-\bar J(\boldtbar )}\bigr |\bolds +[1]_{\beta}\bigr >
$$
explicitly, we use the Wick's theorem and the pair correlation
functions
\beq\label{ex3}
\begin{array}{l}
\displaystyle{
\bigl < \bolds \bigl |
\psi^{*(\mu )}(q)\psi^{(\nu )}(p)\bigr |\bolds \bigr >=
\delta_{\mu \nu}\frac{p^{s_{\mu}}q^{1-s_{\mu}}}{q-p},}
\\ \\
\displaystyle{
\bigl < \bolds +[1]_{\alpha}\bigl |
\psi^{*(\mu )}(q)\psi^{(\nu )}(p)\bigr |\bolds +[1]_{\beta} \bigr >=
\epsilon_{\alpha}(\bolds )\epsilon_{\beta}(\bolds )
\delta_{\mu \beta}\delta_{\nu \alpha}p^{s_{\alpha}}q^{-s_{\beta}},
\quad \alpha \neq \beta .}
\end{array}
\eeq
In fact we will deal with the
slightly modified tau-function\footnote{This modification was
introduced in (1.3.32) of \cite{UT84} for
the one-component case and was interpreted as
$
    \tau'(s,\boldt,\boldtbar)
    =
    \langle s|
     \text{Ad}(e^{J(\boldt)} e^{\bar J(\boldtbar)}) g
    |s \rangle
$
in (2.2) of \cite{Tak91}.}
\beq\label{ex4}
\begin{array}{lll}
\tau'_{\alpha \beta}(\bolds +[1]_{\beta}, \boldt , \boldtbar )&=&
\displaystyle{
\exp \Bigl (\sum_{\gamma}\sum_{k\geq 1}
kt_{\gamma ,k}\bar t_{\gamma ,k}\Bigr )
\tau_{\alpha \beta}(\bolds +[1]_{\beta} , \boldt , \boldtbar )}
\\ && \\
&=& \bigl < \bolds +[1]_{\alpha}\bigl |
e^{\bar J(\boldtbar )}
e^{J(\boldt )}ge^{-J(\boldt )}
e^{-\bar J(\boldtbar )}\bigr |\bolds +[1]_{\beta}\bigr >.
\end{array}
\eeq

Let us first consider the case $\beta =\alpha$. Using the Wick's theorem
and the expansion (\ref{ex2}),
we find for $\tau'(\bolds , \boldt , \boldtbar )=
\bigl <\bolds \bigr |e^{\bar J(\boldtbar )}
e^{J(\boldt )}ge^{-J(\boldt )}
e^{-\bar J(\boldtbar )}\bigr |\bolds \bigr >$:
\beq\label{ex5}
\begin{array}{c}
\displaystyle{
\tau'(\bolds , \boldt , \boldtbar )=
1+ \! \sum_{k=1}^N \frac{1}{k!} 
\Bigl ( \frac{q}{q\! -\! p}\Bigr )^k \!\!
\!\! \sum_{\nu_1, \ldots , \nu_k}
\!\! \det \left (
\begin{array}{cccc}
A_{\nu_1 \nu_1}(\bolds , \boldt , \bar \boldt )
& \!\! A_{\nu_1 \nu_2}(\bolds , \boldt , \bar \boldt ) & \!\!
\ldots & \!\! A_{\nu_1 \nu_k}(\bolds , \boldt , \bar \boldt )
\\ &&& \\
A_{\nu_2 \nu_1}(\bolds , \boldt , \bar \boldt )
& \!\! A_{\nu_2 \nu_2}(\bolds , \boldt , \bar \boldt ) & \!\!
\ldots & \!\! A_{\nu_2 \nu_k}(\bolds , \boldt , \bar \boldt )
\\ &&& \\
\vdots & \vdots & \ddots & \!\! \vdots 
\\ &&& \\
A_{\nu_k \nu_1}(\bolds , \boldt , \bar \boldt )
& \!\! A_{\nu_k \nu_2}(\bolds , \boldt , \bar \boldt ) & \!\!
\ldots & \!\! A_{\nu_k \nu_k}(\bolds , \boldt , \bar \boldt )
\end{array}
\right )}
\\ \\
\displaystyle{
=1+\sum_{\nu}A_{\nu \nu}(\bolds , \boldt , \bar \boldt )+
\frac{1}{2!} \sum_{\nu_1, \nu_2}
\det \left (
\begin{array}{cc}
A_{\nu_1 \nu_1}(\bolds , \boldt , \bar \boldt )&
A_{\nu_1 \nu_2}(\bolds , \boldt , \bar \boldt )
\\ & \\
A_{\nu_2 \nu_1}(\bolds , \boldt , \bar \boldt )&
A_{\nu_2 \nu_2}(\bolds , \boldt , \bar \boldt )
\end{array}
\right ) +\, \ldots \, ,}
\end{array}
\eeq
where
\beq\label{ex6}
A_{\mu \nu}(\bolds , \boldt , \bar \boldt )=A_{\mu \nu}
q^{-s_{\mu}}p^{s_{\nu}}e^{\eta_{\mu \nu}(\boldt , \boldtbar , p,q)}
\eeq
and
$$
\eta_{\mu \nu}(\boldt , \boldtbar , p,q)=
\xi (\boldt_{\mu}, p)+\xi (\boldtbar_{\mu}, p^{-1})-
\xi (\boldt_{\nu}, q)-\xi (\boldtbar_{\nu}, q^{-1}).
$$
To see what it is, we recall the following well-known lemma:

\begin{lem}
\label{lemma:det(1+M)}
For an $N\times N$ matrix $M$, $\det (1_N+M)$ is $1$ plus the sum
of all diagonal minors of $M$
 (the finite-dimensional Fredholm determinant):
\beq\label{ex7}
\det_{N\times N}(1_N +M)=
1+ 
\sum_{k=1}^N \frac{1}{k!} \sum_{\nu_1, \ldots , \nu_k}
\det \left (
\begin{array}{cccc}
M_{\nu_1 \nu_1}
& M_{\nu_1 \nu_2} & 
\ldots & M_{\nu_1 \nu_k}
\\ 
M_{\nu_2 \nu_1}
& M_{\nu_2 \nu_2}&
\ldots & M_{\nu_2 \nu_k}
\\ 
\vdots & \vdots & \ddots & \vdots 
\\ 
M_{\nu_k \nu_1}
& M_{\nu_k \nu_2} &
\ldots & M_{\nu_k \nu_k}
\end{array}
\right ).
\eeq
\end{lem}

\noindent
Therefore, comparing this with (\ref{ex5}), we conclude that
\beq\label{ex8}
\tau'(\bolds , \boldt , \boldtbar )=\det_{N\times N}
\left (1_N + \frac{q}{q-p}\, A(\bolds , \boldt , \boldtbar )\right ).
\eeq
At $N=1$ this formula gives the one-soliton tau-function of the 
one-component Toda lattice. 

In the case $\alpha \neq \beta$ the calculation is slightly more 
involved. We have:
\beq\label{ex9}
\tau '_{\alpha \beta}(\bolds +[1]_{\beta}, \boldt , \bar \boldt )=
\epsilon_{\alpha}(\bolds )\epsilon_{\beta}(\bolds )
\sum_{k=1}^{N-1}\frac{1}{k!}\, D_k^{(\alpha \beta )} (\bolds , \boldt ,
\bar \boldt ),
\eeq
where
\beq\label{ex9a}
\begin{array}{l}
\displaystyle{
D_k^{(\alpha \beta )}(\bolds , \boldt ,
\bar \boldt )=\sum_{{\mu_1, \ldots , \mu_k}\atop
{\nu_1, \ldots , \nu_k}}A_{\mu_1 \nu_1}e^{\eta_{\mu_1 \nu_1}(\boldt ,
\bar \boldt , p,q)}
\ldots A_{\mu_k \nu_k}e^{\eta_{\mu_k \nu_k}(\boldt ,
\bar \boldt , p,q)}}
\\ \\
\displaystyle{
\times \, \bigl <\bolds \bigr |\psi_{s_{\alpha}}^{*(\alpha )}
\psi^{*(\mu_1)}(q)\ldots \psi^{*(\mu_k)}(q)
\psi^{(\nu_k)}(p)\ldots \psi^{(\nu_1)}(p)
\psi_{s_{\beta}}^{(\beta )}\bigl | \bolds \bigr >}
\end{array}
\eeq
(here it is assumed that $s_{\alpha}, s_{\beta}\geq 0$). Applying the 
Wick's theorem, we can represent the multi-point correlation function 
as determinant of the two-point ones:
\beq\label{ex9b}
D_k^{(\alpha \beta )}(\bolds , \boldt ,
\bar \boldt )=\!\! \sum_{{\mu_1, \ldots , \mu_k}\atop
{\nu_1, \ldots , \nu_k}}A_{\mu_1 \nu_1}e^{\eta_{\mu_1 \nu_1}(\boldt ,
\bar \boldt , p,q)}
\ldots A_{\mu_k \nu_k}e^{\eta_{\mu_k \nu_k}(\boldt ,
\bar \boldt , p,q)}
p^{s_{\alpha}}q^{-s_{\beta}}
\det C(\bolds ,\{\mu_i\}, \{\nu_i\}).
\eeq
Here $C(\bolds, \{\mu_i\}, \{\nu_i\})$ is the $(k+1)\times (k+1)$ matrix
$$
\left (
\begin{array}{ccccc}
0 & \delta_{\alpha \nu_1} & 
\delta_{\alpha \nu_2}& \ldots &\!
\delta_{\alpha \nu_k}
\\ &&&& \\
\delta_{\beta \mu_1}& 
\bigl < \bolds \bigr |\psi^{*(\mu_1)}(q)\psi^{(\nu_1)}(p)\bigl |\bolds 
\bigr > & 
\bigl < \bolds \bigr |\psi^{*(\mu_1)}(q)\psi^{(\nu_2)}(p)\bigl |\bolds 
\bigr > & \ldots &\!\!
\bigl < \bolds \bigr |\psi^{*(\mu_1)}(q)\psi^{(\nu_k)}(p)\bigl |\bolds 
\bigr > 
\\ &&&& \\
\delta_{\beta \mu_2}& 
\bigl < \bolds \bigr |\psi^{*(\mu_2)}(q)\psi^{(\nu_1)}(p)\bigl |\bolds 
\bigr > & 
\bigl < \bolds \bigr |\psi^{*(\mu_2)}(q)\psi^{(\nu_2)}(p)\bigl |\bolds 
\bigr > & \ldots &\!\!
\bigl < \bolds \bigr |\psi^{*(\mu_2)}(q)\psi^{(\nu_k)}(p)\bigl |\bolds 
\bigr > 
\\ &&&& \\
\vdots & \vdots & \vdots & \ddots &\!\! \vdots
\\ &&&& \\
\delta_{\beta \mu_k}& 
\bigl < \bolds \bigr |\psi^{*(\mu_k)}(q)\psi^{(\nu_1)}(p)\bigl |\bolds 
\bigr > & 
\bigl < \bolds \bigr |\psi^{*(\mu_k)}(q)\psi^{(\nu_2)}(p)\bigl |\bolds 
\bigr > & \ldots &\!\!
\bigl < \bolds \bigr |\psi^{*(\mu_k)}(q)\psi^{(\nu_k)}(p)\bigl |\bolds 
\bigr > \end{array}
\right ).
$$
Expanding the determinant in the first row and the first column,
we get:
$$
D_k^{(\alpha \beta )}(\bolds, \boldt,\bar \boldt )\! =\!
\Bigl ( \frac{q}{q\! -\! p}\Bigr )^{k-1} \!
\sum_{a,b=1}^k \! (-1)^{a \!+\! b\! +\! 1}\!
A_{\mu_1 \nu_1}(\bolds, \! \boldt,\! \bar \boldt )\ldots
A_{\mu_k \nu_k}(\bolds, \! \boldt,\! \bar \boldt )
\delta_{\mu_a \beta}\delta_{\nu_b \alpha}\!
\det_{{1\leq i,j\leq k}\atop {i\neq a, j\neq b}}\!
\Bigl (\delta_{\mu_i \nu_j}\Bigr )
$$
After some transformations we obtain the following result:
\beq\label{ex10}
D_k^{(\alpha \beta )}=-k \Bigl (\frac{q}{q-p}\Bigr )^{k-1}
\!\! \sum_{\nu_1 , \ldots , \nu_{k-1}}
\det \left (
\begin{array}{ccccc}
A_{\beta \alpha}
& A_{\nu_1 \alpha}
& A_{\nu_2 \alpha} & 
\ldots & \! A_{\nu_{k-1} \alpha}
\\ 
A_{\beta \nu_1}
& A_{\nu_1 \nu_1}
& A_{\nu_2 \nu_1} & 
\ldots & \! A_{\nu_{k-1} \nu_1}
\\ 
A_{\beta \nu_2}
& A_{\nu_1 \nu_2}  
& A_{\nu_2 \nu_2} & 
\ldots & \! A_{\nu_{k-1} \nu_2}
\\ 
\vdots & \vdots & \vdots & \ddots & \! \vdots 
\\ 
A_{\beta \nu_{k-1}}
& A_{\nu_1 \nu_{k-1}} &
A_{\nu_2 \nu_{k-1}} &
\ldots & \! A_{\nu_{k-1}\nu_{k-1}}
\end{array}
\right )
\eeq
(to save the space, here we do not indicate the dependence on
$\bolds , \boldt , \bar \boldt$).
To see what it is, we need the following lemma.

\begin{lem}
\label{lemma:minor}
For an $N\times N$ matrix $M$, the $\alpha \beta$-minor of the matrix 
$1_N +M$ with $\alpha \neq \beta$,
$$
(1_N +M)_{\hat \alpha \hat \beta} =(-1)^{\alpha +\beta}
\det_{
1\leq \mu , \nu \leq N\atop \mu \neq \alpha , \nu \neq \beta} 
(1_N +M)_{\mu \nu},
$$
is expressed in terms of minors of the matrix $M$ in the following
way:
\beq\label{ex11}
\begin{array}{c}
\displaystyle{
(1_N \! +\! M)_{\hat \alpha \hat \beta}=
-\!\! \sum_{k=1}^{N-1}\frac{1}{(k\! -\! 1)!}
\!\! \sum_{\nu_1 , \ldots , \nu_{k-1}}
\!\!\! \det \left (
\begin{array}{ccccc}
M_{\beta \alpha}
& M_{\nu_1 \alpha}
& M_{\nu_2 \alpha} & 
\ldots & \! M_{\nu_{k-1} \alpha}
\\ 
M_{\beta \nu_1}
& M_{\nu_1 \nu_1}
& M_{\nu_2 \nu_1} & 
\ldots & \! M_{\nu_{k-1} \nu_1}
\\ 
M_{\beta \nu_2}
& M_{\nu_1 \nu_2}  
& M_{\nu_2 \nu_2} & 
\ldots & \! M_{\nu_{k-1} \nu_2}
\\ 
\vdots & \vdots & \vdots & \ddots & \! \vdots 
\\ 
M_{\beta \nu_{k-1}}
& M_{\nu_1 \nu_{k-1}} &
M_{\nu_2 \nu_{k-1}} &
\ldots & \! M_{\nu_{k-1}\nu_{k-1}}
\end{array}
\right )}
\\ \\
\displaystyle{
=-M_{\beta \alpha}-\sum_{\nu}
\det \left (
\begin{array}{cc}
M_{\beta \alpha}& M_{\beta \nu}
\\ 
M_{\nu \alpha}& M_{\nu \nu}
\end{array} \right )-
\frac{1}{2!}
\sum_{\nu_1 , \nu_2}
\det \left (
\begin{array}{ccc}
M_{\beta \alpha}& M_{\beta \nu_1}& M_{\beta \nu_2}
\\ 
M_{\nu_1 \alpha}& M_{\nu_1 \nu_1}& M_{\nu_1 \nu_2}
\\ 
M_{\nu_2 \alpha}& M_{\nu_2 \nu_1}& M_{\nu_2 \nu_2}
\end{array}
\right ) -\ldots \, .}
\end{array}
\eeq
\end{lem}

\begin{proof}
We assume that $\alpha \neq \beta$. 
By permutation of rows and columns of the matrix $1_N+M$ in which the 
$\alpha$-th row and the $\beta$-th column are removed we can represent it
as a block matrix of the form
$$
\left (
\begin{array}{cc}
M_{\beta \alpha} & m^{(\beta )} 
\\ 
\hat m^{(\alpha )} & 1_{N-2}\! +\! \tilde M_{N-2}
\end{array}
\right )=:K^{(\alpha \beta )},
$$
where $m^{(\beta )}$ is the $(N-2)$-dimensional row vector with
components $(m^{(\beta )})_{\mu}=M_{\beta \mu}$, $\mu \neq \alpha , \beta$,
$\hat m^{(\alpha )}$ is the $(N-2)$-dimensional 
column vector with components
$(\hat m^{(\alpha )})_{\mu}=M_{\mu \alpha}$, $\mu \neq \alpha , \beta$ and
$\tilde M_{N-2}$ is the $(N-2)\times (N-2)$ square 
matrix which is obtained from
the matrix $M$ by removing $\alpha$-th and $\beta$-th columns and rows.
By counting the number of necessary permutations 
of rows and columns it is easy to see that
$$
(1_N +M)_{\hat \alpha \hat \beta}=-\det K^{(\alpha \beta )}.
$$
Writing $M_{\beta \alpha}=1+ (M_{\beta \alpha}-1)=:
1+\tilde M_{\beta \alpha}$ with $\tilde M_{\beta \alpha}
=M_{\beta \alpha}-1$, we bring
the matrix $K^{(\alpha \beta )}$ to the form $1_{N-1}+\tilde M_{N-1}$,
so its determinant can be represented as the sum of 
diagonal minors of the matrix 
$$\tilde M_{N-1}=\left (
\begin{array}{cc}
M_{\beta \alpha}\! -\! 1 & m^{(\beta )} 
\\ 
\hat m^{(\alpha )} & \tilde M_{N-2}
\end{array}
\right )
$$ 
as is stated in Lemma \ref{lemma:det(1+M)}. 
An easy calculation shows that in this way 
one obtains equation (\ref{ex11}).
\end{proof}

Therefore, the sum in the right-hand side of equation (\ref{ex9})
is nothing else than the $\alpha \beta$-minor of the matrix
$\displaystyle{1_N +\frac{q}{q-p}\, A(\bolds , 
\boldt , \bar \boldt )}$ (up to
a common multiplier). More precisely, we have:
\beq\label{ex12}
\begin{array}{lll}
\tau '_{\alpha \beta}(\bolds +[1]_{\beta}), \boldt , \bar \boldt )&=&
\displaystyle{
\epsilon_{\alpha}(\bolds )\epsilon_{\beta}(\bolds )
\frac{q\! -\! p}{q}
\Bigl (1_N +\frac{q}{q\! -\! p}\, A(\bolds , 
\boldt , \bar \boldt )\Bigr )_{\hat \alpha \hat \beta}}
\\ && \\
&=&\displaystyle{(-1)^{\alpha +\beta}
\epsilon_{\alpha}(\bolds )\epsilon_{\beta}(\bolds )
\frac{q\! -\! p}{q}
\det_{1\leq \mu , \nu \leq N\atop \mu \neq \alpha , \nu \neq \beta} 
\Bigl (\delta_{\mu \nu}+\frac{q}{q\! -\! p}\, A_{\mu \nu}(\bolds , 
\boldt , \bar \boldt )\Bigr ), 
\quad \alpha \neq \beta .}
\end{array}
\eeq
Together with equation (\ref{ex8}) this gives a multi-component analogue
of the one-soliton solution to the Toda lattice.

\section{Concluding remarks}

In this paper we have introduced an extension of the $N$-component
Toda lattice hierarchy. This hierarchy contains $N$ discrete variables
$\bolds =\{s_1, \ldots , s_N\}$
rather than one, as it goes in the version suggested by Ueno and Takasaki 
in 1984 \cite{UT84}. Simultaneously, we have refined some arguments from
\cite{UT84}. 

We have obtained the multi-component Toda lattice hierarchy in three
different ways, deducing it from different starting points.
 
One of them is the Lax formalism whose main ingredients are 
two Lax operators $\boldL$, $\boldLbar$ and auxiliary operators
$\boldU_{\alpha}, \boldUbar_{\alpha}$, 
$\boldP_{\alpha}, \boldPbar_{\alpha}$, $\alpha =1, \ldots , N$,
which are realized as difference operators with $N\times N$ 
matrix coefficients. These operators are subject to certain
algebraic relations. Their evolution in the time variables 
is given by the Lax equations (or 
discrete Lax equations for evolution 
in the discrete variables $s_1, \ldots , s_N$). We have presented
a detailed proof that the Lax representation is equivalent to the
system of Zakharov-Shabat (or zero curvature) equations. Next, we
have proved the existence of the so-called wave operators 
$\boldW$, $\boldWbar$ from which all other operators of the Lax
formalism are obtained by ``dressing''. With the help of the 
wave operators, one can introduce matrix wave functions which
obey an infinite system of linear equations. Compatibility conditions
for this system is just the Zahkarov-Shabat equations. The wave 
functions, together with their adjoint functions, are shown to
obey a fundamental integral bilinear identity. In its turn, this 
identity implies the existence of the matrix tau-function which
is the most fundamental dependent variable of the hierarchy.
The tau-function is shown to satisfy the integral bilinear
equation which is a sort of generating equation for equations 
of the hierarchy.

The alternative starting point is the multi-component KP hierarchy,
which is essentially equivalent to the so-called universal hierarchy
\cite{KZ23}. We have shown that the $N$-component Toda lattice
hierarchy can be embedded into the $2N$-component universal
hierarchy. Namely, we have shown that under certain conditions
the integral bilinear equation for the latter becomes the 
integral bilinear equation for the former (more precisely, 
to identify them, some simple
redefinition of the tau-functions consisting in multiplying them
by some sign factors is necessary). 

Last but not least, there is an approach based on the quantum
field theory of free fermions, which was developed in early 1980's
by Kyoto school. 
For multi-component hierarchies one should deal with multi-component
fermions. In this formalism, the tau-functions are defined as 
expectation values of certain operators constructed from 
free fermions (Clifford group elements); 
and the integral bilinear equation for the tau-function 
is a corollary of the bilinear identity for fermionic operators
which is a characteristic property of Clifford group elements.
In order to deduce this corollary, one needs certain relations
between fermi- and bose-operators which are often referred to
as bosonization rules. In this paper, we have implemented 
this program. As a result, we have obtained the integral bilinear
equation for the tau-function of the multi-component Toda
lattice hierarchy which turns out to be the same as the ones 
obtained in the framework of the other two approaches.

We hope that the present paper provides a complete treatment 
of the subject. Let us list possible directions 
for further work.

One of them is the 
problem of defining multi-component analogues of the Toda
lattices with constraints of types B and C
introduced and studied in \cite{KZ23a,PZ23} and \cite{KZ22}
respectively.
The former hierarchy is 
equivalent to the so-called large BKP hierarchy
(see, e.g., \cite{GWRC24} and \S 7.4 of the book \cite{book}).

Another direction for further work is studying the dispersionless
limit of the multi-component Toda lattice hierarchy (see \cite{TT95}
for the case $N=1$). 
In the dispersionless limit one should re-scale the independent
variables as
$t_{\alpha,k} \to t_{\alpha,k}/\hbar$, 
$\bar t_{\alpha,k} \to \bar t_{\alpha,k}/\hbar$,
$s_{\alpha}\to s_{\alpha}/\hbar$ and consider solutions (tau-functions)
that have an essential singularity at $\hbar =0$ and have the form
$$
\tau (\bolds /\hbar ,{\boldsymbol t} /\hbar , 
\bar {\boldsymbol t} /\hbar)=
e^{\frac{1}{\hbar^2}F(\bolds , {\boldsymbol t} , 
\bar {\boldsymbol t}, \hbar )}
$$
as $\hbar \to 0$, where $F$ is a smooth function of $\bolds $ and ${\boldsymbol t}, \bar {\boldsymbol t}$
having a regular expansion in $\hbar$ as $\hbar \to 0$. The function
$F=F(\bolds , {\boldsymbol t} , \bar 
{\boldsymbol t}, 0)$ is sometimes called the 
dispersionless tau-function
(although the $\hbar \to 0$ limit of the tau-function 
itself does not exist). The bilinear equations for the tau-function
$\tau (\bolds , {\boldsymbol t}, \bar{\boldsymbol t} )$ lead to non-linear equations 
for the $F$-function. It would be interesting to write these equations
explicitly and compare with the ones obtained in \cite{TT07,Z24}.

It seems to be an intriguing problem to study the $N$-component 
Toda lattice in the limit $N\to \infty$ and to see whether any 
new phenomena arise in this limit. 

Finally, it would be desirable to obtain multi-component 
analogues of multi-soliton solutions to the Toda lattice 
in an explicit form. The natural framework for this
is the free fermion technique. The example of one-soliton solution 
is given in this paper in Section \ref{subsection:example}.

\section*{Appendix: Non-abelian Toda lattice}

\addcontentsline{toc}{section}{Appendix: 
Non-abelian Toda lattice}
\def\theequation{A\arabic{equation}}
\def\theHequation{\theequation}
\setcounter{equation}{0}

In the appendix we consider the non-abelian Toda lattice which
provides an explicit example of the general construction of
Section 2. Namely, we derive the non-abelian Toda lattice equation
from the Zakharov-Shabat representation of the matrix Toda lattice
hierarchy which can be regarded as a subhierarchy of the 
multi-component one discussed
in Section 2. In the matrix Toda hierarchy, the independent variables
$s, t_k, \bar t_k$ are introduced by assigning the variables
$s_{\alpha}, t_{\alpha ,k}, \bar t_{\alpha ,k}$ the following
values:
$$
s_{\alpha}=s_{\alpha}^{(0)}+s, \quad 
t_{\alpha ,k}=t_{\alpha ,k}^{(0)}+t_k, \quad
\bar t_{\alpha ,k}=\bar t_{\alpha ,k}^{(0)}+\bar t_k,  
$$
where $s_{\alpha}^{(0)}$, $t_{\alpha ,k}^{(0)}$, 
$\bar t_{\alpha ,k}^{(0)}$ are some fixed parameters,
so the corresponding vector fields are
$$
\p_{t_k}=\sum_{\alpha =1}^N \p_{t_{\alpha ,k}}, \quad
\p_{\bar t_k}=\sum_{\alpha =1}^N \p_{\bar t_{\alpha ,k}}.
$$
Accordingly, the operators $\boldU_{\alpha}, \boldUbar_{\alpha},
\boldP_{\alpha}, \boldPbar_{\alpha}$ do not take part in the construction
and we are left with the two Lax operators with matrix coefficients
\beq\label{A1}
\begin{array}{l}
\displaystyle{
\boldL (s)=\sum_{j=0}^{\infty}b_j(s)e^{(1-j)\p_s}, \quad b_0(s)=1_N,}
\\ \\
\displaystyle{
\boldLbar (s)=\sum_{j=0}^{\infty}\bar 
b_j(s)e^{(j-1)\p_s}, \quad \bar b_0(s)=g(s)g^{-1}(s-1),}
\end{array}
\eeq
where $g(s)$ is an $N\times N$ invertible matrix (which was denoted
by $\tilde w_0$ in Section \ref{sec:def-multi-toda}). 
The dependent variables (in particular,
$g(s)$) are regarded as functions of $s, t_k, \bar t_k$.

The Lax operators
satisfy the Lax equations
\beq\label{A2}
\begin{array}{l}
\p_{t_k}\boldL (s)=[\boldB_k(s), \boldL (s)], \quad
\p_{t_k}\boldLbar (s)=[\boldB_k(s), \boldLbar (s)],
\\ \\
\p_{\bar t_k}\boldL (s)=[\boldBbar_k(s), \boldL (s)], \quad
\p_{\bar t_k}\boldLbar (s)=[\boldBbar_k(s), \boldLbar (s)],
\end{array}
\eeq
where
$$
\boldB_k(s)=(\boldL^k (s))_{\geq 0}, \quad
\boldBbar_k(s)=(\boldLbar^k (s))_{<0}.
$$
Compatibility conditions for the Lax equations are expressed as
Zakharov-Shabat equations.
We will derive the first nontrivial equation of the hierarchy
from the Zakharov-Shabat equation
\beq\label{A3}
[\p_{t_m}-\boldB_m (s), \, \p_{\bar t_n}-\boldBbar_n (s)]=0
\eeq
at $m=n=1$. In what follows we put $t_1=t$, $\bar t_1=\bar t$.

We have:
\beq\label{A4}
\boldB_1=1_Ne^{\p_s} +b_1(s), \quad \boldBbar_1=\bar b_0(s)
e^{-\p_s}.
\eeq
Plugging this into the Zakharov-Shabat equation, we obtain the system
of equations
\beq\label{A5}
\left \{ 
\begin{array}{l}
\p_{\bar t}b_1(s)=\bar b_0(s)-\bar b_0(s+1),
\\ \\
\p_{t}\bar b_0(s)=b_1(s)\bar b_0(s)-\bar b_0(s)b_1(s-1).
\end{array}
\right.
\eeq
Substituting $\bar b_0(s)=g(s)g^{-1}(s-1)$, we represent the second
equation in the form
\beq\label{A6}
g^{-1}(s)h(s)g(s)=g^{-1}(s-1)h(s-1)g(s-1),
\eeq
where 
\beq\label{A7}
h(s)=\p_t g(s)\, g^{-1}(s)-b_1(s).
\eeq
Therefore, $g^{-1}(s)h(s)g(s)=h_0$ does not depend on $s$. We assume
that it does not depend also on the times, so $h_0$ is a constant
matrix. From (\ref{A6}) we have:
\beq\label{A8}
h(s)=g(s)h_0g^{-1}(s).
\eeq
Plugging this into the first equation in (\ref{A5}),
we obtain:
\beq\label{A9}
\p_{\bar t}(\p_t g(s)\, g^{-1}(s))=g(s)g^{-1}(s-1)-
g(s+1)g^{-1}(s)+ \p_{\bar t}(g(s)h_0g^{-1}(s)).
\eeq
The simple redefinition $g(s)\to g(s)e^{th_0}$ kills the last term
in the right-hand side, so we can put $h(s)=0$ without loss of 
generality. In this way we obtain the equation of the non-abelian
Toda lattice:
\beq\label{A10}
\p_{\bar t}(\p_t g(s)\, g^{-1}(s))=g(s)g^{-1}(s-1)-
g(s+1)g^{-1}(s).
\eeq

Actually, using the linear equations \eqref{lin-eq:dWbarhat} for
$\hat\boldWbar$ and the Zakharov-Shabat equation \eqref{A3} for $n=m=1$,
we can easily show that $g(s)=\bar w_0(s)$ satisfies this equation. In
fact, the above redefinition $g(s)\to g(s)e^{th_0}$ corresponds to
$
    \bar w_0(\bolds)
    =
    \tilde w_0(\bolds)\,\tilde{\bar c}(\boldt,\boldtbar)
$
in \remref{rem:w0-tilde-neq-w0-bar}.

The tau-function provides bilinearization of the non-abelian 
Toda lattice equation (\ref{A10}). In terms of the tau-function
from Section \ref{section:tau-function} we have:
\beq\label{A11}
\begin{array}{l}
\displaystyle{
g_{\alpha \beta}(s) = (-1)^{\delta_{\alpha \beta}-1}
\frac{\tau_{\alpha \beta}(s\boldone +[1]_{\beta})}{\tau (s\boldone )},}
\\ \\
\displaystyle{
(g^{-1}(s))_{\alpha \beta} = 
\frac{\tau_{\alpha \beta}((s+1)\boldone -
[1]_{\alpha})}{\tau ((s+1)\boldone )}}
\end{array}
\eeq
and (\ref{A10}) follows from (\ref{n10a}) and (\ref{n11}). 
Note that the tau-functions here can be substituted by the
modified tau-functions $\tau'_{\alpha \beta}$ defined in (\ref{ex4}).

Finally, let us give an explicit 
example of exact solutions to (\ref{A10}) based on the one
considered in Section
\ref{subsection:example}.
Taking into account the difference between the tau-functions 
introduced in Sections \ref{section:tau-function} and
\ref{section:fermions} (see (\ref{differ})), we write:
\beq\label{A12}
\tau '(s\boldone )=\det_{1\leq \mu , \nu \leq N}
\Bigl (\delta_{\mu \nu}+\frac{q}{q-p}A_{\mu \nu}(s, t, \bar t)\Bigr ),
\eeq
\beq\label{A13}
\tau '_{\alpha \beta}(s\boldone +[1]_{\beta})=
\left \{
\begin{array}{l}
\displaystyle{
\det_{1\leq \mu , \nu \leq N}
\Bigl (\delta_{\mu \nu}+\frac{q^{1-\delta_{\beta \mu}}
p^{\delta_{\beta \nu}}}{q-p}A_{\mu \nu}(s, t, \bar t)\Bigr ), \quad
\alpha =\beta ,}
\\ \\
\displaystyle{
(-1)^{\alpha +\beta +1}\frac{q-p}{q}
\det_{{1\leq \mu , \nu \leq N}\atop {\mu \neq \alpha ,
\nu \neq \beta}}
\Bigl (\delta_{\mu \nu}+\frac{q}{q-p}A_{\mu \nu}(s, t, \bar t)\Bigr ),}
\end{array}\right.
\eeq
where
$$
A_{\mu \nu}(s, t, \bar t)=(p/q)^s e^{(p-q)t +(p^{-1}-q^{-1})\bar t}
A_{\mu \nu}.
$$
For example, at $N=2$ we have the solution
\beq\label{A14}
\begin{array}{l}
\displaystyle{
g(s)=\frac{1}{1+q \mbox{tr}B+q^2 \det B}
\left (
\begin{array}{cc}
1\! +\! pB_{11}\! +\! qB_{22}\! +\! pq \det B & (p-q)B_{12}
\\ & \\
(p-q)B_{21} & 1\! +\! qB_{11}\! +\! pB_{22}\! +\! pq \det B
\end{array}
\right )}
\\ \\
\displaystyle{
\phantom{aaaa}=\frac{(1_2 +pB)(1_2+q\tilde B)}{\det (1_2 + qB)}},
\end{array}
\eeq
\beq\label{A15}
\begin{array}{l}
\displaystyle{
g^{-1}(s)=\frac{1}{1+p \mbox{tr}B+p^2 \det B}
\left (
\begin{array}{cc}
1\! +\! qB_{11}\! +\! pB_{22}\! +\! pq \det B & (p-q)B_{12}
\\ & \\
(p-q)B_{21} & 1\! +\! pB_{11}\! +\! qB_{22}\! +\! pq \det B
\end{array}
\right )}
\\ \\
\phantom{aaaa}\displaystyle{
=\frac{(1_2 +qB)(1_2+p\tilde B)}{\det (1_2 + pB)}},
\end{array}
\eeq
where
$$
B=B(s,t, \bar t)=\frac{1}{q-p}\, \Bigl (\frac{p}{q}\Bigr )^s 
e^{(p-q)t+(p^{-1}-q^{-1})\bar t}\, A^T
$$
and $\tilde B$ is the matrix of minors of $B$ ($\tilde B_{11}=
B_{22}$, $\tilde B_{22}=
B_{11}$, $\tilde B_{12}=
-B_{21}$, $\tilde B_{21}=
-B_{12}$).

\section*{Acknowledgments}
\addcontentsline{toc}{section}{Acknowledgments}

We thank S. Kakei and V. Prokofev for discussions.
The work of A.Z. 
has been supported in part within the framework of the
HSE University Basic Research Program (section 5) and
within the state assignment of NRC 
``Kurchatov institute'' (section 6).

\medskip
We dedicate this paper to the memory of Masatoshi Noumi, who passed away
on 20 November 2024. His interest, comments and encouragements to our
works were invaluable for us.

\end{document}